\documentclass[11pt,3p,times,preprint]{elsarticle}

\newcommand{\hide}[1]{}
\usepackage{tikz}
\usepackage{mfirstuc}
\usepackage{paralist}
\usepackage{amsthm,amsfonts,amsmath,amssymb}
\usepackage{multirow}
\usepackage{booktabs}
\usepackage{afterpage}
\usepackage{url}

\usepackage{xspace}
\usepackage{subfig}

\usepackage{hyperref}

\usepackage{microtype}
\usepackage{times}
\usepackage{geometry}
\makeatletter

\def\ps@pprintTitle{%
     \let\@oddhead\@empty
     \let\@evenhead\@empty
     \let\@oddfoot\@empty
     \let\@evenfoot\@oddfoot}

\renewcommand\section{\@startsection {section}{1}{\z@}%
           {-38\p@ \@plus 6\p@ \@minus 3\p@}%
           {9\p@ \@plus 6\p@ \@minus 3\p@}%
           {\normalsize\bfseries\boldmath}}
\renewcommand\subsection{\@startsection{subsection}{2}{\z@}%
           {-28\p@ \@plus 6\p@ \@minus 3\p@}%
           {13\p@ \@plus 6\p@ \@minus 3\p@}%
           {\normalfont\normalsize\itshape}}

\makeatother

\newcommand{\citex}[1]{\citeauthor{#1}~\shortcite{#1}}
\newcommand{\citey}[1]{\citeauthor{#1}~\shortcite{#1}}
\newcommand{\shortcite}[1]{\cite{#1}}
\newcommand*\citexs[1]{\citexds#1\relax}
\def\citexds#1,#2\relax{\citeauthor{#1}~\cite{#1,#2}}

\newtheorem{LEM}{Lemma}[section] 
\newtheorem{THE}{Theorem}[section] 
\newtheorem{COR}{Corollary}[section]

\newtheorem{PRO}{Proposition}[section] 
\newtheorem{DEF}{Definition}[section]
\newtheorem{OBS}{Observation}[section]
\theoremstyle{remark}
\newtheorem*{REM}{Remark}

\newtheorem{EX}{Example}[section] 
\renewenvironment{EX}{\begin{ex}\normalfont}{\hfill
    $\dashv$\end{ex}\medskip}

\def\hy{\hbox{-}\nobreak\hskip0pt} 
\newcommand{\SB}{\{\,}%
\newcommand{\SM}{\;{:}\;}%
\newcommand{\SE}{\,\}}%

\newcommand{\Card}[1]{|#1|}
\newcommand{\CCard}[1]{\|#1\|}
\let\phi=\varphi
\let\epsilon=\varepsilon

\newcommand{\Nat}{\mathbb{N}}

\newcommand{\SSS}{\mathsf{S}}

\newcommand{\CCC}{\ensuremath{\mathcal{C}}\xspace}

\newcommand{\var}{\mathit{var}}

\newcommand{\SharpP}{\#P}
\renewcommand{\P}{\text{\normalfont P}\xspace}
\newcommand{\NP}{\text{\normalfont NP}\xspace}
\newcommand{\coNP}{\text{\normalfont co-NP}\xspace}
\newcommand{\FPT}{\text{\normalfont FPT}\xspace}
\newcommand{\XP}{\text{\normalfont XP}\xspace}
\newcommand{\W}[1][xxxx]{\text{\normalfont W[#1]}\xspace}
\newcommand{\paraNP}{\text{\normalfont paraNP}\xspace}
\newcommand{\coparaNP}{\text{\normalfont co-paraNP}\xspace}
\newcommand{\paracoNP}{\text{\normalfont para-coNP}\xspace}

\newcommand{\BigO}[1]{\ensuremath{\mathcal{O}(#1)}}

\newcommand{\stableset}{\text{\normalfont AS}}
\newcommand{\tautext}[1]{\ensuremath{#1^*}}

\newcommand{\at}{\text{\normalfont at}}
\newcommand{\ta}[1]{2^{#1}}
\newcommand{\class}[1]{\textnormal{\textbf{#1}}}
\newcommand{\parm}[1]{\textnormal{#1}}
\newcommand{\NSTR}{\class{Strat}}
\newcommand{\Horn}{\ensuremath{\class{Horn}}\xspace}
\newcommand{\Normal}{\ensuremath{\class{Normal}}}

\newcommand{\tw}{\text{\normalfont tw}}

\newcommand{\restrictBE}[2]{\ensuremath{{#1}_{\slash #2}}}
\newcommand{\restrictGSS}[2]{\ensuremath{\hat {#1}_{\slash #2}}}
\newcommand{\SCC}{\text{\normalfont SCC}}
\newcommand{\Good}[1]{\ensuremath{\text{\normalfont at}^+(#1)}}

\newcommand{\pnot}{\neg}
\newcommand{\por}{\vee}
\newcommand{\rsep}{;\;}

\newcommand{\NAT}{\mathbb{N}}
\newcommand{\str}{\ensuremath{\textit{str}}}

\newcommand{\ASP}{\text{ASP}\xspace}
\newcommand{\SAT}{\text{SAT}\xspace}
\newcommand{\CSP}{\text{CSP}\xspace}
\newcommand{\ILP}{\text{ILP}\xspace}

\newcommand{\pname}[1]{\ifmmode ${\scshape #1}$ \else{\scshape #1}\fi\xspace}
\newcommand{\pnormal}[1]{\ensuremath{#1_{\text{\normalfont N}}}}
\newcommand{\HS}{\pname{Hitting Set}}

\newcommand{\lift}[1]{#1^{\uparrow}}
\newcommand{\down}[1]{#1^{\downarrow}}

\usepackage{subfig}
\usepackage{tikz}
\usetikzlibrary{calc,arrows,shapes,petri,trees}

\newcommand{\DBC}[0]{\class{no-DBC}\xspace}
\newcommand{\BC}{\class{no-BC}\xspace}
\newcommand{\DC}{\class{no-DC}\xspace}
\newcommand{\DCTWO}{\class{no-DC2}\xspace}
\newcommand{\C}{\class{no-C}\xspace}
\newcommand{\DBEC}{\class{no-DBEC}\xspace}
\newcommand{\BEC}{\class{no-BEC}\xspace}
\newcommand{\DEC}{\class{no-DEC}\xspace}
\newcommand{\EC}{\class{no-EC}\xspace}

\DeclareFontFamily{OT1}{pzc}{}
\DeclareFontShape{OT1}{pzc}{m}{it}{<-> s * [1.10] pzcmi7t}{}
\DeclareMathAlphabet{\mathpzc}{OT1}{pzc}{m}{it}

\newcommand{\Acyc}{\mathpzc{Acyc}}

\newcommand{\DAcyc}{\mathpzc{D}\text{-}\mathpzc{Acyc}}

\newcommand{\problem}[3]
{
\begin{quote}
\begin{samepage}
{\scshape {#1}} \normalfont \vspace{0.4em}\\
\begin{tabular}{lp{0.53\textwidth}}
  \emph{Given:} & \makefirstuc{#2.} \tabularnewline[1pt]
  \emph{Task:} & \makefirstuc{#3.} \tabularnewline[1pt]
\end{tabular}
\end{samepage}
\end{quote}
}

\newcommand{\pproblem}[4]{
\begin{quote}
\begin{samepage}
{\scshape {#1}\nopagebreak[4]} \normalfont \vspace{0.4em}\nopagebreak[4]\\
\begin{tabular}{lp{0.7\textwidth}}
  \emph{Given:} & \makefirstuc{#2} \tabularnewline[1pt]
 \emph{Parameter:} & \makefirstuc{#3} \tabularnewline[1pt]
  \emph{Task:} & \makefirstuc{#4} \tabularnewline[1pt]
\end{tabular}
\end{samepage}
\end{quote}
}

\newcommand{\AspCheck}{\pname{Checking}}
\newcommand{\AspCons}{\pname{Consistency}}
\newcommand{\AspBrave}{\pname{Brave Reasoning}}
\newcommand{\AspCaut}{\pname{Skeptical Reasoning}}
\newcommand{\AspCount}{\pname{Counting}}
\newcommand{\AspEnum}{\pname{Enum}}
\newcommand{\AspReason}{\ensuremath{\mathpzc{AspReason}}\xspace}
\newcommand{\AspFull}{\ensuremath{\mathpzc{AspFull}}\xspace}
\newcommand{\Bound}{\pname{Bound}}

\usepackage{color}
\newcommand{\remark}[1]{{\color{red}[#1]}}

\newcounter{ct_todo}
\newcommand{\comment}[1]{\remark{#1}}

\newcommand{\todo}[1]{{\comment{{\bf TODO: }\color{red}#1}}\xspace}

\newcommand{\delBds}[1]{deletion {\ensuremath{#1}}\hy backdoor}
\newcommand{\strongBds}[1]{strong {\ensuremath{#1}}\hy backdoor}
\newcommand{\strongdelBds}[1]{strong (deletion) {\ensuremath{#1}}\hy backdoor}

\newcommand{\BdCheck}[1]{\pname{Strong \ensuremath{#1}-Back\-door Asp Check}}

\newcommand{\strongBdsDet}[1]{\pname{Strong \ensuremath{#1}\hy Back\-door Detection}}
\newcommand{\delBdsDet}[1]{\pname{Deletion \ensuremath{#1}\hy Back\-door Detection}}

\newcommand{\VC}{\pname{Vertex Cover}}
\newcommand{\sfvs}[1]{\ifx&#1&\else{\ensuremath{#1}}\hy\fi feedback vertex set}
\newcommand{\ect}[1]{even {\ensuremath{#1}}\hy cycle transversal}

\newcommand{\fvsDet}[1][]{\pname{\ifx&#1&\else\ensuremath{#1}\hy\fi
    Feedback Vertex Set}} 
\newcommand{\FVS}{\pname{Feedback Vertex Set}}

\begin{document}
\begin{frontmatter}
\title{Backdoors to Tractable Answer Set Programming\tnoteref{t1,t2}}
\author[vie,po]{Johannes Klaus Fichte} \ead{fichte@kr.tuwien.ac.at}
\author[vie]{Stefan Szeider} \ead{stefan@szeider.net}
\address[vie]{Vienna University of Technology,\\ Favoritenstrasse
  9-11, 1040 Vienna, Austria} 
\address[po]{University of Potsdam,\\ August-Bebel-Strasse 89,
14482 Potsdam, Germany} 
\date{} 

% \tnotetext[t2]{A preliminary version appeared in the Proceedings of
% the 22nd International Conference on Artificial Intelligence
% (\hbox{IJCAI'11})~\cite{FichteSzeider11} and in the New Directions
% in Logic, Language and Computation~\cite{Fichte12}. This is an
% extended and revised version.}

\tnotetext[t1]{Fichte and Szeider’s research was supported by the
  European Research Council, grant reference 239962 (COMPLEX
  REASON). }

\begin{abstract}
  Answer Set Programming (ASP) is an increasingly popular framework
  for declarative programming that admits the description of problems
  by means of rules and constraints that form a disjunctive logic
  program. In particular, many AI problems such as reasoning in a
  nonmonotonic setting can be directly formulated in ASP. Although the
  main problems of ASP are of high computational complexity, located
  at the second level of the Polynomial Hierarchy, several
  restrictions of ASP have been identified in the literature, under
  which ASP problems become tractable.
  
  In this paper we use the concept of backdoors to identify new
  restrictions that make ASP problems tractable.  Small backdoors are sets
  of atoms that represent ``clever reasoning shortcuts'' through the
  search space and represent a hidden structure in the problem input.
  The  concept of backdoors is widely used in the areas of propositional
  satisfiability and constraint satisfaction.  We show that it can be 
  fruitfully adapted to ASP.
  We demonstrate how backdoors can serve as a unifying framework that
  accommodates several tractable restrictions of ASP known from the
  literature. Furthermore, we show how backdoors allow us to deploy
  recent algorithmic results from parameterized complexity theory to
  the domain of answer set programming. 
\end{abstract}


%\maketitle
\end{frontmatter}

\pagebreak
\tableofcontents
\pagebreak

 \section{Introduction}
 \emph{Answer Set Programming} (\ASP) is an increasingly popular
 framework for declarative
 programming~\cite{MarekTruszczynski99,Niemela99}. \ASP admits the
 description of problem by means of rules and constraints that form a
 disjunctive logic program. Solutions to the program are so-called
 stable models or answer sets. Many important problems of AI and
 reasoning can be succinctly represented and successfully solved
 within the \ASP framework. It has been applied to several large
 industrial applications, e.g., social
 networks~\cite{JostSabuncuSchaub12a}, match
 making~\cite{GebserGlaseSabuncuSchaub13}, planning in a
 seaport~\cite{RiccaGrassoAlvianoMannaLioIiritanoLeone12},
 optimization of packaging of Linux
 distributions~\cite{GebserKaminskiKaufmannSchaub11}, and general game
 playing~\cite{Thielscher09}.
% optimization~\cite{AndresKaufmannMattheisSchaub12,GebserKaminskiKaufmannSchaub11},
%  planning~\cite{GebserKaufmannSchaub12,PontelliSonBaralGelfond12,Schneider13},
%  robotics~\cite{AkerPatogluErdem12},
%  scheduling~\cite{RiccaGrassoAlvianoMannaLioIiritanoLeone12}, semantic
%  web~\cite{EiterFinkKrennwallnerRedlSchuller12},
%  verification~\cite{FalknerSchennerFriedrichRyabokon12}, and several
%  other fields.
% Over the last few years, tremendous gains have been made in the
% efficency of \ASP solvers crucially conducted by integrating several
% \CSP and \SAT techniques, e.g.~\cite{GebserKaufmannNeumannSchaub07},
% even though 

 The main computational problems for \ASP (such as deciding whether a
 program has a solution, or whether a certain atom is contained in at
 least one or in all solutions) are located at the second level of the
 Polynomial Hierarchy~\cite{EiterGottlob95}, thus ASP problems are
 ``harder than NP'' and have a higher worst-case complexity than \CSP
 and \SAT.  In the literature, several restrictions have been
 identified that make ASP tractable
 \cite{GelfondLifschitz88,AptBlairWalker88}.

 \subsection{Contribution}
 In this paper we use the concept of \emph{backdoors} to identify new
 restrictions that make ASP problems tractable.  Small backdoors are
 sets of atoms that represent ``clever reasoning shortcuts'' through
 the search space and represent a hidden structure in the problem
 input.  Backdoors were originally introduced by Williams, Gomes, and
 Selman~\cite{WilliamsGomesSelman03,WilliamsGomesSelman03a} as a tool
 for the analysis of decision heuristics in propositional
 satisfiability. Backdoors have been widely used in the areas of
 propositional
 satisfiability~\cite{WilliamsGomesSelman03,RuanKautzHorvitz04,SamerSzeider08c,KottlerKaufmannSinz08}
 and constraint satisfaction~\cite{GottlobSzeider08}, and also for
 abductive reasoning~\cite{PfandlerRummeleSzeider13},
 argumentation~\cite{DvorakOrdyniakSzeider12}, and quantified Boolean
 formulas~\cite{SamerSzeider09a}.  A backdoor is defined with respect
 to some fixed \emph{target class} for which the computational problem
 under consideration is polynomial-time tractable. The size of the
 backdoor can be seen as a distance measure that indicates how far the
 instance is from the target class.

 In this paper we develop a rigorous theory of backdoors for answer
 set programming. We show that the concept of backdoors can be
 fruitfully adapted for this setting, and that backdoors can serve as
 a \emph{unifying framework} that accommodates several tractable
 restrictions of ASP known from the literature.

 For a worst-case complexity analysis of various problems involving
 backdoors, it is key to pay attention to how running times depend on
 the size of the backdoor, and how well running time scales with
 backdoor size.  \emph{Parameterized Complexity}
 \cite{DowneyFellows99,GaspersSzeider12a,GottlobScarcelloSideri02}
 provides a most suitable theoretical framework for such an
 analysis. It provides the key notion of \emph{fixed-parameter
   tractability} which, in our context, means polynomial-time
 tractability for fixed backdoor size, where the order of the
 polynomial does not depend on the backdoor size.  We show how
 backdoors allow us to deploy recent algorithmic results from
 parameterized complexity theory to the domain of answer set
 programming.

 Parameterized complexity provides tools to provide a rigorous
 analysis of \emph{polynomial-time preprocessing} in terms of
 \emph{kernelization}
 \cite{BodlaenderDowneyFellowsHermelin09,Szeider11}.  A kernelization
 is a polynomial-time self-reduction of a parameterized decision
 problem that outputs a decision equivalent problem instance whose
 size is bounded by a function~$f$ of the parameter (the kernel
 size). It is known that every decidable fixed-parameter tractable
 problem admits a kernelization, but some problems admit small kernels
 (of size polynomial in the parameter) and others don't. We provide
 upper and lower bounds for the kernel size of various ASP problems
 (backdoor detection and backdoor evaluation), taking backdoor size as
 the parameter.

 Several algorithms in the literature are defined for disjunction-free
 (i.e., normal) programs only. We provide a general method for
 \emph{lifting} these parameters to disjunctive programs, preserving
 fixed-parameter tractability under certain conditions.

 % We compare the backdoor size with respect to various base classes
 % with each other and with recently studied parameters, and we
 % demonstrate that several structural restrictions considered in the
 % literature can be stated in terms of backdoors.

 Although our main focus is on a theoretical evaluation, we present
 some experimental results where we consider the backdoor size of
 structured programs and random programs of varied density.

% Our main results can be summarized as follows:
% \begin{enumerate}
% \item We show that the most important computational problems of
%   propositional answer set programming, including brave and
%   skeptical reasoning, and even counting all answer sets, are
%   fixed-parameter tractable when parameterized by the size of the
%   backdoor with respect to various target classes.
% \item We show that the detection of backdoors is fixed-parameter
%   tractable for various target classes, including the class of all
%   Horn programs and classes based on various notions of
%   acyclicity. This way we make recent results of fixed-parameter
%   algorithmics accessible to the field of answer set programming.
% \item We establish preprocessing methods in terms of kernelization for
%   finding backdoors and establish kernel lower bounds for backdoor
%   evaluation.
% \item\todo{REFORMULATE} We establish a lifting theorem that allows us to lift all
%   parameters that are based on deleting rules or atoms from rules from
%   normal programs to disjunctive programs.
% \item We compare the backdoor size with respect to various base
%   classes with each other and with recently studied parameters, and we
%   demonstrate that several structural restrictions considered in the
%   literature can be stated in terms of backdoors.
% \item Although our main focus is on theoretical worst-case results, we
%   present some present some first experimental results where we
%   consider the backdoor size of structured programs and random
%   programs of varied density.
% \end{enumerate}

\subsection{Background and Related Work}

\paragraph{Complexity of ASP Problems}
% \citex{GelfondLifschitz88} have introduced and generalized
% \cite{GelfondLifschitz91} the \emph{stable model semantics} whose
% main concepts \emph{Gelfond-Lifschitz reduct} and \emph{stable set}
% provide the foundation for answer set programming.
Answer set programming is based on the \emph{stable-model semantics}
for logic programs~\cite{GelfondLifschitz88,GelfondLifschitz91}.
% These days, the stable model semantics is the most widely used
% semantic for declarative programming with default negation.
The computational complexity of various problems arising in answer set
programming has been subject of extensive studies.
% \citex{BidoitFroidevaux91} and \citex{MarekTruszczynski91a} have
% shown that the main decision problems of disjunction-free (so called
% \emph{normal}) \ASP are \NP-complete ($\coNP$-complete
% respectively).
\citex{EiterGottlob95} have established that the main decision
problems of (disjunctive) \ASP are located at the second level of the
Polynomial Hierarchy.  Moreover, \citex{BidoitFroidevaux91} and
\citex{MarekTruszczynski91a} have shown that the problems remain
$\NP$-hard ($\coNP$-hard respectively) for disjunction-free (so-called
\emph{normal}) programs.  Several fragments of programs where the main
reasoning problems are polynomial-time tractable have been identified,
e.g., Horn programs~\cite{GelfondLifschitz88}, stratified
programs~\cite{AptBlairWalker88} and programs without even
cycles~\cite{Zhao02}. \citex{DantsinEiterGottlobVoronkov01} survey the
classical complexity of the main reasoning problems for various
semantics of logic programming, including fragments of answer set
programming.

% Transformations of fragments of disjunctive programs into \SAT have
% been considered by \citex{Ben-EliyahuDechter94} and
% \citex{Fages94}. \citex{Ben-EliyahuDechter94} have introduced
% head-cycle-free programs where the main reasoning problems of \ASP are
% \NP-complete, \coNP-complete respectively. \citex{Fages94} has shown
% that the answer sets of a tight program (programs where certain
% cycles on a graph representation of the given program are
% forbidden) are exactly the models of Clark's completion~\cite{Clark78}
% and the computational complexity of the main reasoning problems is in
% \NP or \coNP. The definition of tight programs was
% generalized to disjunctive programs in terms of the concept of loop
% formulas by~\citex{LeeLifschitz03}. \citex{Janhunen04} has suggested a
% transformation from normal programs into \SAT. \citex{JanhunenEtAl06}
% have given a transformation from disjunctive programs into \SAT for
% programs where the number of disjunctions in the heads of rules is
% bounded. The transformation provides a \SAT encoding that consists of
% a guess and check approach. Transformations into other problem domains
% have consequently been designed, e.g., difference
% logic~\cite{JanhunenNiemelaSevalnev09}, SAT modulo
% theories~\cite{JanhunenLiuNiemela11}, and mixed integer linear
% programming~\cite{LiuJanhunenNiemela12}.

\paragraph{ASP Solvers}

Various \ASP-solvers have been developed in recent years. Solvers that
deal with one or more fragments of disjunctive programs (normal,
tight, or head-cycle-free) and utilize techniques from \SAT are
Smodels~\cite{NiemelaSimonsSyrjanen00}, Assat~\cite{LinZhao04a},
Cmodels~\cite{Lierler05a}, and the solver
Clasp~\cite{GebserEtAl11}. Solvers that transform normal programs into
other problem domains are Lp2diff~(difference logic,
\cite{JanhunenNiemelaSevalnev09}), Dingo~(satisfiability modulo
theories, \cite{JanhunenNiemela11}), and Mingo~(mixed integer linear
programming, \cite{LiuJanhunenNiemela12}). Solvers that tackle
disjunctive programs are DLV~\cite{LeoneEtAl06},
GnT~\cite{JanhunenEtAl06}, and ClaspD~\cite{DrescherEtAl08}. DLP
utilizes the technique of unfounded sets~\cite{LeoneRulloScarcello97},
GnT uses techniques from \SAT and extends Smodels by means of a guess
and check approach. ClaspD uses techniques from \SAT and is based on
 logical characterizations of disjunctive loop
formulas~\cite{LeeLifschitz03}.

\paragraph{Parameterizations of \ASP}
So far there has been no rigorous study of disjunctive ASP within the
framework of parameterized complexity. However, several results known
from the literature can be stated in terms of parameterized complexity
and provide fixed-parameter tractability.  The considered parameters
include
the number of atoms of a normal program that appear in negative rule
bodies~\cite{Ben-Eliyahu96},
the number of non-Horn rules of a normal program~\cite{Ben-Eliyahu96},
the size of a smallest feedback vertex set in the dependency digraph
of a normal program~\cite{GottlobScarcelloSideri02},
the number of cycles of even length in the dependency digraph of a
normal program~\cite{LinZhao04},
the treewidth of the incidence graph of a normal
program~\cite{JaklPichlerWoltran09,MorakPichlerRummeleWoltran10},
and~a combination of two parameters: the length of the longest cycle in
the dependency digraph and the treewidth of the interaction graph of a
head-cycle-free programs~\cite{Ben-EliyahuDechter94}.
Very recently we established an fpt-reduction that reduces disjunctive
\ASP to normal \ASP; in other words, a reduction from the second level
of the Polynomial Hierarchy to the first level. The combinatorial
explosion is confined to the size of a smallest backdoor with respect
to normal programs, whereas the considered reasoning problem itself
remains intractable~\cite{FichteSzeider13}.

% have established fixed-parameter tractability for counting and
% enumerating answer sets when parameterized by The parameter
% treewidth measures the similarity of a graph to a tree.
% \citex{MorakPichlerRummeleWoltran10} have designed the solver dynASP
% which works fast on instances with low treewidth,~i.e.,
% treewidth~$\leq 7$.

\paragraph{Backdoors}

% Imagine the search process for solving an intractable problem as
% navigating with a flotilla of boats in a great strom. It is quite hard
% to reach your destination offshore. But if not far there is an island
% where you can navigate easily on inland waterways, you will probably
% seek for a shortcut (backdoor) which takes your flotilla in a hard
% trip towards an island where you can navigate each boat one by one on
% the shallow waters inshore, and reach your destination faster. For an
% intractable problem such an ``island of tractability'' is a known
% sub-class of the instances where the problem is tractable, i.e.,
% solvable in polynomial time. By means of a backdoor one tries to
% identify a structural hard parts of a problem instance to reach an
% island of tractability. 

The concept of a backdoor was originally introduced for SAT and CSP 
by \citexs{WilliamsGomesSelman03,WilliamsGomesSelman03a}.
% have introduced
%backdoors as a tool to the performance analysis of branching
%heuristics in \SAT solvers. 
Since then, backdoors have been used frequently in the literature. The
study of the parameterized complexity of backdoor detection was
initiated by \citex{NishimuraRagdeSzeider04-informal} who considered
satisfiability backdoors for the base classes Horn and 2CNF. Since
then, the study has been extended to various other base classes,
including clustering formulas \cite{NishimuraRagdeSzeider07},
renamable Horn formulas~\cite{RazgonOSullivan09}, QHorn
formulas~\cite{GaspersOrdyniakRamanujanSaurabhSzeider13}, Nested
formulas~\cite{GaspersSzeider12c}, acyclic
formulas~\cite{GaspersSzeider12b}, and formulas of bounded incidence
treewidth~\cite{GaspersSzeider13}; for a survey,
see~\cite{GaspersSzeider12a}.
% \citex{SamerSzeider08b}
% have introduced a refined concept of backdoors that takes the
% relationship between backdoor variables into account. 
Several results extend the concept of backdoors to other problems,
e.g., backdoor sets for constraint satisfaction
problems~\cite{WilliamsGomesSelman03}, quantified Boolean
formulas~\cite{SamerSzeider09a}, abstract
argumentation~\cite{OrdyniakSzeider11}, and abductive
reasoning~\cite{PfandlerRummeleSzeider13}. \citex{SamerSzeider08b}
have introduced \emph{backdoor trees} for propositional satisfiability
which provide a more refined concept of backdoor evaluation and take
the interaction of variables that form a backdoor into account.

\subsection{Prior Work and Paper Organization}
This paper is an extended and updated version of the papers that
appeared in the proceedings of the 22nd International Conference on
Artificial Intelligence~\cite{FichteSzeider11} and in the New
Directions in Logic, Language and Computation~\cite{Fichte12}. The
present paper provides a higher level of detail, in particular full
proofs and more examples. Furthermore, the paper extends its previous
versions in the following way: additional attention is payed to the
minimality check (Lemma~\ref{lem:mincheck}).  Theorem~\ref{the:fvs} is
extended to entail some very recent results in parameterized
complexity theory. A completely new section (Section~\ref{sec:kernel})
is devoted to a rigorous analysis of preprocessing methods for the
problems of backdoor detection and backdoor evaluation.
%
% We refine the notation of backdoor evaluation by the notion of
% backdoor trees which takes the interaction of atoms in a backdoor
% into account (Section~\ref{sec:backdoor-trees}).
%
We present a general method to lift parameters from rules of normal
programs to disjunctive programs (Section~\ref{sec:lifting}). We
extend the section on the theoretical comparison of parameters
(Section~\ref{sec:comparision}) by additional comparisons to other
parameters, e.g., weak feedback width and interaction graph treewidth,
and to other classes of programs, e.g., head-cycle-free and tight
programs. Finally, in Section~\ref{sec:experiments} we provide some
empirical data on backdoor detection and discuss the evaluation of
backdoors in a practical setting.

\section{Preliminaries}\label{sec:prelims}
\subsection{Answer Set Programming}
We consider a universe~$U$ of propositional \emph{atoms}.  A
\emph{literal} is an atom~$a\in U$ or its negation~$\neg a$.  A
\emph{disjunctive logic program} (or simply a \emph{program}) $P$ is a
set of \emph{rules} of the following form
\begin{align*}
  x_1\por \dots \por x_l \quad\leftarrow\quad y_1,\dots,y_m,\pnot z_1,
  \dots,\pnot z_n
\end{align*}
where $x_1,\dots,x_l, y_1,\dots,y_m, z_1,\dots, z_n$ are atoms and
$l,m,n$ are non-negative integers. Let $r$ be a rule. We write $\{x_1,
\dots, x_l\} = H(r)$ (the \emph{head} of $r$), $\{y_1, \dots, y_m \} =
B^+(r)$ (the positive body of $r$) and $\{z_1,\dots,z_n\}=B^-(r)$ (the
negative body of $r$). We denote the sets of atoms occurring in a
rule~$r$ or in a program~$P$ by $\at(r)=H(r) \cup B^+(r)\cup B^-(r)$
and $\at(P)=\bigcup_{r\in P} \at(r)$, respectively. A rule~$r$ is
\emph{negation-free} if $B^-(r)=\emptyset$, $r$ is \emph{normal} if
$\Card{H(r)}\leq 1$, $r$ is a \emph{constraint} if $\Card{H(r)}=0$,
$r$ is \emph{constraint-free} if $\Card{H(r)>0}$, $r$ is \emph{Horn}
if it is negation-free and normal, $r$ is \emph{positive} if it is
Horn and constraint-free, $r$ is \emph{tautological} if $B^+(r) \cap
(H(r) \cup B^-(r))\neq \emptyset$, and $r$ is \emph{non-tautological}
if it is not tautological.
%
%, and $r$ is \emph{tautological-free} if $B^+(r)
%\cap (H(r) \cup B^-(r)) = \emptyset$. 
We say that a program has a certain property if all its rules have the
property. 
\class{Horn} refers to the class of all Horn programs. We denote the
class of all normal programs by \Normal.  
Let $P$ and $P'$ be programs. We say that $P'$ is a \emph{subprogram}
of $P$ (in symbols $P' \subseteq P$) if for each rule~$r'\in P'$ there
is some rule~$r\in P$ with $H(r')\subseteq H(r)$, $B^+(r')\subseteq
B^+(r)$, $B^-(r')\subseteq B^-(r)$.
We call a class~$\CCC$ of programs \emph{hereditary} if for each~$P\in
\CCC$ all subprograms of $P$ are in $\CCC$ as well. Note that many
natural classes of programs (and all classes considered in this paper)
are hereditary.

A set~$M$ of atoms \emph{satisfies} a rule~$r$ if $(H(r)\,\cup\,
B^-(r)) \,\cap\, M \neq \emptyset$ or $B^+(r) \setminus M \neq
\emptyset$.  $M$ is a \emph{model} of $P$ if it satisfies all rules of
$P$.  The \emph{Gelfond-Lifschitz (GL) reduct} of a program~$P$ under
a set~$M$ of atoms is the program~$P^M$ obtained from $P$ by first
removing all rules~$r$ with $B^-(r)\cap M\neq \emptyset$ and second
removing all~$\neg z$ where $z \in B^-(r)$ from the remaining rules
$r$~\cite{GelfondLifschitz91}. $M$ is an \emph{answer set} (or
\emph{stable model}) of a program~$P$ if $M$ is a minimal model of
$P^M$. We denote by $\stableset(P)$ the set of all answer sets of~$P$.

\begin{EX}\label{ex:running}
  Consider the program~$P$ consisting of the following rules:
% \(P=\{d \leftarrow a, e \rsep a \leftarrow d,
%   \pnot b, \pnot c \rsep b \leftarrow c\rsep e \por c \leftarrow f
%   \rsep c \leftarrow d \rsep c \leftarrow f, e, \pnot b \rsep f
%   \leftarrow d, c \rsep f \}.\)
\begin{align*}
d &\leftarrow a, e \rsep&
 a &\leftarrow d, \pnot b, \pnot c \rsep&
e \por c &\leftarrow f \rsep\\ 
 f &\leftarrow d, c \rsep&
 c &\leftarrow f, e, \pnot b \rsep&
 c &\leftarrow d \rsep\\
b &\leftarrow c\rsep&
 f&.
\end{align*}
The set~$M= \{ b, c, f\}$ is an answer set of $P$, since $P^M = \{d
\leftarrow a, e \rsep f \leftarrow d, c \rsep b \leftarrow c\rsep e
\por c \leftarrow f\rsep c \leftarrow d \rsep f \}$ and the minimal
models of $P^M$ are $\{b,c,f\}$ and $\{e,f\}$.
\end{EX}

It is well known that normal Horn programs have a unique answer set
and that this set can be found in linear time. Van Emden and
Kowalski~\shortcite{Van-EmdenKowalski76} have shown that every
constraint-free Horn program has a unique minimal model. Dowling and
Gallier~\shortcite{DowlingGallier84} have established a linear-time
algorithm for testing the satisfiability of propositional Horn
formulas which easily extends to Horn programs. In the following we
state the well-known linear-time result.

\begin{LEM}%[\cite{DowlingGallier84}]
\label{lem:horn-lineartime}
Every Horn program has at most one model, and this model can be found
in linear time.
\end{LEM}
% \begin{proof}[Proof (Sketch).] 
%  %Scutella~\cite{Scutella90} corrected the top-down algorithm. for certain cases.
% \end{proof}

\subsection{\ASP Problems}

We consider the following fundamental  \ASP problems.

\problem{\AspCheck}{a program~$P$ and a set~$M\subseteq \at(P)$}
{decide whether $M$ is an answer set of $P$}
\problem{\AspCons}{A program~$P$}{decide whether $P$ has an answer
  set} 
\problem{\AspBrave}{A program~$P$ and an atom~$a^*\in \at(P)$}{decide
  whether $a^*$ belongs to \emph{some} answer set \mbox{of $P$}}
\problem{\AspCaut}{A program~$P$ and an atom~$a^*\in \at(P)$}{decide
  whether $a^*$ belongs to \emph{all} answer sets of $P$}
\problem{\AspCount}{A program~$P$}{Compute the number of answer
  sets of $P$} 
\problem{\AspEnum}{A program~$P$}{List all answer sets of $P$}

We denote by $\AspReason$ the family of the reasoning problems
\AspCheck, \AspCons, and \AspBrave and by $\AspFull$ the family of all
the problems defined above. This $\AspReason$ consists of decision
problems, and $\AspFull$ adds to it a counting and an enumeration
problem. In the sequel we will occasionally write $L_\Normal$ to
denote a problem~$L \in \AspFull$ restricted to input programs from
$\Normal$.

\AspCheck is $\coNP$\hy hard in general~\cite{EiterGottlob95}, but
$\AspCheck_\Normal$ is polynomial~\cite{CadoliLenzerini94}. \AspCons
and \AspBrave are $\Sigma^P_2$-complete, \AspCaut is
$\Pi^P_2$-complete~\cite{EiterGottlob95}. Both reasoning problems
remain $\NP$\hy hard (or $\coNP$\hy hard) for normal
programs~\cite{MarekTruszczynski91}, but are polynomial-time solvable
for Horn programs~\cite{GelfondLifschitz88}.  \AspCount is easily seen
to be $\#P$-hard\footnote{$\#P$ is the complexity class consisting of
  all the counting problems associated with the decision problems in
  \NP.} as it entails the problem $\#$SAT.

\subsection{Parameterized Complexity}\label{sec:PC}
We briefly give a basic background on parameterized complexity. For
more detailed information we refer to other
sources~\cite{DowneyFellows99,FlumGrohe06,GottlobSzeider08,Niedermeier06}.
An instance of a \emph{parameterized problem}~$L$ is a pair~$(I,k)\in
\Sigma^* \times \Nat$ for some finite alphabet~$\Sigma$. For an
instance~$(I,k) \in \Sigma^* \times \Nat$ we call $I$ the \emph{main
  part} and $k$ the \emph{parameter}. $\CCard{I}$ denotes the size
of~$I$. $L$ is \emph{fixed-parameter tractable} if there exist a
computable function~$f$ and a constant~$c$ such that we can decide
whether $(I,k)\in L$ in time $\BigO{f(k)\CCard{I}^c}$. Such an
algorithm is called an \emph{fpt-algorithm}.  If $L$ is a decision
problem, then we identify $L$ with the set of all yes-instances
$(I,k)$. $\FPT$ is the class of all fixed-parameter tractable decision
problems.

% DEFINE LATER
% The
% \emph{unparameterized version} of $L$ is the classical problem $\SB
% I\#u^k : (I, k) \in L \SE$ where $u$ denotes an arbitrary symbol from
% $\Sigma$ and $\#$ is a new symbol not in $\Sigma$.
%

Let $L \subseteq \Sigma^* \times \Nat$ and $L' \subseteq
\Sigma'^*\times \Nat$ be two parameterized decision problems for some
finite alphabets~$\Sigma$ and $\Sigma'$. An \emph{fpt-reduction}~$r$
from $L$ to $L'$ is a many-to-one reduction from~$\Sigma^*\times \Nat$
to~$\Sigma'^*\times \Nat$ such that for all~$I \in \Sigma^*$ we have
$(I,k) \in L$ if and only if $r(I,k)=(I',k')\in L'$ and $k' \leq g(k)$
for a fixed computable function~$g: \Nat \rightarrow \Nat$ and there
is a computable function~$f$ and a constant~$c$ such that $r$ is
computable in time~$\BigO{f(k)\CCard{I}^c}$. Thus, an fpt-reduction
is, in particular, an fpt-algorithm. It is easy to see that the
class~$\FPT$ is closed under fpt-reductions and it is clear for
parameterized problems~$L_1$ and $L_2$ that if $L_1\in \FPT$ and there
is an fpt-reduction from~$L_2$ to~$L_1$, then $L_2 \in \FPT$.
% We would like to note that the theory of fixed-parameter
% intractability is based on
% fpt-reductions~\cite{DowneyFellows99,FlumGrohe06}.

The \emph{Weft Hierarchy} consists of parameterized complexity
classes~$\W[1] \subseteq \W[2]\subseteq\nobreak\cdots$ which are
defined as the closure of certain parameterized problems under
parameterized reductions. There is strong theoretical evidence that
parameterized problems that are hard for classes~$\W[$i$]$ are not
fixed-parameter tractable. A prominent $W[2]$-complete problem is
\HS~\cite{DowneyFellows99} defined as follows:
\pproblem{\HS}{A family of sets~$(\SSS,k)$ where
  $\SSS=\{S_1,\dots,S_m\}$ and an integer~$k$.}{The
  integer~$k$.}{Decide whether there exists set~$H$ of size at
  most~$k$ which intersects with all the~$S_i$ ($H$~is a \emph{hitting
    set} of $\SSS$).}
The class~$\XP$ of \emph{non-uniform} tractable problems consists of
all parameterized decision problems that can be solved in polynomial
time if the parameter is considered constant. That is, $(I,k)\in L$
can be decided in time~$\BigO{\CCard{I}^{f(k)}}$ for some computable
function~$f$. The parameterized complexity class~$\paraNP$ contains
all parameterized decision problems~$L$ such that $(I,k)\in L$ can be
decided \emph{non-deterministically} in time~$O(f(k)\CCard{I}^c)$ for
some computable function~$f$ and constant~$c$. A parameterized
decision problem is $\paraNP$-complete if it is in $\NP$ and $\NP$\hy
complete when restricted to a finite number of parameter
values~\cite{FlumGrohe06}. By $\coparaNP$ we denote the class of all
parameterized decision problems whose complement (yes and no instances
swapped) is in $\paraNP$. Using the concepts and
terminology~of~\citex{FlumGrohe06}, $\coparaNP=\paracoNP$.

% Let $V$ be a finite set and $E\subseteq V \times V$, then we call
% the tuple~$G=(V,E)$ a \emph{directed graph}, $V$ the \emph{vertices}
% of $G$ and $E$ the \emph{edges} of $G$. Now, let $E$ be a subset of
% $2$-elementary subsets of $V$, then we call the tuple~$G=(V,E)$ an
% \emph{graph}.

%
% A \emph{signed (directed) graph} is a graph whose edges are either
% positive (unlabeled) or negative.
%

\subsection{Graphs}
\label{sec:graphs}
We recall some notations of graph theory. We consider undirected and
directed graphs. An \emph{undirected graph} or simply a \emph{graph}
is a pair~$G=(V,E)$ where $V\neq \emptyset$ is a set of
\emph{vertices} and $E \subseteq \SB \{u,v\}\subseteq V \SM u \neq v
\SE$ is a set of \emph{edges}. We denote an edge~$\{v,w\}$ by $uv$ or
$vu$.  A graph~$G'=(V',E')$ is a \emph{subgraph} of $G$ if
$V'\subseteq V$ and $E'\subseteq E$ and an \emph{induced subgraph} if
additionally for any $u,v \in V'$ and $uv \in E$ also $uv \in E'$. A
\emph{path of length~$k$} is a graph with $k+1$ pairwise distinct
vertices~$v_1,\dots,v_{k+1}$, and $k$ distinct edges~$v_iv_{i+1}$
where $1\leq i \leq k$ (possibly $k=0$). A \emph{cycle of length $k$},
is a graph that consists of $k$ distinct vertices~$v_1,v_2,\ldots,
v_k$ and $k$ distinct edges $v_1v_2, \dots, v_{k-1}v_k,v_kv_1$.
Let $G=(V,E)$ be a graph. $G$ is \emph{bipartite} if the set~$V$ of
vertices can be divided into two disjoint sets~$U$ and $V$ such that
there is no edge~$uv \in E$ with $u,v\in U$ or $u,v\in V$.
$G$ is \emph{complete} if for any two vertices~$u,v \in V$ there is an
edge~$uv \in E$. $G$ contains a \emph{clique} on $V'\subseteq V$ if
the induced subgraph~$(V',E')$ of $G$ is a complete graph.
A \emph{connected component}~$C$ of $G$ is an inclusion-maximal
subgraph~$C=(V_C,E_C)$ of $G$ such that for any two vertices~$u, v \in
V_C$ there is a path in~$C$ from~$u$ to~$v$.

A \emph{directed graph} or simply a \emph{digraph} is a pair~$G=(V,E)$
where $V \neq \emptyset$ is a set of vertices and $E\subseteq \SB
(u,v) \in V \times V \SM u \neq v \SE$ is a set of \emph{directed
  edges}. A digraph~$G'=(V',E')$ is a \emph{subdigraph} of $G$ if
$V'\subseteq V$ and $E'\subseteq E$ and an \emph{induced subdigraph}
if additionally for any $u,v \in V'$ and $(u,v) \in E$ also $(u,v) \in
E'$.  A \emph{directed path of length~$k$} is a digraph with $k+1$
pairwise distinct vertices $v_1,\dots,v_{k+1}$, and $k$ distinct
edges~$(v_i,v_{i+1})$ where $1\leq i \leq k$ (possibly $k=0$). A
\emph{directed cycle of length~$k$}, is a digraph that consists of $k$
distinct vertices~$v_1,v_2,\ldots, v_k$ and $k$ distinct
edges~$(v_1,v_2), \dots, (v_{k-1},v_k), (v_k,v_1)$.

We sometimes denote a (directed) path or (directed) cycle as a
sequence of vertices. Please observe that according to the above
definitions, the length of an undirected cycle is at least~3, whereas
the length of a directed cycle is at least~2.

A \emph{strongly connected component}~$C$ of a digraph~$G=(V,E)$ is an
inclusion-maximal directed subgraph~$C=(V_C,E_C)$ of $G$ such that for
any two vertices~$u, v \in V_C$ there are paths in~$C$ from~$u$ to~$v$
and from~$v$ to~$u$. The strongly connected components of $G$ form a
partition of the set~$V$ of vertices, we denote this partition by
$\SCC(G)$.

For further basic terminology on graphs and digraphs we refer to a
standard text~\cite{Diestel00,BondyMurty08}.

\subsection{Satisfiability Backdoors}
We also need some notions from \emph{propositional satisfiability}.  A
\emph{literal} is an atom or its negation and a \emph{clause} is a
finite set of literals, a CNF formula is a finite set of clauses.  A
\emph{truth assignment} is a mapping $\tau:X\rightarrow \{0,1\}$
defined for a set~$X\subseteq U$ of atoms. For $x\in X$ we put
$\tau(\pnot x)=1 - \tau(x)$.
% NOT CLEAR WHAT THIS MEANS?
%If it is clear from the context
%we abbreviate $x$ if $\tau(x) = 1$ and $\bar x$ for $\tau(x) = 0 $
%e.g., $\tau \cup \{x\}$.  
By $\ta{X}$ we denote the set of all truth assignments~$\tau: X
\rightarrow \{0,1\}$. The \emph{truth assignment reduct} of a CNF
formula~$F$ with respect to~$\tau \in \ta{X}$ is the CNF
formula~$F_\tau$ obtained from~$F$ by first removing all clauses~$c$
that contain a literal set to~$1$ by $\tau$, and second removing from
the remaining clauses all literals set to~$0$ by $\tau$. $\tau$
\emph{satisfies} $F$ if $F_\tau=\emptyset$, and $F$ is
\emph{satisfiable}
% USED?
%(in symbols $\text{sat}(F)$) 
if it is satisfied by some~$\tau$.

% $\tau(x)=\begin{cases} \tau_1(x) & \text{ if } x \in X_1\\ \tau_2(x)
%   & \text{ if } x \in X_2\end{cases}$

% REPETITION
% Strong-Backdoors of propositional CNF formulas are small sets of atoms
% introduced originally by Williams, Gomes, and
% Selman~\cite{WilliamsGomesSelman03,WilliamsGomesSelman03a} as a
% concept for the analysis of decision heuristics in propositional
% satisfiability (see for a survey~\cite{GaspersSzeider12a}).
%
The following is obvious from the definitions:
\begin{OBS}%
\label{obs:cnfsat}
Let $F$ be a CNF formula and $X$ a set of atoms. $F$ is satisfiable if
and only if $F_\tau$ is satisfiable for at least one truth assignment
$\tau\in \ta{X}$.
\end{OBS}%
This leads to the definition of a strong backdoor relative to a
class~$\CCC$ of polynomially solvable CNF formulas: a set~$X$ of atoms
is a \emph{strong $\CCC$-backdoor} of a CNF formula~$F$ if $F_\tau\in
\CCC$ for all truth assignments~$\tau\in \ta{X}$. Assume that the
satisfiability of formulas~$F\in \CCC$ of size~$\CCard{F}=n$ can be
decided in time~$O(n^c)$. Then we can decide the satisfiability of an
arbitrary formula~$F$ for which we know a \strongBds{\CCC} of size~$k$
in time~$O(2^k n^c)$ which is efficient as long as $k$ remains small.

A further variant of backdoors are deletion backdoors defined by
removing literals from a CNF formula. $F-X$ denotes the formula
obtained from~$F$ by removing all literals~$x,\neg x$ for $x \in X$
from the clauses of $F$. Then a set~$X$ of atoms is a \emph{deletion
  $\CCC$\hy backdoor} of $F$ if $F-X \in \CCC$. In general,
\delBds{\CCC}s are not necessarily \strongBds{\CCC}s. If all subsets
of a formula in $\CCC$ also belong to~$\CCC$ ($\CCC$ is
clause-induced), then \delBds{\CCC}s are \strongBds{\CCC}s.

Before we can use a strong backdoor we need to find it first. For most
reasonable target classes~$\CCC$ the detection of a \strongBds{\CCC}
of size at most~$k$ is $\NP$-hard if $k$ is part of the
input. However, as we are interested in finding \emph{small}
backdoors, it makes sense to parameterize the backdoor search by $k$
and consider the parameterized complexity of backdoor
detection. Indeed, with respect to the classes of Horn CNF formulas
and 2-CNF formulas, the detection of strong backdoors of size at
most~$k$ is fixed-parameter
tractable~\cite{NishimuraRagdeSzeider04-informal}.  The parameterized
complexity of backdoor detection for many further target classes has
been investigated~\cite{GaspersSzeider12a}.

% For other target classes (clustering formulas and renamable Horn
% formulas) the detection of deletion backdoors (a subclass of strong
% backdoors) of size at most~$k$ is fixed-parameter
% tractable~\cite{NishimuraRagdeSzeider07,RazgonOSullivan08}.
% Furthermore, for the class of all acyclic formulas (CNF formulas
% whose incidence graphs are forests) the detection of strong
% backdoors of size at most~$k$ has an fpt-approximation and the
% detection of deletion backdoors of size at most~$k$ is
% fixed-parameter tractable~\cite{GaspersSzeider12}

\section{ Answer Set Backdoors}\label{sec:backdoors}
\subsection{Strong Backdoors}
In order to translate the notion of backdoors to the domain of \ASP,
we first need to come up with a suitable concept of a reduction with
respect to a truth assignment. The following is a natural definition
which generalizes a concept of~\citex{GottlobScarcelloSideri02}.
 
\begin{DEF}
\label{def:tar}
Let $P$ be a program, $X$ a set of atoms, and $\tau\in \ta{X}$. The
\emph{truth assignment reduct} of $P$ under $\tau$ is the logic
program $P_\tau$ obtained from~$P$ by
\begin{enumerate}
\item removing all rules~$r$ with $H(r)\cap \tau^{-1}(1)\neq
  \emptyset$ or $H(r)\subseteq X$;
\item removing all rules~$r$ with $B^+(r) \cap \tau^{-1}(0)\neq
  \emptyset$;
\item removing all rules~$r$ with $B^-(r) \cap \tau^{-1}(1)\neq
  \emptyset$;
\item removing from the heads and bodies of the remaining rules all
  literals~$v,\pnot v$ with $v\in X$.
\end{enumerate}
\end{DEF}

\begin{DEF}
  Let $\CCC$ be a class of programs. A set~$X$ of atoms is a
  \emph{strong $\CCC$\hy backdoor} of a program~$P$ if $P_{\tau}\in
  \CCC$ for all truth assignments~$\tau\in \ta{X}$.
\end{DEF}
By a \emph{minimal} \strongBds{\CCC} of a program $P$ we mean a
\strongBds{\CCC} of $P$ that does not properly contain a smaller
\strongBds{\CCC} of $P$; a \emph{smallest} \strongBds{\CCC} of $P$ is
one of smallest cardinality.

\begin{EX}\label{ex:strong-bds}
  We consider the program of Example~\ref{ex:running}. The
  set~$\{b,c\}$ is a \strongBds{\Horn} since all four truth assignment
  reducts~$P_{\bar b \bar c}=\{d \leftarrow a, e \rsep a \leftarrow d
  \rsep e \leftarrow f \rsep f\}$, $P_{\bar b, c}=\{d \leftarrow a, e
  \rsep f \leftarrow d \rsep f \}$, $P_{b \bar c}=\{d \leftarrow a, e
  \rsep e \leftarrow f\rsep f\}$, and $P_{b c}=\{d \leftarrow a, e
  \rsep f \leftarrow d \rsep f\}$ are in the class~\Horn.
\end{EX}

\subsection{Deletion Backdoors}
Next we define a variant of answer set backdoors similar to
satisfiability deletion backdoors. For a program~$P$ and a set~$X$ of
atoms we define $P-X$ as the program obtained from $P$ by deleting $a,
\pnot a$ for $a \in X$ from the rules of $P$.
% contained in $X$ from all the rules (heads and bodies) of~$P$
The definition gives rise to deletion backdoors. We will see that
finding deletion backdoors is in some cases easier than finding strong
backdoors.

\begin{DEF}
  Let $\CCC$ be a class of programs. A set~$X$ of atoms is a
  \emph{deletion $\CCC$\hy backdoor} of a program~$P$ if $P-X\in
  \CCC$.
%The set~$X$ is a \emph{smallest \delBds{\CCC}} if $P$ has no
%  \delBds{\CCC} that is smaller than $X$. We say $X$ is a
%  \emph{minimal \delBds{\CCC}} if $P$ has no \delBds{\CCC} that is a
%  proper subset of $X$.
\end{DEF}

In general, not every \strongBds{\CCC} is a \delBds{\CCC}, and not
every \delBds{\CCC} is a \strongBds{\CCC}. But we can strengthen one
direction requiring the base class to satisfy the very mild condition
of being hereditary (see Section~\ref{sec:prelims}) which holds for
all base classes considered in this paper.

% We call $\CCC$ to be \emph{rule-induced} if for each~$P\in \CCC$,
% $P'\subseteq P$ implies $P'\in \CCC$.
% Note that many natural classes of programs (and all classes considered
% in this paper) are rule-induced.

\begin{LEM}\label{lem:rule-induced}
  If $\CCC$ is hereditary, then every~\delBds{\CCC} is a
  \strongBds{\CCC}.
\end{LEM}
\begin{proof}
  Let $P$ be a program, $X\subseteq \at(P)$, and $\tau\in \ta{X}$. Let
  $r'\in P_\tau$. It follows from Definition~\ref{def:tar} that $r'$
  is obtained from some~$r\in P$ by deleting~$v,\lnot v$ for all~$v\in
  X$ from the head and body of $r$.  Consequently $r'\in P-X$. Hence
  $P_\tau \subseteq P-X$ which establishes the proposition.
%   We show the statement by proving that $P_\tau \subseteq P-X$ for
%   every $\tau\in \ta{X}$ and for every program $P\in \CCC$. Let $P$ be
%   a program and $X\subseteq \at(P)$. We choose arbitrarily a truth
%   assignment $\tau \in \ta{X}$. For a rule $r\in P$ if the conditions
%   (1), (2), and (3) of Definition~\ref{def:tar} do not apply, then all
%   literals~$x,\pnot x$ with $x\in X$ are removed from the heads $H(r)$
%   and bodies~$B^+(r)\cup B^-(r)$ by the truth assignment reduct of $P$
%   under $\tau$. This is also done by $P-X$. If at least one of the
%   conditions (1), (2), or (3) applies, then $r\notin P_\tau$. Hence
%   $P_\tau \subseteq P-X$ for every $\tau \in \ta{X}$ and every
%   program~$P$. Thus the proposition follows.
\end{proof}

% \begin{proof}
%   The statement follows from the fact that $P_\tau \subseteq P-X$ for
%   every $\tau\in \ta{X}$ and every program~$P$.
% \end{proof}

\subsection{Backdoor Evaluation}
An analogue to Observation~\ref{obs:cnfsat} does not hold for \ASP,
even if we consider the most basic problem \AspCons. Take for example
the program~$P=\SB x \leftarrow y \rsep y \leftarrow x \rsep
\leftarrow x \rsep z \leftarrow \pnot x \SE$ and the
set~$X=\{x\}$. Both reducts~$P_{x=0}=\SB z \SE$ and $P_{x=1}=\SB y
\SE$ have answer sets, but $P$ has no answer set.
However, we can show a somewhat weaker asymmetric variant of
Observation~\ref{obs:cnfsat}, where we can map each answer set of~$P$
to an answer set of $P_\tau$ for some~$\tau\in \ta{X}$. This is made
precise by the following definition and lemma (which are key for a
backdoor approach to answer set programming).

% $P=\SB x \leftarrow\pnot x;\, y \leftarrow \SE $ and the set
% $X=\{x\}$.  Both reducts $P_{x=0}=\SB y \leftarrow \SE$ and
% $P_{x=1}=\SB y \leftarrow \SE$ have answer sets, but $P$ has no answer
% set.  \stefan{Is there a less trivial example?}  

\begin{DEF}
Let $P$ be a program and $X$ a set of atoms.  We define
 \[\stableset(P,X) =  \SB M\cup \tau^{-1}(1) \SM 
 \tau\in \ta{X\cap\, \at(P)}, M \in \stableset(P_\tau)\SE.\]
\end{DEF}

\begin{LEM}\label{lem:subset}
  $\stableset(P) \subseteq \stableset(P,X)$ holds for every
  program~$P$ and every set~$X$ of atoms.
\end{LEM}
\begin{proof}
  Let $M\in \stableset(P)$ be chosen arbitrarily. We put $X_0 = (X
  \setminus M)\cap \at(P)$ and $X_1=X \cap M$ and define a truth
  assignment~$\tau\in \ta{X\cap \at(P)}$ by setting $\tau^{-1}(i)=X_i$
  for $i\in \{0,1\}$.  Let $M'=M\setminus X_1$.  Observe that $M'\in
  \stableset(P_\tau)$ implies $M\in \stableset(P,X)$ since $M= M'\cup
  \tau^{-1}(1)$ by definition.  Hence, to establish the lemma, it
  suffices to show that $M'\in \stableset(P_\tau)$.  We have to show
  that $M'$ is a model of $P_\tau^{M'}$, and that no proper subset of
  $M'$ is a model of $P_\tau^{M'}$.

%
%UP NEXT
%
  In order to show that $M'$ is a model of $P_\tau^{M'}$, choose
  $r'\in P_\tau^{M'}$ arbitrarily. By construction of $P_\tau^{M'}$
  there is a corresponding rule~$r\in P$ with $H(r')=H(r)\setminus
  X_0$ and $B^+(r')=B^+(r)\setminus X_1$ which gives rise to a rule
  $r''\in P_\tau$, and in turn, $r''$ gives rise to $r'\in
  \P_\tau^{M'}$.  Since $B^-(r)\cap X_1=\emptyset$ (otherwise $r$
  would have been deleted forming $P_\tau$) and $B^-(r)\cap
  M'=\emptyset$ (otherwise $r''$ would have been deleted forming
  $P_\tau^{M'}$), it follows that $B^-(r)\cap M=\emptyset$. Thus $r$
  gives rise to a rule~$r^*\in P^M$ with $H(r)=H(r^*)$ and
  $B^+(r)=B^+(r^*)$. Since $M\in \stableset(P)$, $M$ satisfies $r^*$,
  i.e., $H(r)\cap M\neq \emptyset$ or $B^+(r)\setminus M\neq
  \emptyset$. However, $H(r)\cap M=H(r')\cap M'$ and $B^+(r)\setminus
  M = B^+(r')\setminus M'$, thus $M'$ satisfies $r'$.  Since $r'\in
  P_\tau^{M'}$ was chosen arbitrarily, we conclude that $M'$ is a
  model of $P_\tau^{M'}$.

  In order to show that no proper subset of $M'$ is a model of
  $P_\tau^{M'}$ choose arbitrarily a proper subset~$N'\subsetneq M'$.
  Let $N=N'\cup X_1$. Since $M'=M\setminus X_1$ and $X_1\subseteq M$
  it follows that $N \subsetneq M$. Since $M$ is a minimal model of
  $P^M$, $N$ cannot be a model of $P^M$. Consequently, there must be a
  rule~$r\in P$ such that $B^-(r)\cap M= \emptyset$ (i.e., $r$ is not
  deleted by forming $P^M$), $B^+(r)\subseteq N$ and $H(r)\cap
  N=\emptyset$.  However, since $M$ satisfies $P^M$, and since
  $B^+(r)\subseteq N\subseteq M$, $H(r)\cap M\neq \emptyset$. Thus $r$
  is not a constraint. Moreover, since $H(r)\cap M\neq \emptyset$ and
  $M\cap X_0=\emptyset$, it follows that $H(r)\setminus X_0\neq
  \emptyset$. Thus, since $H(r)\cap X_1 =\emptyset$, $H(r) \setminus X
  \neq \emptyset$. We conclude that $r$ is not deleted when forming
  $P_\tau$ and giving rise to a rule~$r'\in P_\tau$, which in turn is
  not deleted when forming $P_\tau^{M'}$, giving rise to a rule~$r''$,
  with $H(r'') = H(r)\setminus X_0$, $B^+(r'')=B^+(r)\setminus X_1$,
  and $B^-(r'')=\emptyset$.  Since $B^+(r'')\subseteq N'$ and
  $H(r'')\cap N=\emptyset$, $N'$ is not a model of $P_\tau^{M'}$.

  Thus we have established that $M'$ is a stable model of $P_\tau$, and
  so the lemma follows.
\end{proof}

In view of Lemma~\ref{lem:subset} we shall refer to the elements in
$\stableset(P,X)$ as ``answer set candidates.''

\begin{EX}\label{ex:expoit_bds}
  We consider program~$P$ of Example~\ref{ex:running} and the
  \strongBds{\Horn} $X=\{b,c\}$ of Example~\ref{ex:strong-bds}. The
  answer sets of $P_\tau$ are $\stableset(P_{\bar b \bar
    c})=\{\{e,f\}\}$, $\stableset(P_{\bar b c})=\{\{f\}\}$,
  $\stableset(P_{b \bar c})=\{\{e,f\}\}$, and
  $\stableset(P_{bc})=\{\{f\}\}$ for $\tau \in {\ta{\{b,c\}}}$.  We
  obtain the set~$\stableset(P,X)=\{ \{e,f\}, \{c,f\}, \{b,e,f\},
  \{b,c,f\}\}$.
\end{EX}

In view of Lemmas~\ref{lem:subset}, we can compute $\stableset(P)$ by
(i)~computing $\stableset(P_\tau)$ for all $\tau\in \ta{X}$ (this
produces the set~$\stableset(P,X)$ of candidates for $\stableset(P
)$), and (ii)~checking for each $M \in \stableset(P,X)$ whether it is
an answer set of $P$. The check~(ii) entails (iia)~checking whether
$M\in \stableset(P,X)$ is a model of P and (iib) whether $M\in
\stableset(P,X)$ is a minimal model of $P^M$. We would like to note
that in particular any constraint contained in $P$ is removed in the
truth assignment reduct~$P_\tau$ but considered in
check~(iia). Clearly check (iia) can be carried out in polynomial time
for each~$M$. Check (iib), however, is $\coNP$\hy hard in
general~\cite{MarekTruszczynski91}, but polynomial for normal
programs~\cite{CadoliLenzerini94}.

Fortunately, for our considerations it suffices to perform check~(iib)
for programs that are ``close to \Normal,'' and so the check is
fixed-parameter tractable in the size of the given backdoor.  More
precisely, we consider the following parameterized problem and
establish its fixed-parameter tractability in the next lemma.

\pproblem{\BdCheck{\CCC}}{A program~$P$, a \strongBds{\CCC} $X$ of $P$
  and a set~$M\subseteq \at(P)$.}{The size~$\Card{X}$ of the
  backdoor.}{Decide whether $M$ is an answer set of $P$.}

\begin{LEM}\label{lem:mincheck}
  Let \CCC be a class of normal programs. The problem~\BdCheck{\CCC}
  is fixed-parameter tractable.
%Let $\CCC$ be a class of normal programs. Given a program $P$, a strong
%$\CCC$\hy backdoor set $X$ of $P$, and $M \in
%\stableset(P,X)$. Then deciding whether $M$ is an answer set of $P$ is
%fixed-parameter tractable for pa\-ram\-eter~$\Card{X}$. In
%particular, this decision can be made in time $O(2^k n)$ where $n$ denotes
%the input size of $P$ and $k=\Card{X}$.
\end{LEM}
\begin{proof}
  Let $\CCC$ be a class of normal programs, $P$ a program, and $X$ a
  \strongBds{\CCC}~$X$ of $P$ with $\Card{X}=k$. We can check in
  polynomial time whether $M$ is a model of $P$ and whether $M$ is a
  model of $P^M$. If it is not, we can reject $M$, and we are
  done. Hence assume that $M$ is a model of $P^M$.  In order to check
  whether $M \in \stableset(P)$ we still need to decide whether $M$ is
  a minimal model of $P^M$.  We may assume, w.l.o.g., that $P$
  contains no tautological rules, as it is clear that the test for
  minimality does not depend on tautological rules.
  
  Let $X_1\subseteq M \cap X$. We construct from $P^M$ a program
  $P^M_{X_1\subseteq X}$ by
  (i)~removing all rules~$r$ for which $H(r)\cap X_1\neq \emptyset$,
  and
  (ii)~replacing for all remaining rules~$r$ the head~$H(r)$ with
  $H(r)\setminus X$, and the positive body~$B^+(r)$ with
  $B^+(r)\setminus X_1$.

  \emph{Claim: $P^M_{X_1\subseteq X}$ is Horn.}
  
  To show the claim, consider some rule~$r'\in P^M_{X_1\subseteq X}$.
  By construction, there must be a rule~$r\in P$ that gives raise to a
  rule in $P^M$, which in turn gives raise to $r'$. Let $\tau \in
  \ta{X}$ be the assignment that sets all atoms in $X \cap H(r)$ to~0,
  and all atoms in $X \setminus H(r)$ to~1.  Since $r$ is not
  tautological, it follows that $r$ is not deleted when we obtain
  $P_\tau$, and it gives rise to a rule~$r^*\in P_\tau$, where
  $H(r^*)=H(r)\setminus X$. However, since~$\CCC$ is a class of normal
  programs, $r^*$ is normal. Hence $1 \geq \Card{H(r^*)} = \Card{ H(r)
    \setminus X} = H(r')$, and the claim follows.

  To test whether $M$ is a minimal model of $P^M$, we run the
  following procedure for every set~$X_1\subseteq M \cap X$.

  \begin{quote}
    If $P^M_{X_1\subseteq X}$ has no model, then stop and return TRUE.

    Otherwise, compute the unique minimal model~$L$ of the Horn
    program~$P^M_{X_1\subseteq X}$. If $L \subseteq M \setminus X$, $L
    \cup X_1 \subsetneq M$, and $L \cup X_1$ is a model of $P^M$, then
    return FALSE. Otherwise return TRUE.
  \end{quote}

  For each set~$X_1\subseteq M \cap X$ the above procedure runs in
  linear time by Lemma~\ref{lem:horn-lineartime}. As there are
  $O(2^k)$ sets~$X_1$ to consider, we have a total running time of
  $O(2^k n)$ where $n$ denotes the input size of $P$ and $k=\Card{X}$.
  It remains to establish the correctness of the above procedure in
  terms of the following claim.
  
  \emph{Claim: $M$ is a minimal model of $P^M$ if and only if the
    algorithm returns TRUE for each~$X_1\subseteq M \cap X$.}
 
  ($\Rightarrow$). Assume that $M$ is a minimal model of $P^M$, and
  suppose to the contrary that there is some~$X_1\subseteq M \cap X$
  for which the algorithm returns FALSE. Consequently,
  $P^M_{X_1\subseteq X}$ has a unique minimal model~$L$ with
  $L \subseteq M \setminus X$,
  $L \cup X_1 \subsetneq M$, and where
  $L \cup X_1$ is a model of $P^M$. This contradicts the assumption
  that $M$ is a minimal model of $P^M$. Hence the only-if direction of
  the lemma is shown.
               
  ($\Leftarrow$). Assume that the algorithm returns TRUE for each~$X_1
  \subseteq M \cap X$. We show that $M$ is a minimal model of
  $P^M$. Suppose to the contrary that $P^M$ has a model~$M'\subsetneq
  M$.

  We run the algorithm for $X_1:=M' \cap X$. By assumption, the
  algorithm returns TRUE. There are two possibilities:
  (i)~$P^M_{X_1\subseteq X}$ has no model, or (ii)~$P^M_{X_1 \subseteq
    X}$ has a model, and for its unique minimal model $L$ the
  following holds:
  $L$ is not a subset of $M \setminus X$, or
  $L \cup X_1$ is not a proper subset of $M$, or
  $L \cup X_1$ is not a model of $P^M$.
  
  We show that case~(i) is not possible by showing that $M' \setminus
  X$ is a model of $P^M_{X_1\subseteq X}$.

  To see this, consider a rule~$r'\in P^M_{X_1\subseteq X}$, and let
  $r\in P^M$ such that $r'$ is obtained from $r$ by removing~$X$ from
  $H(r)$ and by removing~$X_1$ from~$B^+(r)$. Since $M'$ is a model of
  $P^M$, we have (a)~$B^+(r) \setminus M' \neq\emptyset$ or (b)~$H(r)
  \cap M'\neq \emptyset$. Moreover, since $B^+(r') = B^+(r) \setminus
  X_1$ and $X_1 = M'\cap X$, (i)~implies $\emptyset \neq B^+(r)
  \setminus M' = B^+(r) \setminus X_1 \setminus M' = B^+(r')\setminus
  M' \subseteq B^+(r') \setminus (M' \setminus X)$, and since $H(r)
  \cap X_1 = \emptyset$, (ii)~implies $\emptyset \neq H(r)\cap M'=
  H(r)\cap (M'\setminus X_1)= H(r)\cap (M'\setminus X)= (H(r)\setminus
  X) \cap (M'\setminus X)= H(r') \cap (M'\setminus X)$.
% Moreover, since $X_1 \subseteq M'$,
% (a)~implies that $B^+(r')\setminus (M' \setminus X_1)
% \neq \emptyset$.
% Since $H(r)\capX_1=\emptyset$ and $X_1=M'\cap
% X$, (b)~implies that 
% $H(r') \setminus (X \cap (M' \setminus X)) \neq \emptyset$ 
% and thus 
% $H(r')\cap (M' \setminus
% X)\neq \emptyset$. 
  Hence $M' \setminus X$ satisfies $r'$. Since $r'\in
  P^M_{X_1\subseteq X}$ was chosen arbitrarily, we conclude that $M'
  \setminus X$ is a model of $P^M_{X_1\subseteq X}$.
 
  Case~(ii) is not possible either, as we can see as follows. Assume
  $P^M_{X_1 \subseteq X}$ has a model, and let $L$ be its unique
  minimal model. Since $M'\setminus X$ is a model of $P^M_{X_1
    \subseteq X}$, as shown above, we have $L\subseteq M'\setminus X$.
 
% $L$ is  a subset of $M_R$ 
  We have $L \subseteq M \setminus X$ since $L\subseteq M' \setminus
  X$ and $M' \setminus X \subseteq M \setminus X$.

% $L \cup X_1$ is  a proper subset of $M$ 
  Further we have $L \cup X_1 \subsetneq M$ since $L \cup X_1
  \subseteq (M' \setminus X) \cup X_1 = (M' \setminus X) \cup (M'\cap
  X) =M' \subsetneq M$.

% $L \cup X_1$ is a model of $P^M$.
  And finally $L \cup X_1$ is a model of $P^M$, as can be seen as
  follows.
  Consider a rule~$r\in P^M$. If $X_1 \cap H(r)\neq \emptyset$, then
  $L\cup X_1$ satisfies $r$; thus it remains to consider the case~$X_1
  \cap H(r) = \emptyset$. In this case there is a rule~$r'\in P^M_{X_1
    \subseteq X}$ with $H(r') = H(r)\setminus X$ and $B^+(r') = B^+(r)
  \setminus X_1$. Since $L$ is a model of $P^M_{X_1 \subseteq X}$, $L$
  satisfies $r'$. Hence (a)~$B^+(r')\setminus L \neq\emptyset $ or
  (b)~$H(r') \cap L \neq \emptyset$. Since $B^+(r')=B^+(r)\setminus
  X_1$, (a)~implies that $B^+(r)\setminus ( L\cup X_1)\neq \emptyset$;
  and since $H(r') \subseteq H(r)$, (b)~implies that $H(r)\cap (L \cup
  X_1) \neq \emptyset$. Thus $L \cup X_1$ satisfies~$r$. Since $r\in
  P^M$ was chosen arbitrarily, we conclude that $L \cup X_1$ is a
  model of~$P^M$.

  Since neither case (i) nor case (ii) is possible, we have a
  contradiction, and we conclude that $M$ is a minimal model of $P^M$.

  Hence the second direction of the claim is established, and so the
  lemma follows.
\end{proof}

\begin{figure}
\centering
	\begin{tikzpicture}[-latex,node distance=3em]
          \node(F)[align=center]{Find $\CCC$-backdoor\\ \(X \subseteq
            \at(P)\)}; \node(P)[below of=F,node distance=7em]{\(P\)};
          \node(RUNTM1)[below of=F,node distance=14em]{?};
          \node(A)[align=center,right of=F,node distance=7em]{Apply\\
            \(\tau_i: X \rightarrow\{0,1\}\) };%\in 2^B\)
          \node(Pt1)[below of=A]{\(P_{\tau_1}\in \CCC\)};
          \node(Pt2)[below of=Pt1]{\(P_{\tau_2}\in \CCC\)};
          \node(Ptn)[below of=Pt2]{\(\cdots\)}; \node(Pt2x)[below
          of=Ptn]{\(P_{\Card{\ta{X}}}\in \CCC\)};
        \path(P) edge node[above] {\(\tau_1\)} (Pt1);
        \path(P) edge node[above] {\(\tau_2\)} (Pt2);
        \path(P) edge node[above] {\ldots} (Ptn);
        \path(P) edge node[below,xshift=-8pt,yshift=-3pt]
        {\(\tau_{\Card{\ta{X}}}\)} (Pt2x); \node(RUNTM2)[below
        of=A,node distance=14em]{\(\mathcal{O}(\Card{\ta{X}} \cdot
          n)\)}; \node(C)[align=center,right of=A,node
        distance=9em]{Determine answer sets\\ of simplified programs};
        \node(ASPt1)[below of=C]{\(\stableset(P_{\tau_1})\)};
        \node(ASPt2)[below of=ASPt1]{\(\stableset(P_{\tau_2})\)};
        \node(ASPtn)[below of=ASPt2]{\(\cdots\)}; \node(ASPt2x)[below
        of=ASPtn]{\(\stableset(P_{\tau_{\Card{\ta{X}}}})\)};
        \path(Pt1) edge node[above] {} (ASPt1);
        \path(Pt2) edge node[above] {} (ASPt2);
        \path(Ptn) edge node[above] {} (ASPtn);
        \path(Pt2x) edge node[below] {} (ASPt2x); \node(RUNTM3)[below
        of=C,node distance=14em]{\(\mathcal{O}(\Card{\ta{X}}\cdot
          n^c)\)}; \node(Cand)[align=center,right of=C,node
        distance=9em]{Check\\ candidates}; \node(AS)[below of=Cand,
        node distance=7em]{$\stableset(P,X)$};
        \path(ASPt1) edge node[above right] {\(\cup\,
          \tau^{-1}_1(1)\)} (AS);
        \path(ASPt2) edge node[above] {\ldots} (AS);
        \path(ASPtn) edge node[above] {\ldots} (AS);
        \path(ASPt2x) edge node[below right] {\(\cup\,
          \tau^{-1}_{\Card{\ta{X}}}(1)\)} (AS); \node(RUNTM4)[below
        of=Cand,node
        distance=14em,xshift=6mm]{\(\mathcal{O}(\Card{\ta{X}}^2 \cdot
          n^c )\)}; \node(SOL)[right of=Cand,node
        distance=6em]{Solutions}; \node(ASP)[below of=SOL,node
        distance=7em]{$\stableset(P)$}; \path(AS) edge (ASP);
        \node(RUNTM5)[below of=SOL,node distance=14em]{};
	\end{tikzpicture}
	\caption{Exploit pattern of ASP backdoors if the target class \(\mathcal{C}\) is normal and enumerable where $n$ denotes the input size of $P$.}
	\label{fig:exploit-backdoors}
	% \vspace{-0.6cm}
\end{figure}
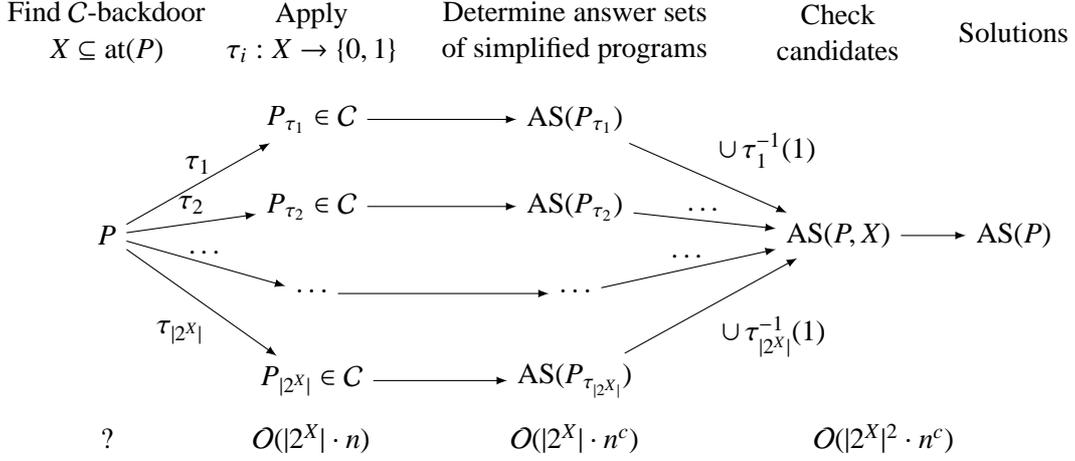

Figure~\ref{fig:exploit-backdoors} illustrates how we can exploit a
\strongBds{\CCC} to find answer sets. For a given program~$P$ and a
\strongBds{\CCC}~$X$ of $P$ we have to consider $\Card{\ta{X}}$ truth
assignments to the atoms in the backdoor~$X$.  For each truth
assignment~$\tau \in \ta{X}$ we reduce the program $P$ to a
program~$P_\tau$ and compute the set~$\stableset(P_\tau)$. Finally, we
obtain the set~$\stableset(P)$ by checking for each~$M\in
\stableset(P_\tau)$ whether it gives rise to an answer set of $P$.

\begin{EX}
  We consider the set~$\stableset(P,X)=\{ \{e,f\}, \{c,f\}, \{b,e,f\},
  \{b,c,f\}\}$ of answer set candidates of Example~\ref{ex:expoit_bds}
  and check for each candidate~$L=\{e,f\}$, $M = \{c,f\}$,
  $N=\{b,e,f\}$, and $O=\{b,c,f\}$ whether it is an answer set
  of~$P$. Therefore we solve the problem \BdCheck{\Horn} by means of
  Lemma~\ref{lem:mincheck}. 

  First we test whether the sets~$L$, $M$, $N$ and $O$ are models of
  $P$. We easily observe that $N$ and $O$ are models of $P$. But $L$
  and $M$ are not models of $P$ since they do not satisfy the rule~$c
  \leftarrow e, f, \pnot b$ and $b \leftarrow c$ respectively, and we
  can drop them as candidates. Then we positively answer the question
  whether $N$ and $O$ are models of its GL-reducts $P^N$ and $P^O$
  respectively. 

  Next we consider the minimality and apply the algorithm of
  Lemma~\ref{lem:mincheck} for each subset of the
  backdoor~$X=\{b,c\}$. 
  We have the GL-reduct $P^N=\{d \leftarrow a, e \rsep e \por c
  \leftarrow f \rsep f \leftarrow d, c \rsep c \leftarrow d \rsep b
  \leftarrow c\rsep f\}$.
  For $X_1 = \emptyset$ we obtain $P^N_{X_1\subseteq X}=\{d \leftarrow
  a, e \rsep e \leftarrow f \rsep f \leftarrow d, c \rsep \leftarrow d
  \rsep \leftarrow c \rsep f\}$. The set $L=\{e,f\}$ is the unique
  minimal model of $P^N_{X_1 \subseteq X}$. Since $L \subseteq N
  \setminus X$, $L \cup X_1 \subsetneq N$, and $L \cup X_1$ is a model
  of $P^N$, the algorithm returns FALSE. We conclude that $N$ is not a
  minimal model of $P^N$ and thus $N$ is not an answer set of $P$. 

  We obtain the GL-reduct $P^O=\{d \leftarrow a, e \rsep e \por c
  \leftarrow f \rsep f \leftarrow d, c \rsep c \leftarrow d \rsep b
  \leftarrow c\rsep f\}$.
  For~$X_1 =\emptyset$ we have $P_{X_1\subseteq X}=\{d \leftarrow a ,
  e \rsep e \leftarrow f \rsep f \leftarrow d, e \rsep \leftarrow d
  \rsep \leftarrow c \rsep f \}$. The set~$L=\{e,f\}$ is the unique
  minimal model of $P_{X_1 \subseteq X}$. Since $L\cup X_1 \subsetneq
  O$, the algorithm returns TRUE.  
  For~$ X_2=\{b\}$ we get $P_{X_2\subseteq X}=\{d \leftarrow a, e
  \rsep e \leftarrow f \rsep f \leftarrow d, e \rsep \leftarrow d
  \rsep f \}$ and the unique minimal model $L=\{e,f\}$. Since $L
  \subseteq O \setminus X$, the algorithm returns TRUE. 
  For $X_3=\{c\}$ we obtain $P_{X_3\subseteq X}=\{d \leftarrow a, e
  \rsep f \leftarrow d \rsep \leftarrow \rsep f\}$ and no minimal
  model. Thus the algorithm returns TRUE.  
  For $X_4=\{b,c\}$ we have $P_{X_4\subseteq X}=\{d \leftarrow a, e
  \rsep f \leftarrow d \rsep f \}$ and the unique minimal model
  $L=\{f\}$. Since $L \cup X_1 \subsetneq M$, the algorithm returns
  TRUE. Since only $\{b,c,f\} \in \stableset(P,X)$ is an answer set of
  $P$, we obtain $\stableset(P)=\{\{b,c,f\}\}$.
\end{EX}

In view of Lemmas~\ref{lem:subset} and \ref{lem:mincheck}, the
computation of $\stableset(P)$ is fixed-parameter tractable for
parameter~$k$ if we know a \strongBds{\CCC}~$X$ of size at most~$k$
for $P$, and each program in $\CCC$ is normal and its stable sets can
be computed in polynomial time. This consideration leads to the
following definition and result.
\begin{DEF}
  A class~$\CCC$ of programs is \emph{enumerable} if for each~$P\in
  \CCC$ we can compute~$\stableset(P)$ in polynomial time. If
  $\stableset(P)$ can be computed even in linear time, then we call
  the class \emph{linear-time enumerable}.
\end{DEF}

\begin{THE}
\label{the:evaluation}
Let $\CCC$ be an enumerable class of normal programs. The problems in
$\AspFull$ are all fixed-parameter tractable when parameterized by the
size of a \strongBds{\CCC}, assuming that the backdoor is given as an
input.
\end{THE}
\begin{proof}
  Let $X$ be the given backdoor, $k=\Card{X}$ and $n$ the input size
  of $P$. Since $P_\tau \in \CCC$ and $\CCC$ is enumerable, we can
  compute $\stableset(P_\tau)$ in polynomial time for each~$\tau \in
  \ta{X}$, say in time~$O(n^c)$ for some constant~$c$. Observe that
  therefore $\Card{\stableset(P_\tau)}\leq O(n^c)$ for each~$\tau \in
  \ta{X}$. Thus we obtain $\stableset(P,X)$ in time~$O(2^k n^c)$, and
  $\Card{\stableset(P,X)}\leq O(2^k n^c)$. By Lemma~\ref{lem:subset},
  $\stableset(P)\subseteq \stableset(P,X)$. By means of
  Lemma~\ref{lem:mincheck} we can decide whether $M \in AS(P)$ in
  time~$O(2^k n)$ for each~$M\in \stableset(P,X)$. Thus we determine
  from~$\stableset(P,X)$ the set of all answer sets of $P$ in
  time~$O(2^k\cdot n^c \cdot 2^k \cdot n + 2^k \cdot n^c)=O(2^{2k}
  n^{c+1})$. Once we know $\stableset(P)$, then we can also solve all
  problems in $\AspFull$ within polynomial time.
\end{proof}

\begin{REM}
  If we know that each program in $\CCC$ has at most one answer set,
  and $P$ has a \strongBds{\CCC} of size~$k$, then we can conclude
  that $P$ has at most~$2^k$ answer sets. Thus, we obtain an upper
  bound on the number of answer sets of $P$ by computing a small
  \strongBds{\CCC} of $P$.
\end{REM}

The following definition will be useful in the sequel.
\begin{DEF}\label{def:taut}
  Let~$\CCC$ be a class of programs. We denote by $\tautext{\CCC}$ the
  class containing all programs that belong to $\CCC$ after removal of
  tautological rules and constraints.
\end{DEF}
In fact, it turns out that for several of our algorithmic results that
work for $\CCC$\hy backdoors also work for $\tautext{\CCC}$-backdoors,
but the latter can be much smaller than the former. Hence we will
often formulate and establish results in terms of the more general notion
$\tautext{\CCC}$.  
\begin{OBS}
  Whenever a class~$\CCC$ of programs is (linear-time) enumerable,
  then so is $\tautext{\CCC}$.
\end{OBS}
\begin{proof}
  Let $\CCC$ be enumerable, let $P^*\in \tautext{C}$, and let $P$ be
  the program obtained from~$P^*$ by removing tautological rules and
  constraints. Since $\CCC$ is enumerable, we can compute
  $\stableset(P)$ in polynomial time (or linear time, if $\CCC$ is
  linear-time enumerable).  By well-known results
  \cite{BrassDix98,BuccafurriLeoneRullo97} $\stableset(P) \subseteq
  \stableset(P^*)$, and in order to check whether some~$M \in
  \stableset(P)$ belongs to $\stableset(P^*)$ we only need to check
  whether $M$ satisfies all the constraints of $P^*$, which can be
  done in linear time.
\end{proof}

\subsection{Backdoor Detection}
Theorem~\ref{the:evaluation} draws our attention to enumerable classes
of normal programs. Given such a class~$\CCC$, is the detection of
$\CCC$\hy backdoors fixed-parameter tractable? If the answer is
affirmative, we can drop in Theorem~\ref{the:evaluation} the
assumption that the backdoor is given as an input for this class.

% We note that tautological rules can be omitted from a program
% without changing its answer sets~\cite{BrassDix98}.  Furthermore,
% removing constraints has no influence on the running time since we
% can remove constraints previous to backdoor evaluation, add the
% constraints after determining candidates, and check the constraints
% in linear time obviously.

Each class~$\CCC$ of programs gives rise to the following two
parameterized decision problems:

\pproblem{\strongBdsDet{\CCC}}{a program~$P$ and an integer~$k$.}{the
  integer~$k$.}{Decide whether $P$ has a \strongBds{\CCC}~$X$ of size
  at most~$k$.}

\pproblem{\delBdsDet{\CCC}}{a program~$P$ and an integer~$k$.}{the
  integer~$k$.}{Decide whether $P$ has a \delBds{\CCC}~$X$ of size at
  most~$k$.}

By a standard construction, known as self-reduction or
self-transformation~\cite{Schnorr81,DowneyFellows99}, one can use a
decision algorithm for \delBdsDet{\CCC} to actually find the
backdoor. We only require the base class to be hereditary.

\begin{LEM}
  Let $\CCC$ be a hereditary class of programs.  If \delBdsDet{\CCC}
  is fixed-parameter tractable, then also finding a \delBds{\CCC} of a
  given program~$P$ of size at most~$k$ is fixed-parameter tractable
  (for parameter~$k$).
%   If \strongBdsDet{\CCC} is fixed-parameter tractable, then also
%   finding a \strongBds{\CCC} of a given program $P$ of size at most
%   $k$ is fixed-parameter tractable (for parameter $k$).
\end{LEM}
\begin{proof}
  We proceed by induction on $k$. If $k=0$ the statement is clearly
  true. Let $k>0$.  Given $(P,k)$ we check for all~$x\in \at(P)$
  whether $P-\{x\}$ has a \delBds{\CCC} of size at most~$k-1$.  If the
  answer is NO for all $x$, then $P$ has no \delBds{\CCC} of
  size~$k$. If the answer is YES for $x$, then by induction hypothesis
  we can compute a \delBds{\CCC} $X$ of size at most~$k-1$ of $P-x$,
  and $X\cup \{x\}$ is a \delBds{\CCC} of $P$.
\end{proof}

% \begin{DEF}
%   $k$-\textsc{Strong $\CCC\hy$Back\-door Detection}: given a program
%   $P$, find a strong $\CCC$\hy backdoor $X$ of $P$ of size at most
%   $k$, or report that such $X$ does not exist.
% \end{DEF}
% \noindent
% We also consider the similar problem for deletion $\CCC$-Backdoors
% which is in some cases easier to solve than the strong variant.
% \pproblem{\delBdsDet{\CCC}}{a program~$P$ and an integer~$k$.}{the
%   integer~$k$.}{find a \delBds{\CCC}~$X$ of $P$ of size at most~$k$,
%   or report that such $X$ does not exist.}

% \begin{DEF}
%   $k$-\textsc{Deletion $\CCC$-Backdoor Detection}: given a program
%   $P$, find a deletion $\CCC$\hy backdoor $X$ of $P$ of size at most
%   $k$, or report that such $X$ does not exist.
% \end{DEF}

\section{Target Class Horn}\label{sec:horn} 
In this section we consider the important case $\Horn$ as the target
class for backdoors.  As a consequence of
Lemma~\ref{lem:horn-lineartime}, $\Horn$ is linear-time enumerable.
The following lemma shows that strong and deletion
$\tautext{\Horn}$-backdoors coincide.

\begin{LEM}\label{lem:Horn-strong-deletion}
  A set~$X$ is a \strongBds{\tautext{\Horn}} of a program~$P$ if and
  only it is a \delBds{\tautext{\Horn}} of~$P$.
\end{LEM}
\begin{proof}
  Since $\tautext{\Horn}$ is hereditary, Lemma~\ref{lem:rule-induced}
  establishes the if-direction.  For the only-if direction, we assume
  for the sake of a contradiction that $X$ is a
  \strongBds{\tautext{\Horn}} of $P$ but not a
  \delBds{\tautext{\Horn}} of $P$. Hence there is a rule~$r'\in P-X$
  which is neither tautological nor a constraint nor Horn. Let $r\in
  P$ be a rule from which $r'$ was obtained in forming $P-X$.  We
  define $\tau\in \ta{X}$ by setting all atoms in $X \cap (H(r)\cup
  B^-(r))$ to 0, all atoms in $X\cap B^+(r)$ to 1, and all remaining
  atoms in $X\setminus \at(r)$ arbitrarily to 0 or 1.  Since $r$ is
  not tautological, this definition of $\tau$ is sound.  It follows
  that $r'\in P_\tau$, contradicting the assumption that $X$ is a
  \strongBds{\tautext{\Horn}} of $P$.
\end{proof}
 
\begin{DEF}
  Let $P$ be a program. The \emph{negation dependency graph}~$N_P$ is
  the graph defined on the set of atoms of the given program~$P$,
  where two atoms~$x,y$ are joined by an edge~$xy$ if there is a
  rule~$r\in P$ with $x\in H(r)$ and $y\in H(r)\cup B^-(r)$.
\end{DEF}

Tautological rules and constraints do not produce any edges in the
negation dependency graph, hence, if we delete such rules from the
program, we still obtain the same graph.

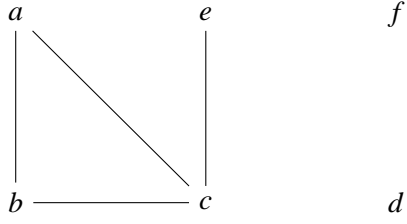
\begin{figure}
  \centering
  \begin{tikzpicture}[-,node distance=25mm]
    \node[](a){$a$}; 
    \node[below of=a](b){$b$}; 
    \node[right of=b](c){$c$}; 
    \node[right of=a](e){$e$}; 
    \path(a) edge (b);
    \path(a) edge (c);
    \path(c) edge (b);
    \path(e) edge (c);
    \node[right of=c](d){$d$}; 
    \node[right of=e](f){$f$}; 
  \end{tikzpicture}
 \caption{Negation dependency graph~$N_P$ of the program~$P$ of
   Example~\ref{ex:running}.}
  \label{fig:Np} 
\end{figure}

\begin{EX} 
  Figure~\ref{fig:Np} visualizes the negation dependency graph~$N_P$
  of the program~$P$ of Example~\ref{ex:running}.
\end{EX}

The following lemma states how we can use recent results on the vertex
cover problem to find deletion backdoors for the target class~\Horn. A
\emph{vertex cover} of a graph~$G=(V,E)$ is a set~$S \subseteq V$ such
that for every edge~$uv\in E$ we have $\{u,v\} \cap S \neq \emptyset$.

\begin{LEM}\label{lem:hornbdsdet-vc}
  Let $P$ be a program.  A set~$X\subseteq \at(P)$ is a \delBds{\Horn}
  of $P$ if and only if $X$ is a vertex cover of the negation
  dependency graph~$N_P$.
\end{LEM}

\begin{proof}
  Let $X\subseteq H(r)\cup B^-(r)$ be a \delBds{\Horn} of $P$.
  Consider an edge~$uv$ of $N_P$. By construction of $N_P$ there is a
  corresponding rule~$r \in P$ with (i)~$u,v\in H(r)$ and $u\neq v$ or
  (ii)~$u\in H(r)$ and $v\in B^-(r)$. Since $X$ is a \delBds{\Horn},
  $\Card{H(r)-X}\leq 1$ and $B^-(r) - X = \emptyset$. Thus if Case~(i)
  applies, $\{u,v\}\cap X \neq \emptyset$. If Case~(ii) applies, again
  $\{u,v\}\cap X \neq \emptyset$. We conclude that $X$ is a vertex
  cover of $N_P$.

  Conversely, assume that $X$ is a vertex cover of $N_P$. Consider a
  rule~$r\in P-X$ for proof by contradiction. If $\Card{H(r)}\geq 2$
  then there are two variables~$u,v\in H(r)$ and an edge~$uv$ of $N_P$
  such that $\{u,v\}\cap X =\emptyset$, contradicting the assumption
  that $X$ is a vertex cover. Similarly, if $\Card{B^-(r)}\geq 1$ then
  we take a variable~$u \in B^-(r)$ and a variable~$v\in H(r)$; such
  $v$ exists since $r$ is not a constraint. Thus $N_P$ contains the
  edge $uv$ with $\{u,v\}\cap X\neq \emptyset$, contradicting the
  assumption that $X$ is a vertex cover. Hence the claim holds.
\end{proof}

\begin{EX}
  For instance, the negation dependency graph~$N_P$ of the program~$P$
  of Example~\ref{ex:running} consists of the triangle~$\{a,b,c\}$ and
  a path $(c,e)$. Then $\{b,c\}$ is a vertex cover of $G$. We observe
  easily that there exists no vertex cover of size $1$. Thus $\{b,c\}$
  is a smallest \strongBds{\tautext{\Horn}} of $P$.
\end{EX}

% \begin{LEM}\label{lem:strong-horn-bds-npc}
%   The problem \strongDelBdsDet{\tautext{\Horn}} is \NP-complete.
% \end{LEM}
% \begin{proof}
%   It is easy to see that both problems are in \NP, since we can simply
%   guess for a given program~$P$ a set of atoms and check in polynomial
%   time whether it is a \strongBds{\tautext{\Horn}} (\delBds{\tautext{\Horn}}) of
%   $P$. 

%   We obtain the hardness for the problem \delBdsDet{\tautext{\Horn}}
%   by a polynomial-time reduction from the \NP-complete problem
%   \VC. Let $G=(V,E)$ be a graph. We construct a program~$P_G$ as
%   follows: we take the vertices in $V$ as the atoms of $P_G$ and for
%   each edge $uv \in E$ we add a new rule $r_{uv}$ where
%   $H(r_{uv})=\{u\}$ and $B^-(r_{uv})=\{v\}$ (it is insignificant
%   which atom belongs to the head and which atom belongs to the
%   negative body). Now a set~$C$ is a vertex cover of the graph~$G$
%   if and only if $C$ is a \delBds{\tautext{\Horn}} of $P_G$. By
%   Lemma~\ref{lem:Horn-strong-deletion} a set~$C$ is a
%   \strongBds{\tautext{\Horn}} of a program~$P$ if and only if $C$ is
%   a \delBds{\tautext{\Horn}} of $P$. Consequently both problems are
%   \NP-complete.

% \end{proof}

\begin{THE}\label{the:horn}
  \textsc{Strong $\tautext{\Horn}$\hy Backdoor Detection} is
  fixed-parameter tractable. In fact, given a program with $n$~atoms
  we can find a \strongBds{\tautext{\Horn}} of size at most~$k$ in
  time~$O(1.2738^k + kn)$ or decide that no such backdoor exists.
\end{THE}
\begin{proof}
  Let $P^*$ be a given program. We delete from $P^*$ all tautological
  rules and all constraints and obtain a program~$P$ with $n$ atoms.
  We observe that the \strongBds{\tautext{\Horn}}s of $P^*$ are
  precisely the \strongBds{\Horn}s of $P$. Let $N_P$ be the negation
  dependency graph of $P$. According to Lemma~\ref{lem:hornbdsdet-vc}
  a set~$X\subseteq \at(P)$ is a vertex cover of $N_P$ if and only if
  $X$ is a \delBds{\tautext{\Horn}} of $P$. Then a vertex cover of
  size at most~$k$, if it exists, can be found in time~$O(1.2738^k +
  kn)$ by~\citex{ChenKanjXia10}. By
  Lemma~\ref{lem:Horn-strong-deletion} this vertex cover is also a
  \strongBds{\tautext{\Horn}} of~$P$.
\end{proof}

Now we can use Theorem~\ref{the:horn} to strengthen the
fixed-parameter tractability result of Theorem~\ref{the:evaluation} by
dropping the assumption that the backdoor is given.
\begin{COR}\label{cor:horn_bds}
  All the problems in $\AspFull$ are fixed-parameter tractable when
  parameterized by the size of a smallest \strongBds{\tautext{\Horn}}
  of the given program.
\end{COR}

\section{Target Classes Based on Acyclicity}\label{sec:acyclic-classes}
There are two causes for a program to have a large number of answer
sets: (i)~disjunctions in the heads of rules, and (ii)~certain cyclic
dependencies between rules.  Disallowing both yields enumerable
classes.

%In the following we will study backdoor detection for various classes
%of stratified programs. We define the classes by requiring normality
%and acyclicity (the absence of certain types cycles). 
In order to define acyclicity we associate with each disjunctive
program~$P$ its \emph{dependency digraph}~$D_P$ and its
\emph{(undirected) dependency graph}~$U_P$.  These definitions extend
similar notions defined for normal programs by
\citex{AptBlairWalker88} and \citex{GottlobScarcelloSideri02}.

\begin{DEF}
  Let $P$ be a program. The \emph{dependency digraph} is the
  digraph~$D_P$ which has as vertices the atoms of $P$ and a directed
  edge~$(x,y)$ between any two atoms~$x$,~$y$ for which there is a
  rule~$r\in P$ with $x\in H(r)$ and $y\in B^+(r)\cup B^-(r)$.  We
  call the edge~$(x,y)$ \emph{negative} if there is a rule~$r\in P$
  with $x\in H(r)$ and $y\in B^-(r)$.
\end{DEF}

\begin{DEF}
  Let $P$ be a program. The \emph{(undirected) dependency graph}~is
  the graph~$U_P$ obtained from the dependency digraph~$D_p$
\begin{enumerate}
\item by replacing each negative edge~$e=(x,y)$ with two edges~$xv_e$,
  $v_ey$ where $v_e$ is a new \emph{negative vertex}, and
\item by replacing each remaining directed edge~$(u,v)$ with an
  edge~$uv$.
\end{enumerate}
\end{DEF}

\begin{EX}\label{ex:graphs}
  Figure~\ref{fig:DpUp} visualizes the dependency digraph~$D_P$ and
  the dependency graph $U_p$ of the program $P$ of
  Example~\ref{ex:running}.
\end{EX}

\begin{figure}[htb]
  \centering 
\subfloat{
  \begin{tikzpicture}[-latex,node distance=25mm]
    \node[](c){$c$};
    \node[left of=c](b){$b$};
    \node[above left of=c](a){$a$};
    \node[right of=c](f){$f$};
    \node[above right of=c](e){$e$};
    \node[above right of=a](d){$d$};
    \path(d) edge[bend right] (a);
    \path(a) edge[bend right] (d);
    \path(d) edge (e);
    \path(e) edge (f);
    \path(f) edge[bend right=50] (d);
    \path(e) edge[bend right] node[transparent,label=above:$\pnot$]{} (c);
    \path(c) edge[] node[transparent,label=right:$\pnot$]{} (e);
    \path(c) edge[bend right] (f);
    \path(f) edge[bend right] node[transparent,label=below:$\pnot$]{} (c);
    \path(c) edge[bend right] node[transparent,label=below:$\pnot$]{} (b);
    \path(b) edge[bend right] (c);
    \path(a) edge[] node[transparent,label=above left:$\pnot$]{} (b);
    \path(a) edge[] node[transparent,label=above:$\pnot$]{} (c);
    \path(c) edge (d);
  \end{tikzpicture}
 \label{fig:Dp}
}
 \quad \subfloat{
\begin{tikzpicture}[-,node distance=15mm]
    \node[](c){$c$};
    \node[left of=c,yshift=1em](vcb){$v_{c,b}$};
    \node[left of=vcb](b){$b$};
    \node[above right of=b] (vab){$v_{a,b}$};
    \node[right of=vab,xshift=-0.5em] (vac){$v_{a,c}$};
    \node[above right of=vab](a){$a$};
    \node[right of=c](vfc){$v_{f,c}$};
    \node[right of=vfc](f){$f$};
    \node[transparent,above right of=c](e1){};
    \node[above right of=e1](e){$e$};
    \node[right of=e1,xshift=-2em](vec){$v_{e,c}$};
    \node[left of=e1,xshift=3em](vce){$v_{c,e}$};
    \node[above right of=a](d){$d$};
    \path(a) edge[] (d);
    \path(d) edge (e);
    \path(e) edge (f);
    \path(f) edge[bend right=50] (d);
    \path(e) edge[] (vce);
    \path(vce) edge[] (c);
    \path(c) edge[] (vec);
    \path(vec) edge[] (e);
    \path(c) edge[bend right] (f);
    \path(f) edge[] (vfc);
    \path(vfc) edge[] (c);
    \path(c) edge (vcb);
    \path(vcb) edge (b);
    \path(b) edge[bend right] (c);
    \path(a) edge[] (vab);
    \path(vab) edge[] (b);
    \path(a) edge (vac);
    \path(vac) edge (c);
    \path(c) edge (d);
  \end{tikzpicture}
}
\caption{Dependency digraph~$D_P$ (left) and dependency graph~$U_P$
  (right) of the program~$P$ of Example~\ref{ex:running}.}
	\label{fig:DpUp}
\end{figure}
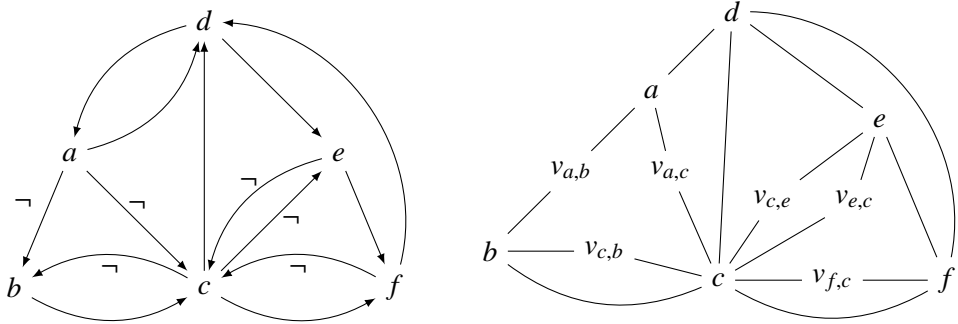

\begin{DEF}
  Let $P$ be a program.
\begin{enumerate}
\item A \emph{directed cycle of $P$} is a directed cycle in the
  dependency digraph~$D_P$.
\item  A directed cycle is \emph{bad} if it contains a negative edge,
  otherwise it is \emph{good}.
\item A directed cycle is \emph{even} if it contains an even number
  of negative edges, otherwise it is \emph{odd}.
\item A \emph{cycle of $P$} is a cycle in the dependency graph~$U_P$.
\item A cycle is \emph{bad} if it contains a negative vertex,
  otherwise it is \emph{good}.
\item A cycle is \emph{even} if it contains an even number
  of negative vertices, otherwise it is \emph{odd}.
\end{enumerate}
\end{DEF}

\begin{DEF}
  The following classes of programs are defined in terms of the
  absence of various kinds of cycles:
\begin{itemize}
\item \C contains all programs that have no cycles,
\item \BC contains all programs that have no bad cycles,
\item \DC contains all programs that have no directed cycles,
\item \DCTWO contains all programs that have no directed cycles of
  length at least~3 and no directed bad cycles
  % no directed cycles (except directed cycles of length~2), that have
  % no directed cycles of length at least~3 and no directed bad cycles
  % of length~2
\item \DBC contains all programs that have no directed bad cycles,
\item \EC contains all programs that have no even cycles,
\item \BEC contains all programs that have bad even cycles,
\item \DEC contains all programs that have no directed even
  cycles, and
\item \DBEC contains all programs that have no directed bad even
  cycles.
\end{itemize}
We let $\Acyc$ denote the family of all the eight classes defined
above. We also write $\DAcyc$ to denote the subfamily~$\{\DC$,
$\DCTWO$, $\DBC$, $\DEC$, $\DBEC\} \subseteq \Acyc$.
\end{DEF}

\begin{EX}
  Consider the dependency graphs of the program~$P$ of
  Example~\ref{ex:running} as depicted in Figure~\ref{fig:DpUp}.  For
  instance the sequence~$(d,e,f)$ is a cycle, $(d,a)$ is a directed
  cycle (of length~2), $(d,e,f)$ and $(c,e,f)$ are directed cycles (of
  length~3), $(a,v_{(a,c)},c,d)$ is a bad cycle, $(c,f)$ is a directed
  bad cycle. The sequence~$(d,e,f)$ is an even cycle and an even
  directed cycle, $(c,e)$ is an directed bad even cycle.
 
  The set~$X=\{c\}$ is a \strongBds{\DBEC} since the truth assignment
  reducts~$P_{c=0}=P_0=\{ d \leftarrow \rsep a \leftarrow \pnot b
  \rsep e \leftarrow f \rsep f\}$ and $P_{1}=\{ d \leftarrow a, e
  \rsep f \leftarrow d \rsep b \rsep f\}$ are in the target class
  \DBEC. $X$ is also a \strongBds{\BEC}, since $P_0\in \BEC$ and $P_1
  \in \BEC$. The answer sets of $P_\tau$ are $\stableset(P_{\bar c}) =
  \{\{e,f\}\}$ and $\stableset(P_{c}) = \{\{b,f\}\}$. Thus
  $\stableset(P,X) = \{ \{e,f\},\{b,c,f\}\}$, and since only
  $\{b,c,f\}$ is an answer set of $P$, we obtain
  $\stableset(P)=\{\{b,c,f\}\}$.
\end{EX}

The dependency and dependency digraphs contain cycles through head
atoms for non-singleton heads. This has the following consequence.

\begin{OBS} \label{obs:normal} $\CCC \subseteq \Normal$ holds for
  all~$\CCC\in \Acyc$.
\end{OBS}

If we have two programs $P\subseteq P'$, then clearly the dependency
(di)graph of $P$ is a sub(di)graph of the dependency (di)graph of
$P'$. This has the following consequence.

\begin{OBS}\label{obs:acyclic_rule_induced}
  All~$\CCC \in \Acyc$ are hereditary, and so is $\tautext{\CCC}$.
\end{OBS}

% \begin{figure}
%   \vspace{-3em}
%   \centering
%   \begin{tikzpicture}[-latex,node distance=25mm,font=\small]
% 	\node[color=black](dbec){$\DBEC$};
% 	\node[left of=dbec] (dbc) {$\DBC$};
% 	\node[left of=dbc] (bc) {$\BC$};
% 	\node[left of=bc] (c) {$\C$};
% 	\node[color=black,above left of=dbec] (dec) {$\DEC$};
%         \node[left of=dec] (dc2){$\DCTWO$};
% 	\node[left of=dc2] (dc) {$\DC$};
% 	\node[color=black,below left of=dbec] (ebc) {$\BEC$};
% 	\node[color=black,left of=ebc] (ec) {$\EC$};
%         \node[left of=ec](horn){$\Horn$};
% 	\path(c) edge node[] (CEC) {} (ec);
% 	\path(c) edge (bc);
% 	\path(bc) edge node[](BC2DBC) {}  (dbc);
% 	\path(dbc) edge node[right] (DBC2DBEC) {} (dbec);
%         \path(dc) edge (dc2);
%         \path(c) edge (dc2);
% 	\path(dc2) edge node[] (DC2DBC) {} (dbc);
% 	\path(dec) edge (dbec);
% 	\path(ebc) edge (dbec);
%         \path(horn) edge (bc);
% 	\path(bc) edge (ebc);	
% 	\path(dc2) edge node[above left=0.75cm] (DC2DEC) {} (dec);
% 	\path(ec) edge node[below =1cm] (EC2EBC) {} (ebc);
% 	\node[right of=ebc, transparent] (ebc2) {};
% 	\path(ebc) edge[transparent] node[below =1cm,xshift=-1cm] (HELPER) {}
%         (ebc2);
% \end{tikzpicture}%
% \vspace{-2em}
% \caption{Relationship between classes of programs with respect to
%   their generality. \label{fig:lattice}}
% \end{figure}

\begin{figure}
\centering
\begin{tikzpicture}[latex-,level 1/.style={sibling distance=8em},level distance=5em,font=\small]
  \node[] (dbec) {$\DBEC$}
  child { node (dec){$\DEC$}
    child{ node (dctwo) {$\DCTWO$} edge from parent[draw=none]
      child{ node (dc) {$\DC$}}}
  }
  child { node (dbc){$\DBC$}
    child{ node (bc){$\BC$}
      child{ node(c){$\C$}}
      }
  }
  child { node (bec){$\BEC$}
    child{ node (ec){$\EC$}
      child{ node(horn){$\Horn$} edge from parent[transparent]}
    }
  };
  \path(dctwo) edge (c);
  \path(bc) edge (horn);
  \path(ec) edge (c);
  \path(bec) edge (bc);
  \path(dbc) edge (dctwo);

\end{tikzpicture}%
\caption{Relationship between classes of programs with respect to
  their generality. \label{fig:lattice}}
\label{fig:diagram-strongBds}
\end{figure}
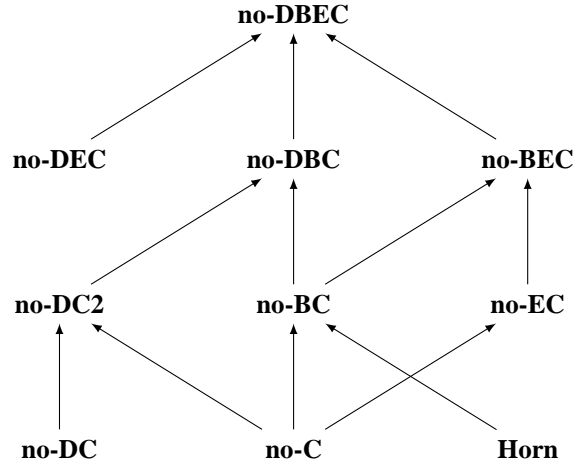

\begin{samepage}
The following is a direct consequence of the definitions of the
various classes in $\Acyc$.
\begin{OBS}\label{obs:lattice} 
  Let $\CCC,\CCC' \in \Acyc \cup \{\Horn\}$ such that the digraph in
  Figure~\ref{fig:lattice} contains a directed path from the
  class~$\CCC$ to the class~$\CCC'$, then $\CCC\subseteq \CCC'$. If no
  inclusion between two classes is indicated, then the classes are in
  fact incomparable.
\end{OBS}
\end{samepage}
\begin{proof}
  We first consider the acyclicity-based target classes. By definition
  we have $\DC\subsetneq \DBC$ and $\C\subsetneq \BC \subsetneq \DBC$;
  it is easy to see that the inclusions are proper.  However, contrary
  to what one expects, $\C\not\subseteq \DC$, which can be seen by
  considering the program~$P_1=\{x\leftarrow y,\ y\leftarrow x\}$.
  But the class \DCTWO which requires that a program has no directed
  cycles but may have directed good cycles of length~2 (as in $P_1$)
  generalizes both classes~$\C$ and $\DC$. By definition we have $\DBC
  \subsetneq \DBEC$, $\DEC \subsetneq \DBEC$, $\EC \subsetneq \BEC$,
  $\C \subsetneq \EC$, and $\DC \subsetneq \DEC$.

  Next we consider the target class~\Horn. Let $\CCC \in
  \{\C,\DC,\EC\}$.  We easily observe that $\Horn \not \subseteq \CCC$
  by considering the program~$P_2 = \{ a \leftarrow b \rsep b
  \leftarrow c \rsep c \leftarrow a \}$ which is obviously Horn but
  does not belong to $\CCC$. Conversely, we observe that $\CCC \not
  \subseteq \Horn$ by considering the program~$P_3=\{a \leftarrow
  \pnot b\}$ which belongs to $\CCC$ but is obviously not Horn. Thus
  $\CCC$ and $\Horn$ are incomparable. We observe that $\Horn
  \subsetneq \BC$ by again considering the program~$P_3$ which belongs
  to $\BC$, but is obviously not Horn, and by considering the fact
  that all rules~$r$ in a Horn program~$P$ satisfy $\Card{H(r)}\leq 1$
  and $B^-(r)=\emptyset$ which yields that the dependency graph~$U_P$
  contains no bad vertices and hence gives us that $U_P$ contains no
  bad cycles.

\end{proof}

The class $\DBC$ coincides with the well-known class of
\emph{stratified}
programs~\cite{AptBlairWalker88,Gelder89a,ChandraHarel85}. A normal
program~$P$ is \emph{stratified} if there is a mapping~$\str: \at(P)
\rightarrow \NAT$, called \emph{stratification}, such that for each
rule~$r$ in $P$ the following holds: (i)~if $x\in H(r)$ and $y \in
B^+(r)$, then $\str(x) \leq \str(y)$ and (ii)~if $x\in H(r)$ and $y
\in B^-(r)$, then $\str(x) < \str(y)$.

\begin{LEM}[\citey{AptBlairWalker88}]
\label{lem:strat-dbc}
$\NSTR=\DBC$.
\end{LEM}

The class~$\DBEC$, the largest class in $\Acyc$, has already been
studied by Zhao and Lin \shortcite{Zhao02,LinZhao04}, who showed that
every program in $\DBEC$ has at most one answer set, and this answer
set can be found in polynomial time. The proof involves the
well-founded semantics \shortcite{GelderRossSchlipf91}. For~$\DBC$ the
unique answer set can even be found in linear
time~\cite{NiemelaRintanen94}.

In our context this has the following important consequence.
\begin{PRO}\label{pro:enum}
  All classes in $\Acyc$ are enumerable, the classes~$\CCC\in \Acyc$
  with $\CCC\subseteq \DBC$ are even linear-time enumerable.
\end{PRO}

In view of Observation~\ref{obs:normal} and
Proposition~\ref{pro:enum}, all classes in $\Acyc$ satisfy the
requirement of Theorem~\ref{the:evaluation} and are therefore in
principle suitable target classes of a backdoor approach.
% However, before we can exploit such backdoors by means of
% Theorem~\ref{the:evaluation}, we need to find them first.
Therefore we will study the parameterized complexity of
\strongBdsDet{\CCC} and \delBdsDet{\CCC} for $\CCC\in Acyc$.  As we
shall see in the two subsections, the results for \strongBdsDet{\CCC}
are throughout negative, however for \delBdsDet{\CCC} there are
several (fixed-parameter) tractable cases.

\subsection{Strong Backdoor Detection}\label{sec:acyc-bdDet}

\begin{THE}\label{the:w2}
  For every target class~$\CCC\in \Acyc$ the problem
  \strongBdsDet{\CCC} is $\W[2]$-hard. If $\DC \subseteq \CCC$, then
  even \strongBdsDet{\tautext{\CCC}} is $\W[2]$-hard. Hence all these
  problems are unlikely to be fixed-parameter tractable.
\end{THE}
\begin{proof} 
  We give an fpt-reduction from the $\W[2]$-complete problem \HS to
  \strongBdsDet{\CCC}, see Section~\ref{sec:PC}.  Let $(\SSS,k)$ be an
  instance of this problem with $\SSS=\{S_1,\dots,S_m\}$.  We
  construct a program~$P$ as follows.  As atoms we take the elements
  of $U=\bigcup_{i=1}^m S_i$ and new atoms~$a_i^j$ and $b_i^j$ for
  $1\leq i \leq m$, $1\leq j \leq k+1$.  For each~$1\leq i \leq m$ and
  $1\leq j \leq k+1$ we take two rules~$r_i^j$, $s_i^j$ where
  $H(r_i^j)=\{a_i^j\}$, $B^-(r_i^j)=S_i \cup \{b_i^j\}$,
  $B^+(r_i^j)=S_i$; $H(s_i^j)=\{b_i^j\}$, $B^-(s_i^j)=\{a_i^j\}$,
  $B^+(s_i^j)=\emptyset$.

  We show that $\SSS$ has a hitting set of size at most~$k$ if and
  only if $P$ has a \strongBds{\CCC} of size at most~$k$.

  ($\Rightarrow$).  Let $H$ an hitting set of $\SSS$ of size at
  most~$k$. We choose an arbitrary truth assignment~$\tau \in \ta{H}$
  and show that $P_\tau \in C$. Since $H$ is a hitting set, each
  rule~$r_i^j$ will be removed when forming $P_\tau$. Hence the only
  rules left in $P_\tau$ are the rules~$s_i^j$, and so $P_\tau \in
  \DC\cap \C \subseteq \CCC$. Thus $H$ is a \strongBds{\CCC} of $P$.

  ($\Leftarrow$).  Let $X$ be a \strongBds{\CCC} of $P$ of size at
  most~$k$. We show that $H=X\cap U$ is a hitting set of $\SSS$.
  Choose $1\leq i \leq m$ and consider $S_i$.  We first consider the
  case~$\DC\subseteq \CCC$. For each~$1\leq j \leq k+1$ the
  program~$P$ contains a bad even directed cycle~$(a_i^j,b_i^j)$.  In
  order to destroy these cycles, $X$ must contain an atom from~$S_i$,
  since otherwise, $X$ would need to contain for each~$1\leq j \leq
  k+1$ at least one of the atoms from each cycle, but then
  $\Card{X}\geq k+1$, contradicting the assumption on the size of
  $X$. Hence $H$ is a hitting set of $\SSS$.  Now we consider the case
  $\C\subseteq \CCC$. For each $1\leq j \leq k+1$ the program~$P$
  contains a bad even cycle~$(a_i^j, v_{a_i^j,b_i^j}, b_i^j,
  v_{b_i^j,a_i^j})$. In order to destroy these cycles, $X$ must
  contain an atom from $S_i$, since otherwise, $X$ would need to
  contain an atom from each cycle, again a contradiction. Hence $H$ is
  a hitting set of $\SSS$. Hence the $\W[2]$\hy hardness of
  \strongBdsDet{\CCC} follows.

  In order to show that \strongBdsDet{\tautext{\CCC}} is $\W[2]$-hard
  for $\DC \subseteq \CCC$, we modify the above reduction from \HS by
  redefining the rules~$r_i^j$, $s_i^j$. We put $H(r_i^j)=\{a_i^j\}$,
  $B^-(r_i^j)=S_i \cup \{b_i^j\}$, $B^+(r_i^j)=\emptyset$;
  $H(s_i^j)=\{b_i^j\}$, $B^-(s_i^j)=\{a_i^j\}$, $B^+(s_i^j)=U$. By the
  very same argument as above we can show that $\SSS$ has a hitting
  set of size at most~$k$ if and only if $P$ has a
  \strongBds{\tautext{\CCC}} of size at most~$k$. We would like to
  state that this reduction does not work for the undirected cases as
  it yields undirected cycles~$(b^j_i,u,b^{j'}_{i'},u')$ for any $u,u'
  \in U$.

\end{proof}

For the class~$\DBEC$ we can again strengthen the result and show that
detecting a \strongBds{\DBEC}  is already $\coNP$\hy hard for
backdoor size 0; hence the problem is \coparaNP{}-hard (see
Section~\ref{sec:PC}).

\begin{THE}\label{the:paranp}
  The problem~\strongBdsDet{\tautext{\DBEC}} is \coparaNP{}-hard, and
  hence not fixed-parameter tractable unless $\P=\coNP$.
\end{THE}
\begin{proof}
  We reduce from the following problem, which is $\NP$\hy complete
  \cite{FortuneHopcroftWyllie80,LapaughPapadimitriou84},

  \problem{\textsc{Directed Path via a Node}}{a digraph~$G$ and
    $s,m,t\in V$ distinct vertices}{decide whether $G$ contains a
    directed path from~$s$ to~$t$ via~$m$}
  Let $G=(V,E)$ be a digraph and $s,m,t\in V$ distinct vertices.  We
  define a program~$P$ as follows: For each edge~$e=(v,w) \in E$ where
  $w\neq m$ we take a rule~$r_e$: $w \leftarrow v$. For each
  edge~$e=(v,m)$ we take a rule~$r_{e}$: $m \leftarrow \pnot
  v$. Finally we add the rule~$r_{s,t}$: $s \leftarrow \pnot t$. We
  observe that the dependency digraph of $P$ is exactly the digraph we
  obtain from $G$ by adding the ``reverse'' edge~$(t,s)$ (if not
  already present), and by marking $(t,s)$ and all incoming edges of
  $m$ as negative.
  
  We show that $G$ has a path from~$s$ to~$t$ via~$m$ if and only if
  $P\notin \DBEC$.  Assume $G$ has such a path. Then this path must
  contain exactly one incoming edge of $m$, and hence it contains
  exactly one negative edge. The path, together with the negative edge
  $(t,s)$, forms a directed bad even cycle of $P$, hence $P\notin
  \DBEC$.  Conversely, assume $P\notin \DBEC$. Hence the dependency
  digraph of $P$ contains a directed bad even cycle, i.e., a cycle
  that contains at least two negative edges. As it can contain at most
  one incoming edge of $m$, the cycle contains exactly one incoming
  edge of $m$ and the reverse edge~$(t,s)$. Consequently, the cycle
  induces in $G$ a directed path from~$s$ to~$t$ via~$m$.
%   The reduction shows that \strongBdsDet{\DBEC} is \coNP{}-hard
%   already for the parameter $k=0$. Next, we describe how this can be
%   generalized to arbitrary $k$. Let $P^k_{s,m,t}$ denote the program
%   obtained from $P_{s,m,t}$ by adding rules $s_i \leftarrow m_i$, $m_i
%   \leftarrow t_i$, and $t_i \leftarrow \pnot s_i$ where $s_i, m_i,
%   t_i$ are new atoms, for $1 \leq i \leq k$. The program $P^k_{s,m,t}$
%   has a \strongBds{\DBEC} of size at most~$k$ if and only if
%   $P^k_{s,m,t}(G) \in \DBEC$. So the problem \strongBdsDet{\DBEC} is
%   \coNP{}-hard. Hence the theorem follows.
\end{proof}  

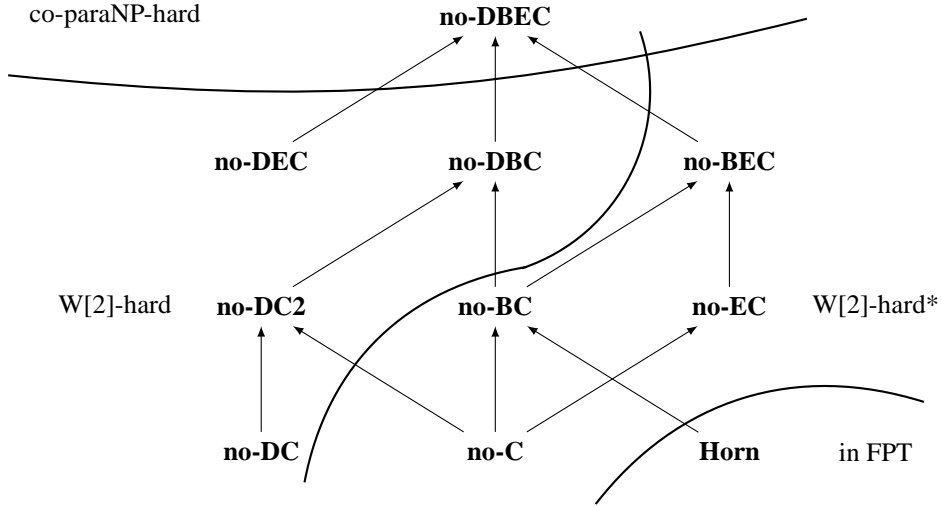
\begin{figure}
\centering
\begin{tikzpicture}[latex-,level 1/.style={sibling distance=8em},level distance=5em,font=\small]
  \node[] (dbec) {$\DBEC$}
  child { node (dec){$\DEC$}
    child{ node (dctwo) {$\DCTWO$} edge from parent[draw=none]
      child{ node (dc) {$\DC$} edge from parent[solid]}
    }
  }
  child { node (dbc){$\DBC$}
    child{ node (bc){$\BC$}
      child{ node(c){$\C$}}
      }
  }
  child { node (bec){$\BEC$}
    child{ node (ec){$\EC$}
      child{ node(horn){$\Horn$} edge from parent[transparent]}
      }
  };
  \path(dctwo) edge (c);
  \path(bc) edge (horn);
  \path(ec) edge (c);
  \path(bec) edge (bc);
  \path(dbc) edge (dctwo);

  \node[left of=dctwo,node distance=5em](xp){$\W[2]$-hard};
  \node[right of=ec,node distance=5em](wtwostar){$\W[2]$-hard\textup{*}};
%  \node[below of=xp,node distance=1em](xp2){\it(\FPT open)};
  \node[above of=xp,node distance=10em](coparanp){\coparaNP-hard};
  \node[below of=wtwostar,node distance=5em](fpt){in \FPT};

  \node[below of=xp,node distance=3em,xshift=-4em](HELPER_XP_LOWER){};
  \node[above of=xp,node distance=3em,xshift=-4em](HELPER_XP_UPPER){};

  \node[right of=dec,node distance=3em,xshift=2em,yshift=2.9em](DEC2DBC){};
  \node[below right of=dc,node distance=2em](DCB){};
  \node[above right of=bc,node distance=2em](BCT){};
  \node[above right of=bc,node distance=2em,xshift=-0.85em,yshift=-0.22em](BCTTWO){};
  \node[above right of=dbc,node distance=4em,xshift=2em,yshift=2em](DBCT){};
  \path (DCB) edge[bend left=35,-,thick] (BCT);
  \path (BCTTWO) edge[bend right=45,-,thick] (DBCT);
  %\path (HELPER_XP_LOWER) edge[bend right=35,-,ultra thick] (DEC2DBC);

  \node[above left of=horn,xshift=-3em,yshift=-3.1em](HORNM){};
  \node[above left of=horn,xshift=-3em,yshift=-4em](HORNMTWO){};
  \node[below of=fpt,node distance=3em,xshift=1em,yshift=0em](HORNL){};
  \node[above of=fpt,xshift=2em,yshift=-1em](HORNR){};

%  \path (HORNL) edge[bend left=25,-,ultra thick] (HORNM);
  \path (HORNMTWO) edge[bend left=35,-,thick] (HORNR);
  %\path (BCTTWO) edge[bend right=45,-,ultra thick] (DBCT);

  \node[right of=dbec,node distance=11em](DBEC_R){};
  \node[below of=coparanp,node distance=2em,xshift=-4em](COPARA_B){};
  \path(COPARA_B) edge[bend right=10,-,thick] (DBEC_R);

\end{tikzpicture}%
\caption{Known complexity of the problem \strongBdsDet{\CCC}. 
%The XP-result is established in Theorem~\ref{the:horn}. 
% The $\W[2]$-hardness results are established in
%   Theorem~\ref{the:w2}. The \coparaNP-hardess result is established in
%   Theorem~\ref{the:paranp}. 
  (*)~When we permit tautologies in the rules.}
\label{fig:diagram-strongBds}
\end{figure}

Figure~\ref{fig:diagram-strongBds} illustrates the known complexity
results of the problem \strongBdsDet{\CCC}. An arrow from $\CCC$ to
$\CCC'$ indicates that $\CCC'$ is a proper subset of $\CCC$ and hence
the size of a smallest \strongBds{\CCC'} is at most the size of a
smallest \strongBds{\CCC}.

% \begin{EX}
%   TODO
% \end{EX}

\subsection{Deletion Backdoor Detection}\label{sec:db-cycl}
The $\W[2]$-hardness results of Theorems~\ref{the:w2} and
\ref{the:paranp} suggest to relax the problem and to look for
\emph{deletion backdoors} instead of strong backdoors.  In view of
Lemma~\ref{lem:rule-induced} and
Observation~\ref{obs:acyclic_rule_induced}, every deletion backdoor is
also a strong backdoor for the considered acyclicity-based target
classes, hence the backdoor approach of Theorem~\ref{the:evaluation}
works.

Fortunately, the results of this section show that the relaxation
indeed gives us fixed-parameter tractability of backdoor detection for
most considered classes. Figure~\ref{fig:lattice-delbds} illustrates
these results. We obtain these results by making use of very recent
progress in fixed-parameter algorithmics on various variants of the
\emph{feedback vertex set} or the \emph{cycle transversal} problems.

Consider a graph~$G=(V,E)$ and a set~$W\subseteq V$. A cycle in $G$ is
a $W$-cycle if it contains at least one vertex from $W$. A set
$T\subseteq V$ is a \emph{$W$-cycle transversal} of $G$ if every
$W$-cycle of $G$ is also a $T$\hy cycle.  A set~$T\subseteq V$ is an
\emph{even-length $W$-cycle transversal} of $G$ if every $W$-cycle of
$G$ of even length is also a $T$\hy cycle. A $V$\hy cycle transversal
is also called a \emph{feedback vertex set}.

We give analog definitions for a digraph~$G=(V,E)$ and $W\subseteq V$.
A directed cycle in $G$ is a directed $W$\hy cycle if it contains at
least on vertex from $W$. A set~$T\subseteq V$ is a \emph{directed
  $W$-cycle transversal} of $G$ if every directed $W$\hy cycle of $G$
is also a directed $T$\hy cycle. A set~$T\subseteq V$ is an
\emph{directed even-length $W$-cycle transversal} of $G$ if every
directed $W$\hy cycle of $G$ of even length is also a directed $T$\hy
cycle. A directed $V$\hy cycle transversal is also called a
\emph{directed feedback vertex set}.

% The definition generalizes to cycles that run through a given
% set~$W$ of vertices. So let $W\subseteq V$. A \emph{$W$\hy feedback
% vertex set} is a set~$S\subseteq V$ that hits every cycle in $G$
% that runs through the set~$W$. A $V$\hy feedback vertex set of $G$
% is obviously a feedback vertex set of $G$. In the following, we are
% interested in the problem:

% \pproblem{\fvsDet[W]}{A graph $G=(V,E)$, a set $W \subseteq V$, and
% an integer~$k$.}{The integer $k$.}{Decide whether $G$ has a
% $W$-feedback vertex set~$S\subseteq V$ of size at most $k$.}

% An \emph{even $W$\hy cycle transversal} (even $W$\hy feedback vertex
% set) of $G$ is a set $S\subseteq V$ that hits every $W$\hy cycle in
% $G$ that runs through the set~$W$ and has even length.

% \pproblem{\ectDet{W}}{A graph $G=(V,E)$, a set $W \subseteq V$, and
% an integer~$k$.}{The integer $k$.}{Decide whether $G$ has an even
% $W$\hy cycle transversal~$S \subseteq V$ of size at most $k$.}

% The notions naturally extend to directed graphs.
% %
% }

\begin{THE}\label{the:fvs}
  The problem \delBdsDet{\tautext{\CCC}} is fixed-parameter tractable
  for all~$\CCC\in \Acyc \setminus \{\DEC,\DBEC\}$.
%\{\C$, \BC, \DC, \DCTWO, \DBC, \EC, $\BEC\}$ be a class
%  of programs.
\end{THE}
\begin{proof}
  Let $P^*$ be a the program and $k\geq 0$.  We delete from $P^*$ all
  constraints and tautological rules.  Now, the
  \delBds{\tautext{\CCC}}s of~$P^*$ are exactly the \delBds{\CCC}s of
  $P$. Hence we can focus on the latter. Let $U_p$ be the dependency
  graph and $D_p$ the dependency digraph of $P$, respectively. Next we
  consider the various target classes~$\CCC$ mentioned in the
  statement of the theorem, one by one, and show how we can decide
  whether $P$ has a \delBds{\CCC} of size at most~$k$.

  First we consider ``undirected'' target classes.  Downey and Fellows
  \cite{DowneyFellows99} have shown that finding an feedback vertex
  set of size at most~$k$ is fixed-parameter tractable. We apply their
  algorithm to the dependency graph~$U_p$. If the algorithm produces a
  feedback vertex set~$S$ of size at most~$k$, then we can form a
  \delBds{\C} of $P$ of size at most~$k$ by replacing each negative
  vertex in $S$ by one of its two neighbors, which always gives rise
  to an atom of~$P$. If $U_p$ has no feedback vertex set of size at
  most~$k$, then $P$ has no \delBds{\C} of size at most~$k$.  Hence
  \delBdsDet{\C} is fixed-parameter tractable.  Similarly,
  \delBdsDet{\BC} is fixed-parameter tractable by finding a \sfvs{W}
  of $U_p$, taking as $W$ the set of bad vertices of
  $U_p$. \citex{CyganPilipczukPilipczukWojtaszczyk11} and
  \citex{KawarabayashiKobayashi10} showed that finding a \sfvs{W} is
  fixed-parameter tractable, hence so is \delBdsDet{\BC}.

  In order to extend this approach to \delBdsDet{\EC}, we would like
  to use fixed-parameter tractability of finding an \ect{W}, which was
  established by~\citex{MisraRamanRamanujanSaurabh12} for $W=V$, and
  by \citex{KakimuraKawarabayashiKobayashi12} for general $W$.  In
  order to do this, we use the following trick of Aracena, Gajardo,
  and Montalva~\shortcite{MontalvaAracenaGajardo08}, that turns cycles
  containing an even number of bad vertices into cycles of even
  length. From $D_p$ we obtain a graph~$U_P'$ by replacing each
  negative edge~$e=(x,y)$ with three edges~$xu_e$, $u_ev_e$, and
  $v_ey$ where $u_e$ and $v_e$ are new negative vertices, and by
  replacing each remaining directed edge~$(u,v)$ with two edges~$xw_e$
  and $w_ey$ where $w_e$ is a new (non-negative) vertex. We observe
  that $U_p'$ can be seen as being obtained from $D_p$ by subdividing
  edges. Hence there is a natural 1-to-1 correspondence between cycles
  in $U_p$ and cycles in $U_p'$. Moreover, a cycle of $U_p$ containing
  an even number of negative vertices corresponds to a cycle of $U_p'$
  of even length, and a bad cycle of $U_p$ corresponds to a bad cycle
  of $U_p'$.  Thus, when we have an even cycle transversal~$S$ of
  $U_p'$, we obtain a \delBds{\EC} by replacing each negative vertex
  $v \in S$ by its non-negative neighbor. Hence \delBdsDet{\EC} is
  fixed-parameter tractable.  For \delBdsDet{\BEC} we proceed
  similarly, using a \ect{W} of $U_p'$, letting $W$ be the set of
  negative vertices of $U_p'$.

  We now proceed with the remaining ``directed'' target classes $\DC$,
  $\DCTWO$, and $\DBC$.

  % A \emph{directed $S$\hy feedback vertex set} of $G$ is a set
  % $F\subseteq V$ such that every directed cycle of $G$ that contains
  % at least one edge from $S$ runs through a vertex in $F$.

  Let $G=(V,E)$ be a digraph. Evidently, the directed feedback vertex
  sets of $D_p$ are exactly the \delBds{\DC}s of $P$. Hence, by using
  the fixed-parameter algorithm of \citex{ChenLiuLuOsullivanRazgon08}
  for finding directed feedback vertex sets we obtain fixed-parameter
  tractability of \delBdsDet{\DC}.
  
  To make this work for \delBdsDet{\DCTWO} we consider instead of
  $D_p$ the digraph~$D_p'$ obtained from $D_p$ by replacing each
  negative edge~$e=(u,v)$ by two (non-negative) edges~$(u,w_e)$,
  $(w_e,v)$, where $w_e$ is a new vertex. The directed cycles of $D_p$
  and $D_P'$ are in a 1-to-1 correspondence. However, directed cycles
  of length~$2$ in $D_p'$ correspond to good cycles of length~$2$ in
  $D_p$. \citex{BonsmaLokshtanov11} showed that finding a directed
  feedback vertex set that only needs to cut cycles of length at
  least~$3$ is fixed-parameter tractable. Applying this algorithm to
  $D_P'$ (and replacing each vertex~$w_e$ in a solution with one of
  its neighbors) yields fixed-parameter tractability of
  \delBdsDet{\DCTWO}.

  The approach for \delBdsDet{\DC} extends to \delBdsDet{\DBC} by
  considering directed \sfvs{W}s of the digraph~$D_p'$ obtained from
  $D_p$ using a simple construction already mentioned by
  \citex{CyganPilipczukPilipczukWojtaszczyk11} where we replace each
  negative edge~$e=(u,v)$ by two (non-negative) edges~$(u,w_e)$,
  $(w_e,v)$ and $W = \SB w_e \SM e \text{ is a negative edge}
  \SE$. The directed $W$-cycles of $D'_p$ and the directed bad cycles
  of $D_P$ are obviously in a 1-to-1 correspondence. Thus when we have
  a directed \sfvs{W}~$S$ of~$D_P'$, we obtain a \delBds{\DBC} by
  replacing each vertex~$v \in S \cap W$ by its neighbor. The
  fixed-parameter tractability of finding a directed \sfvs{W} was
  shown by \citex{ChitnisCyganHajiaghayiMarx12}.
\end{proof}

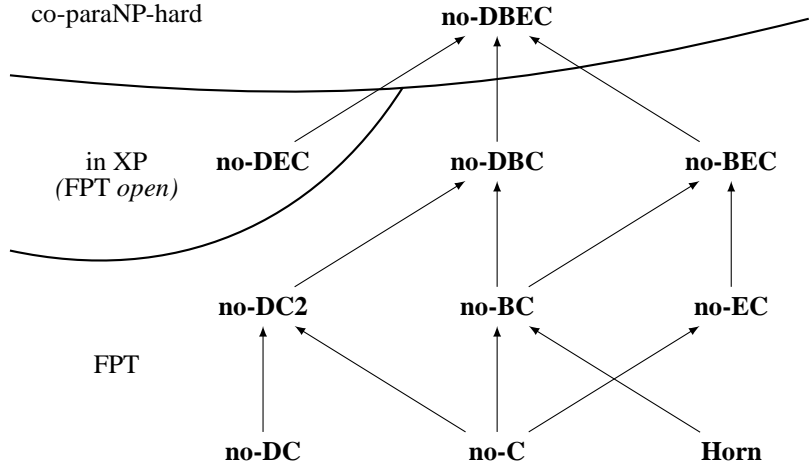
\begin{figure}
\centering
\begin{tikzpicture}[latex-,level 1/.style={sibling distance=8em},level distance=5em,font=\small]
  \node[] (dbec) {$\DBEC$}
  child { node (dec){$\DEC$}
    child{ node (dctwo) {$\DCTWO$} edge from parent[draw=none]
      child{ node (dc) {$\DC$}}}
  }
  child { node (dbc){$\DBC$}
    child{ node (bc){$\BC$}
      child{ node(c){$\C$}}
    %  child{ node(horn){$\Horn$}}
      }
  }
  child { node (bec){$\BEC$}
    child{ node (ec){$\EC$}
      child{ node(horn){$\Horn$} edge from parent[transparent]}
    }
  };
  \path(dctwo) edge (c);
  \path(bc) edge (horn);
  \path(ec) edge (c);
  \path(bec) edge (bc);
  \path(dbc) edge (dctwo);

  \node[left of=dec,node distance=5em](xp){in \XP};
  \node[below of=xp,node distance=1em](xp2){\it(\FPT open)};
  \node[above of=xp,node distance=5em](coparanp){\coparaNP-hard};
  \node[below of=xp,node distance=7em](fpt){\FPT};

  \node[below of=xp,node distance=3em,xshift=-4em](HELPER_XP_LOWER){};
  \node[above of=xp,node distance=3em,xshift=-4em](HELPER_XP_UPPER){};

  \node[right of=dec,node distance=3em,xshift=2em,yshift=2.9em](DEC2DBC){};
  \path (HELPER_XP_LOWER) edge[bend right=35,-, thick] (DEC2DBC);

  \node[right of=dbec,node distance=11em](DBEC_R){};
  \node[below of=coparanp,node distance=2em,xshift=-4em](COPARA_B){};
  \path(COPARA_B) edge[bend right=10,-, thick] (DBEC_R);

\end{tikzpicture}%
\caption{Relationship between classes of programs and known complexity
  of the problem \delBdsDet{\CCC}. An arrow from $\CCC$ to $\CCC'$
  indicates that \delBds{\CCC}s are smaller than \delBds{\CCC'}s. The
  FPT-results are established in Theorems~\ref{the:horn} and
  \ref{the:fvs}. The XP-result is established in
  Theorem~\ref{the:xp}. The \coparaNP-hardness result is established
  in Theorem~\ref{the:paranp_del}.}
\label{fig:lattice-delbds}
\end{figure}

According to Observation~\ref{obs:acyclic_rule_induced}, the classes
mentioned in Theorem~\ref{the:fvs} are hereditary. Hence using
Theorem~\ref{the:fvs} we can drop the assumption in
Theorem~\ref{the:evaluation} that the backdoor is given and obtain
directly:
\begin{THE}\label{cor:fpt}
  For all~$\CCC\in \Acyc \setminus \{\DEC,\DBEC\}$ all problems in
  $\AspFull$ are fixed-parameter tractable when parameterized by the
  size of a smallest \delBds{\tautext{\CCC}}.
\end{THE}

Let us now turn to the two classes~$\DEC$, $\DBEC$ excluded in
Theorem~\ref{the:fvs}. We cannot establish that
\delBdsDet{\tautext{\DEC}} is fixed-parameter tractable, as the
underlying even cycle transversal problem seems to be currently out of
reach to be solved. However, in Theorem~\ref{the:xp} below, we can at
least show that for every constant~$k$, we can decide in polynomial
time whether a \strongBds{\tautext{\DEC}} of size at most~$k$ exists;
thus the problem is in $\XP$. For \delBdsDet{\tautext{\DBEC}} the
situation is different: here we can rule out fixed-parameter
tractability under the complexity theoretical assumption $\P\neq\coNP$
(Theorem~\ref{the:paranp_del}).

% Next we will see that there is still hope that the problem
% \delBdsDet{\DEC} is fixed-parameter tractable and show that it is very
% unlikely due to standard complexity theoretical assumptions that the
% problem \delBdsDet{\DBEC} is fixed-parameter tractable.
\begin{samepage}
\begin{THE}\label{the:xp}
  The problem \delBdsDet{\tautext{\DEC}} is in \XP.
\end{THE}
\begin{proof}
  Let $P$ be a program, $n$ the input size of~$P$, and $k$ be a
  constant. W.l.o.g., we assume that $P$ has no tautological rules or
  constraints.  We are interested in a \delBds{\DEC} of $P$ of size at
  most~$k$.  We loop over all possible sets~$X\subseteq \at(P)$ of
  size at most~$k$. Since $k$ is a constant, there is a polynomial
  number~$\BigO{n^k}$ of such sets~$X$. To decide whether $X$ is a
  \delBds{\DEC} of $P$, we need to check whether $P-X \in \DEC$.
  % THIS WAS FOR STRONG BACKDOORS
  % We
 % need a constant number ($1$ or $2^k$) membership checks. Hence, if
 % we show that a single membership check runs in polynomial time, we
 % have established the theorem. 
  For the membership check $P - X \in \DEC$ we have to decide whether
  $D_{P-X}$ contains a bad even cycle. We use a directed variant of
  the trick in the proof of Theorem~\ref{the:fvs} (in fact, the
  directed version is slightly simpler). Let $D_{P-X}$ be the
  dependency digraph of $P-X$. From $D_{P-X}$ we obtain a new
  digraph~$D'_{P-X}$ by subdividing every non-negative edge, i.e., we
  replace each non-negative edge~$e=(x,y)$ by two (non-negative)
  edges~$(x,u_e)$, $(u_e,y)$ where $u_e$ is a new vertex.  Obviously,
  directed even cycles in $D_{P-X}$ are in 1-to-1 correspondence with
  directed cycles of even length in $D'_{P-X}$.  Whether a digraph
  contains a directed cycle of even length can be checked in
  polynomial time by means of the following results.
  \citex{VaziraniYannakakis88} have shown that finding a cycle of even
  length in a digraph is equivalent to finding a so-called Pfaffian
  orientation of a graph. Since Robertson, Seymour, and
  Thomas~\cite{RobertsonSeymourThomas99} have shown that a Pfaffian
  orientation can be found in polynomial time, the test works in
  polynomial time.
\end{proof}
\end{samepage}

\begin{THE}\label{the:paranp_del}
  The problem \delBdsDet{\tautext{\DBEC}} is \coparaNP{}-hard, and hence
  not fixed-parameter tractable unless $\P=\coNP$
\end{THE}
\begin{proof}
  The theorem follows from the reduction in the proof of
  Theorem~\ref{the:paranp}.
\end{proof}
% THIS IS NOT AN EX
%\begin{EX}

%  If there is no arrow between two classes (or the arrow does not
%  follow by transitivity of set inclusion), then the two classes are
%  incomparable.
%\end{EX}

% \begin{REM}
%   We obtain as a corollary of~Theorem~\ref{the:paranp_del} that the
%   problem \DPSFVS (which asks given a digraph~$G=(V,E)$, the
%   set~$H\subseteq E$ of negative edges, and an integer~$k$, to give
%   a set~$X\subseteq V$ of size at most~$k$ such that $G[V-X]$
%   contains no bad even cycles, or report that such $X$ does not
%   exist) is \coparaNP-hard.
% \end{REM}

\newcommand{\sbd}{\parm{sb}}
\newcommand{\dbd}{\parm{db}}
\newcommand{\twinc}{\parm{treewidth}_\parm{inc}}
\newcommand{\twdep}{\parm{treewidth}_\parm{dep}}

\section{Kernelization}\label{sec:kernel} 

%By preprocessing we mean the polynomial-time reduction of an instances
%of a computational problem to an equivalent instances of smaller
%size. 
If we want to solve a hard problem, then in virtually every setting,
it is beneficial to first apply a polynomial preprocessing to a given
problem instance.  In particular, polynomial-time preprocessing
techniques have been developed in \ASP solving (see e.g.,
\cite{FaberLeoneMateisPfeifer99,GebserKaufmannNeumannSchaub08,GebserKaminskiKaufmannSchaub11a}). However,
polynomial-time preprocessing for \NP-hard problems has mainly been
subject of empirical studies where provable performance guarantees are
missing, mainly due to the fact that if we can show that if we can
reduce in polynomial-time a problem instance by just one bit, then by
iterating this reduction we can solve the instances in polynomial
time.  Contrastingly, the framework of parameterized complexity offers
with the notion of \emph{kernelization} a useful mathematical
framework that admits the rigorous theoretical analysis of
polynomial-time preprocessing for \NP-hard problems.
%
% \citex{Szeider11} has considered kernelizations for various
%AI problems.  \citex{MorakWoltran12} have established a method based
%on kernelizations for preprocessing complex non-ground rules in \ASP.
%
A kernelization is a polynomial-time reduction that replaces the input
by a smaller input, called a ``kernel'', whose size is bounded by some
computable function of the parameter only.  A well known result of
parameterized complexity theory is that a decidable problem is
fixed-parameter tractable if and only if it admits a
kernelization~\cite{DowneyFellowsStege99}. The result leads us to the
question of whether a problem also has a kernelization that reduces
instances to a size which is polynomially bounded by the parameter,
so-called \emph{polynomial kernels}. Indeed, many \NP-hard
optimization problems admit polynomial kernels when parameterized by
the size of the solution~\cite{Rosamond10}. In the following we
consider kernelizations for backdoor detection and backdoor evaluation
in the context of ASP. We establish that for some target classes,
backdoor detection admits a polynomial kernel.
% provide strong evidence We point out that the question of polynomial
% kernels is undecided for some target classes.
We further provide strong theoretical evidence that for all target
classes considered, backdoor evaluation does admit a polynomial
kernel.

We will later use the following problem: \pproblem{\VC}{A
  graph~$G=(V,E)$ and an integer~$k$.}{The integer~$k$.}{Decide
  whether there is a vertex cover~$S \subseteq V$ (see
  Section~\ref{sec:horn}) of size at most~$k$ .}

Next we give a more formal definition of kernelization.  Let $L,L'
\subseteq \Sigma^* \times \Nat$ be parameterized problems. A
\emph{bi-kernelization} is a polynomial-time many-to-one reduction
from the problem~$L$ to problem~$L'$ where the size of the output is
bounded by a computable function of the parameter. That is, a
bi-kernelization is an algorithm that, given an instance~$(I,k) \in
\Sigma^* \times \Nat$ outputs for a constant~$c$ in time $O((\CCard{I}
+ k)^{d})$ a pair~$(I',k') \in \Sigma^* \times \Nat$, such that
(i)~$(I,k) \in L$ if and only if $(I',k') \in L'$ and (ii)~$\CCard{I'}
+ k' \leq g(k)$ where $g$ is an arbitrary computable function, called
the size of the kernel. If $L'=L$ then the reduction is called a
\emph{kernelization}, the reduced instance a \emph{kernel}.  If $g$ is
a polynomial then we say that $L$ admits a \emph{polynomial
  (bi-)kernel}, for instance, the problem \VC has a kernel of at
most~$2k$ vertices and thus admits a polynomial kernel
\cite{ChenKanjXia10}. $L$ is called \emph{compressible} if it admits a
polynomial bi-kernel.

\hide{A kernelization of a decidable parameterized problem~$L$
  immediately yields a fixed-parameter tractable algorithm for $L$
  since we can apply the kernelization and then use any brute force
  method and we still have a fixed-parameter tractable algorithm. On
  the other hand a fixed-parameter tractable algorithm gives a
  kernelization. Suppose that an algorithm~$A$ is fixed-parameter
  tractable, i.e., the algorithm~$A$ solves instances~$(I,k)$ in
  time~$\BigO{f(k) \CCard{I}^c}$ for a computable function~$f$ and a
  constant~$c$. Let $(I,k)$ be an instance. If $\CCard{I} \leq f(k)$
  then we output $(I,k)$. If $\CCard{I} \geq f(k)$ then we apply the
  algorithm~$A$ which decides in time $\BigO{\CCard{I}^{c + 1}}$
  whether $(I,k) \in L$. If $(I,k) \in L$ we output an arbitrary
  instance $(I',1)\in L$ ; if $(I,k) \not \in L$ we output an
  arbitrary instance $(I', 1) \not \in L$ of the problem. In total the
  algorithm runs in time at most $\BigO{\CCard{I}^{c+ 1}}$ which is in
  fact polynomial. The connection was observed by Downey, Fellows, and
  Stege~\cite{DowneyFellowsStege99} and is stated in the following
  proposition: }

The following proposition states the connection between
fixed-parameter tractable problems and kernels, as observed by Downey,
Fellows, and Stege~\cite{DowneyFellowsStege99}:

\begin{PRO}[\citey{DowneyFellowsStege99}, \citey{FlumGrohe06}]\label{pro:kernel}
  A parameterized problem is fixed-parameter tractable if and only if
  it is decidable and has a kernelization.
\end{PRO}

Thus, our fixed-parameter tractability results of
Theorems~\ref{the:evaluation}, \ref{the:horn}, and \ref{the:fvs}
immediately provide that the mentioned problems admit a kernelization.
In the following we investigate whether these problems admit
polynomial kernels.

% \begin{COR}\label{cor:kernel}
%   Let $\CCC \in \{$\Horn, \C, \BC, \DC, \DCTWO, $\DBC\}$ be a class of
%   programs. The problems \pname{Consistency}, \pname{Brave} and
%   \pname{Skeptical Reasoning} has a kernelization when parameterized
%   by the size of a \strongBds{\taut{\CCC}}. Further the problems
%   \strongBdsDet{\tautext{Horn}} and \delBdsDet{\tautext{CCC}} have a kernelization.
% \end{COR}

\subsection{Backdoor Detection}

The first result of this section is quite positive.

\begin{THE}\label{the:kern-horn-c}
  For $\CCC \in \{$\Horn, $\C\}$ the problem
  \delBdsDet{\tautext{\CCC}} admits a polynomial kernel.
  For~$\CCC=\Horn$ the kernel has a linear number of atoms,
  for~$\CCC=\C$ the kernel has a quadratic number of atoms.
% and hence the
%  problem \strongBdsDet{\tautext{\Horn}} also admits a polynomial
%  kernel.
\end{THE}
\begin{proof}
  First consider the case $\CCC=\Horn$. Let $(P,k)$ be an instance of
  \delBdsDet{\tautext{\Horn}}. We obtain in polynomial time the
  negation dependency graph~$N_P$ of $P$ and consider $(N_P,k)$ as an
  instance of \VC.  We use the kernelization algorithm of
  \citex{ChenKanjXia10} for \VC and reduce in polynomial time
  $(N_p,k)$ to a \VC instance~$(G,k')$ with at most~$2k$ many
  vertices. It remains to translate $G$ into a program~$P'$ where
  $N_{P'}=G$ by taking for every edge~$xy\in E(G)$ a rule~$x
  \leftarrow \neg y$. Now $(P',k')$ is a polynomial kernel with a
  linear number of atoms.

  Second consider the case~$\CCC=\C$. Let $(P,k)$ be an instance of
  \delBdsDet{\tautext{\C}}.  We obtain in polynomial time the
  dependency graph~$U_P$ of $P$ and consider $(U_P,k)$ as an instance
  of \FVS (see Section~\ref{sec:db-cycl}).  We use the kernelization
  algorithm of \citex{Thomasse09} for \FVS and reduce in polynomial
  time~$(U_p,k)$ to a \FVS instance~$(G',k')$ with at most~$4 k^2$
  vertices. As above we translate $G$ into a program~$P'$ where
  $U_{P'}=G$ by taking for every edge~$xy\in E(G)$ a rule~$x
  \leftarrow \neg y$. Now $(P',k')$ is a polynomial kernel with a
  quadratic number of atoms.
\end{proof}

Similar to the construction in the proof of Theorem~\ref{the:fvs} we
can reduce for the remaining classes the backdoor detection problem to
variants of feedback vertex set. However, for the other variants of
feedback vertex set no polynomial kernels are known.

We would like to point out that the kernels obtained in the proof of
Theorem~\ref{the:kern-horn-c} are equivalent to the input program with
respect to the existence of a backdoor, but not with respect to the
decision of reasoning problems. In the next subsection we consider
kernels with respect to reasoning problems.

\subsection{Backdoor Evaluation}

Next we consider the problems in $\AspReason$. We will see that
neither of them admit a polynomial kernel when parameterized by the
size of a \strongBds{\CCC} for the considered target classes, subject
to standard complexity theoretical assumptions.

Our superpolynomial lower bounds for kernel size are based on a result
by \citex{FortnowSanthanam11} regarding satisfiability parameterized
by the number of variables.

\pproblem{\pname{Sat[Vars]}}{A CNF formula~$F$.}{The number~$k$ of
  variables of $F$.}{Decide whether $F$ is satisfiable.}

% from Szeider11
% A \emph{composition algorithm} for a parameterized problem $L\subseteq
% \Sigma^* \times \NAT$ is an algorithm that receives as input a
% sequence $(I_1,k),\ldots,(I_t,k) \in \Sigma^* \times \NAT$ of
% instances of $L$, uses time polynomial in 􏱠$\sum_{i=1}^{t} \Card{I_i}
% +k$, and outputs an instance $(I, k′)\in \Sigma^* \times \NAT$
% (i)~$(I',k')\in L$ if and only if $(I_i, k) \in L$ for some $1 \leq i
% \leq k$, and (ii)~$k'$ is polynomial in $k$. A parameterized problem
% is \emph{compositional} if it has a composition algorithm. The
% following result, which is due to Fortnow and
% Santhanam~\cite{FortnowSanthanam08}, is the basis for our kernel lower
% bounds.
% \begin{PRO}[\cite{FortnowSanthanam08}]
%   Let $L$ be a parameterized problem whose unparameterized version is
%   \NP-complete. If $L$ is compositional, then it does not admit a
%   polynomial kernel unless $\NP\subseteq \coNP/poly$ i.e., the
%   Polynomial Hierarchy collapses.
% \end{PRO}

\begin{PRO}[\citey{FortnowSanthanam11}]\label{pro:poly_kernels}
  If \pname{Sat[Vars]} is compressible, then the Polynomial Hierarchy
  collapses to its third level.
\end{PRO}

The following theorem extends a result for normal programs
\cite{Szeider11}. We need a different line of argument, as the
technique used in \cite{Szeider11} only applies to problems in $\NP$
or $\coNP$.

%  for normal programs. Here our lower bounds
% are based on theoretical evidence that \pname{Sat[Vars] is not
% compressible, and not on theoretical evidence that \pname{Sat[Vars]}
% has no polynomial kernel. This allows us to us obtain kernel lower
% bounds for problems whose unparameterized versions are outside NP, and
% hence for reasoning problems on disjunctive programs.

\begin{THE}\label{the:non_poly_kernels}
  Let $\CCC\in \Acyc \cup \{\Horn \}$.  Then no problem in
  \AspReason admits a polynomial kernel when
  parameterized by the size of a smallest \strongBds{\CCC} or
  \delBds{\CCC}, unless the Polynomial Hierarchy collapses to its
  third level.
\end{THE}
\begin{proof}
  We show that the existence of a polynomial kernel for any of the
  above problems implies that \pname{Sat[Vars]} is compressible, and
  hence by Proposition~\ref{pro:poly_kernels} the collapse would
  follow.

  First consider the problem \pname{Consistency}.  From a CNF
  formula~$F$ with $k$ variables we use a reduction of
  Niemela~\cite{Niemela99} and construct a program~$P_1$ as follows:
  Among the atoms of our program~$P_1$ will be two atoms~$a_{x}$ and
  $a_{\bar x}$ for each variable~$x\in \var(F)$, an atom~$b_C$ for
  each clause~$C \in F$. We add the rules~$a_{\bar x} \leftarrow \pnot
  a_x$ and $a_x \leftarrow \pnot a_{\bar x}$ for each variable~$x \in
  \var(F)$. For each clause~$C \in F$ we add for each~$x \in C$ the
  rule~$b_C \leftarrow a_{x}$ and for each~$\neg x \in C$ the
  rule~$b_C \leftarrow a_{\bar x}$. Additionally, for each clause~$C
  \in F$ we add the rule~$ \leftarrow \pnot b_C$.  Now it is easy to
  see that the formula~$F$ is satisfiable if and only if the
  program~$P_1$ has an answer set.  We observe that $X = \SB a_x \SM
  x\in \var(F)\SE$ ($X = \SB a_x, a_{\bar x} \SM x\in \var(F)\SE$) is
  a smallest deletion (and smallest strong) $\CCC$\hy backdoor of
  $P_1$ for each $\CCC\in \Acyc$ ($\CCC = \Horn$). Hence $(P_1,k)$,
  $(P_1,2k)$ respectively, is an instance of \pname{Consistency},
  parameterized by the size of a smallest \strongBds{\CCC} or
  \delBds{\CCC}, and if this problem would admit a polynomial kernel,
  this would imply that \pname{Sat[Vars]} is compressible.

  For the problem \pname{Brave Reasoning} we modify the reduction from
  above. We create a program~$P_2$ that consists of all atoms and
  rules from $P_1$. Additionally, the program~$P_2$ contains an
  atom~$t$ and a rule~$r$ with $H(r)=\{t\}$, $B^+(r) = \emptyset$, and
  $B^-(r)=\emptyset$.  We observe that the formula~$F$ is satisfiable
  if and only if the atom~$t$ is contained in some answer set
  of~$P_2$. Since $X$ is still a backdoor of size~$k$ ($2k$), and a
  polynomial kernel for \pname{Brave Reasoning}, again it would yield
  that \pname{Sat[Vars]} is compressible.

  Let \pname{UnSat[Vars]} denote the problem defined exactly like
  \pname{Sat[Vars]}, just with yes and no answers swapped.  A
  bi-kernelization for \pname{UnSat[Vars]} is also a bi-kernelization
  for \pname{Sat[Vars]} (with yes and no answers swapped). Hence
  \pname{Sat[Vars]} is compressible if and only if \pname{UnSat[Vars]}
  is compressible.  An argument dual to the previous one for
  \pname{Brave Reasoning} shows that a polynomial kernel for
  \pname{Skeptical Reasoning}, parameterized by backdoor size, would
  yield that \pname{UnSat[Vars]} is compressible, which, as argued
  above, would yield that \pname{Sat[Vars]} is compressible.
\end{proof}

\section{Lifting Parameters}\label{sec:lifting}
In this section we will introduce a general method to lift
ASP-parameters that are defined for normal programs to disjunctive
programs. Thereby we extend several algorithms that have been
suggested for normal programs to disjunctive programs.  The lifting
method also gives us an alternative approach to obtain some results of
Section~\ref{sec:acyclic-classes}. Throughout this section we assume
for simplicity that the input program $P$ has no tautological rules or
constraints, all considerations can be easily extended to the general
case.
  
%
%\todo{Can we define the lifting lemma for all parameters where
%  enumeration is bounded}
The following definition allows us to speak about parameters for
programs in  a  more abstract way.

\begin{DEF}\label{def:asp-parm}
  An \emph{ASP-parameter} is a function~$p$ that assigns every
  program~$P$ some non-negative integer~$p(P)$ such that $p(P') \leq
  p(P)$ holds whenever $P'$ is obtained from~$P$ by deleting rules or
  deleting atoms from rules. If $p$ is only defined for normal
  programs, we call it a \emph{normal ASP-parameter}.  For an ASP
  parameter~$p$ we write $\down{p}$ to denote the normal ASP-parameter
  obtained by restricting $p$ to normal programs.
\end{DEF}

%NEXT

We impose the condition $p(P') \leq p(P)$ for technical reasons. This
is not a limitation, as most natural parameters satisfy this condition.

There are natural ASP-parameters associated with backdoors:
\begin{DEF}
  For a class~$\CCC$ of programs and a program~$P$ let
  $\sbd_\CCC(P)$ denote the size of a smallest \strongBds{\CCC} and
  $\dbd_\CCC(P)$ denote the size of a smallest \delBds{\CCC} of $P$.
\end{DEF}

We will ``lift'' normal ASP-parameters to general disjunctive programs
as follows.

\begin{DEF}\label{def:lift}
  For a normal ASP-parameter $p$ we define the ASP-parameter
  $\lift{p}$ by setting, for each disjunctive program~$P$,
  $\lift{p}(P)$ as the minimum $\Card{X}+p(P-X)$ over all
  inclusion-minimal \delBds{\Normal}s~$X$ of $P$.
\end{DEF}

The next lemma shows that this definition is compatible with
\delBds{\CCC}s if $\CCC\subseteq \Normal$.  In other words, if $\CCC$
is a class of normal programs, then we can divide the task of finding
a \delBds{\CCC} for a program~$P$ into two parts: (i)~to find a
\delBds{\Normal}~$X$, and (ii)~to find a \delBds{\CCC} of $P-X$.

\begin{LEM}[Self Lifting]\label{lem:selflift}
  Let $\CCC$ be a class of normal programs. Then
  $\dbd_\CCC=\lift{(\down{ \dbd_\CCC })}$.
\end{LEM}
\begin{proof}
  Let $\CCC$ be a class of normal programs, and $P$ a program.
  Let $X$ be a \delBds{\CCC} of $P$ of size~$\dbd_\CCC(P)$. Thus $P-X
  \in \CCC \subseteq \Normal$. Hence $X$ is a \delBds{\Normal} of $P$.
  We select an inclusion-minimal subset~$X'$ of $X$ that is still a
  \delBds{\Normal} of $P$ (say, by starting with $X'=X$, and then
  looping over all the elements $x$ of $X$, and if $X'-x$ is still a
  \delBds{\CCC}, then setting $X':=X'-x$.)  What we end up with is an
  inclusion-minimal \delBds{\Normal}~$X'$ of $P$ of size at
  most~$\dbd_\CCC(P)$. Let $P'=P-X'$ and $X''=X-X'$. $P'$ is a normal
  program. Since $P'-X''=P-X$, it follows that $P'-X''\in \CCC$. Hence
  $X''$ is a \delBds{\CCC} of $P$.  Thus, by the definition of
  $\lift{db_\CCC}$, we have that $\lift{db_\CCC}(P) \leq \Card{X'} +
  \Card{X''} = \dbd_\CCC(P)$.

  Conversely, let $\lift{\dbd_\CCC}(P)=k$.  Hence there is a
  \delBds{\Normal}~$X'$ of $P$ such that
  $\Card{X'}+\dbd_\CCC(P-X')=k$. Let $P'=P-X'$. Since
  $\dbd_\CCC(P')\leq k-\Card{X'}$, it follows that $P'$ has a
  \delBds{\CCC}~$X''$ of size~$k-\Card{X'}$. We put $X=X' \cup X''$
  and observe that $P-X=P'-X'' \in \CCC$. Hence $X$ is a \delBds{\CCC}
  of $P$. Since $\dbd_\CCC(P)\leq \Card{X}\leq \Card{X'} +
  \Card{X''}\leq \lift{\dbd_\CCC}(P)\leq k$, the lemma follows.
\end{proof}

\begin{EX}\label{ex:lifting_parameters}
  Consider the program~$P$ of Example~\ref{ex:running} and let
  $\parm{\#neg}(P)$ denote the number of atoms that appear in negative
  rule bodies of a normal program (we will discuss this parameter in
  more detail in Section~\ref{sec:horn-based-parameters}). 

  We determine $\lift{\parm{\#neg}}(P)=2$ by the following
  observations: The set~$X_1 = \{c\}$ is a \delBds{\Normal} of $P$
  since $P-X_1=\SB d \leftarrow a, e \rsep a \leftarrow d, \pnot b
  \rsep e \leftarrow f \rsep f \leftarrow d \rsep \leftarrow f, e,
  \pnot b \rsep \leftarrow d \rsep b \rsep f \SE$ belongs to the class
  \Normal. The set~$X_2=\{e\}$ is a \delBds{\Normal} of~$P$ since
  $P-X_2 = \SB d \leftarrow a \rsep a \leftarrow d, \pnot b, \pnot c
  \rsep c \leftarrow f \rsep f \leftarrow d, c \rsep c \leftarrow f,
  \pnot b \rsep c \leftarrow d \rsep b \leftarrow c \rsep f \SE$
  belongs to the class~\Normal.  Observe that $X_1$ and $X_2$ are the
  only inclusion-minimal \delBds{\Normal}s of the program~$P$. We
  obtain $\lift{\parm{\#neg}}(P,X_1)=2$ since
  $\parm{\#neg}(P-X_1)=1$. We have $\lift{\parm{\#neg}}(P,X_2)=3$
  since $\parm{\#neg}(P-X_2)=2$. Thus $\lift{\parm{\#neg}}(P)=2$.
\end{EX}

For every ASP-parameter~$p$ we consider the following problem.
\pproblem{$\Bound[p]$}{a program~$P$ and an integer~$k$.}{the
  integer~$k$.}{Decide whether $p(P) \leq k$ holds.}

For a problem~$L\in\, \AspFull$ and an ASP-parameter~$p$ we write
$L[p]$ to denote the problem~$L$ parameterized by~$p$. That is, the
instance of the problem is augmented with an integer~$k$, the
parameter, and for the input program~$P$ it holds that $p(P)\leq k$.
Moreover, we write $\pnormal{L[p]}$ to denote the restriction of
$L[p]$ where instances are restricted to normal programs~$P$.
Similarly, $\pnormal{\Bound[p]}$ is the restriction of $\Bound[p]$ to
normal programs. For all the problems~$\pnormal{L[p]}$, $p$ only needs
to be a normal ASP-parameter.
%
% Furthermore, we write $\pnormal{\Bound[p]}$ for the parameterized
% problem that takes as input a normal program~$P$ and an integer~$k$
% (the parameter), and asks whether $p(P)\leq k$.
%
% For uniformity
%

Next we state the main result of this section.

\begin{THE}[Lifting]\label{the:lift}
  Let $p$ be a normal ASP-parameter such that $\pnormal{\Bound[p]}$
  and $\pnormal{\AspEnum[p]}$ are fixed-parameter tractable. Then for
  all~$L \in\,\AspFull$ the problem~$L[\lift{p}]$ is fixed-parameter
  tractable.
\end{THE}

We need some definitions and auxiliary results to establish the
theorem.

% In order to determine all answer sets of the program~$P$ we
% determine the answer sets of the partial truth assignment reducts by
% applying the specific algorithm by Ben-Eliyahu. We have the answer
% sets $\stableset(P_{c})=\{\}$ and $\stableset(P_{\bar c})=\{\}$. The
% answer sets extend
% \end{EX}

\begin{DEF}
  Let $P$ be a disjunctive program. The \emph{head dependency
    graph}~$H_P$ of the program~$P$ is the graph which has as vertices
  the atoms of~$P$ and an edge between any two distinct atoms if they
  appear together in the head of a rule of~$P$.
\end{DEF}

\begin{LEM}\label{lem:hornbdsdet-vc}
  Let $P$ be a disjunctive program.  A set~$X\subseteq \at(P)$ is a
  \delBds{\Normal} of~$P$ if and only if $X$ is a vertex cover of the
  head dependency graph~$H_P$.
\end{LEM}
\begin{proof}
  Let $X$ be a \delBds{\Normal} of $P$. Consider an edge~$uv$ of
  $H_P$, then there is a rule~$r\in P$ with $u,v\in H(r)$ and $u\neq
  v$.  Since $X$ is a \delBds{\Normal} of $P$, we have $\{u,v\}\cap X
  \neq \emptyset$. We conclude that $X$ is a vertex cover of $H_P$.

  Conversely, assume that $X$ is a vertex cover of $H_P$. We show that
  $X$ is a \delBds{\Normal} of $P$. Assume to the contrary, that $P-X$
  contains a rule~$r$ whose head contains two variables~$u,
  v$. Consequently, there is an edge~$uv$ in $H_P$ such that $\{u,
  v\}\cap X =\emptyset$, contradicting the assumption that $X$ is a
  vertex cover.
\end{proof}

\begin{LEM}\label{lem:list-vc}
  Let $G=(V,E)$ be a graph, $n=\Card{E}$, and let $k$ be a
  non-negative integer. $G$ has at most~$2^k$ inclusion-minimal vertex
  covers of size at most~$k$, and we can list all such vertex covers
  in time~$\BigO{2^k n}$.
\end{LEM}
\begin{proof}
  We build a binary search tree~$T$ of depth at most~$k$ where each
  node~$t$ of $T$ is labeled with a set~$S_t$. We build the tree
  recursively, starting with the root~$r$ with label~$S_r=\emptyset$.
  If $S_t$ is a vertex cover of $G$ we stop the current branch, and
  $t$ becomes a ``success'' leaf of $T$.  If $t$ is of distance~$k$
  from the root and $S_t$ is not a vertex cover of $G$, then we also
  stop the current branch, and $t$ becomes a ``failure'' leaf of $T$.
  It remains to consider the case where $S_t$ is not a vertex cover
  and $t$ is of distance smaller than $k$ from the root. We pick an
  edge~$uv\in E$ such that $u,v\notin S_t$ (such edge exists,
  otherwise $S_t$ would be a vertex cover) and add two
  children~$t,t''$ to $t$ with labels~$S_{t'}=S_t\cup \{u\}$ and
  $S_{t'}=S_t\cup \{v\}$. It is easy to see that for every
  inclusion-minimal vertex cover~$S$ of $G$ of size at most~$k$ there
  is a success leaf~$t$ with $S_t=S$. Since $T$ has $\BigO{2^k}$
  nodes, the lemma follows.
\end{proof}

From Lemmas~\ref{lem:hornbdsdet-vc}  and \ref{lem:list-vc} we
immediately obtain the next result.

\begin{PRO}\label{pro:norm-list}
  Every disjunctive program of input size~$n$ has at most~$2^k$
  inclusion-minimal \delBds{\Normal}s of size at most~$k$, and all
  these backdoors can be enumerated in time~$\BigO{2^kn}$.
\end{PRO}

\begin{proof}[Proof of Proposition~\ref{the:lift}]
  Let $p$ be a normal ASP-parameter such that $\pnormal{\AspEnum[p]}$
  and $\pnormal{\Bound[p]}$ are fixed-parameter tractable.  Let $P$ be
  a given disjunctive program of input size~$n$ and $k$ an integer
  such that $\lift{p}(P)\leq k$. In the following, when we say some
  task is solvable in ``\emph{fpt-time}'', we mean that it can be
  solved in time~$O(f(k)\,n^c)$ for some computable function~$f$ and a
  constant~$c$.

  By Proposition~\ref{pro:norm-list} we can enumerate all
  inclusion-minimal \delBds{\Normal}s of size at most~$k$ in
  time~$\BigO{2^kn}$. We can check whether $p(P-X)\leq k-\Card{X}$ for
  each such backdoor~$X$ in fpt-time since $\pnormal{\Bound[p]}$ is
  fixed-parameter tractable by assumption.  Since $\lift{p}(P)\leq k$,
  there is at least one such set~$X$ where the check succeeds.

  We pick such set~$X$ and compute $\stableset(P,X)$ in fpt-time. That
  this is possible can be seen as follows. Obviously, for each truth
  assignment~$\tau\in \ta{X}$ the program~$P_\tau$ is normal since
  $P-X$ is normal, and clearly $p(P_\tau)\leq p(P-X)\leq k$ by
  Definition~\ref{def:asp-parm}.  We can compute $\stableset(P_\tau)$
  in fpt-time since $\pnormal{\AspEnum[p]}$ is fixed-parameter
  tractable by assumption. Since there are at most~$2^k$ such
  programs~$P_\tau$, we can indeed compute the set~$\stableset(P,X)$
  in fpt-time.

  By Lemma~\ref{lem:subset} we have $\stableset(P)\subseteq
  \stableset(P,X)$, hence it remains to check for each~$M\in
  \stableset(P,X)$ whether it gives rise to an answer set of $P$.
  Since $X$ is a \delBds{\Normal} of $P$, and since one easily
  observes that $\Normal$ is hereditary, it follows by
  Lemma~\ref{lem:rule-induced} that $X$ is a \strongBds{\Normal} of
  $P$. Hence Lemma~\ref{lem:mincheck} applies, and we can decide
  whether $M \in \stableset(P)$ in time~$\BigO{2^k n}$. Hence we can
  determine the set~$\stableset(P)$ in fpt-time. Once we know the
  set~$\stableset(P)$, we obtain for every problem~$L \in\, \AspFull$
  that $L[\lift{p}]$ is fixed-parameter tractable.
\end{proof}

% \begin{proof}[of Theorem~\ref{the:lift}]
%   Let $p$ be a normal ASP-parameter such that
%   $\pnormal{\AspEnum(p)}$ is fixed-parameter tractable.  Then all
%   the problems Let $n$ denote the number of atoms in $P$.  We use
%   Proposition~\ref{pro:norm-list} to list all deletion
%   $\Normal$-backdoors of size at most $k$ in time $\BigO{2^kn}$. For
%   each such backdoor $X$ we check whether $p(P-X)\leq k-\Card{X}$
%   (this can be done in fpt-time, since $\pnormal{\Bound(p)}$ is
%   fixed-parameter tractable).

%   If it is not, we reject $X$ and continue. If $p(P-X)\leq
%   k-\Card{X}$ then we solve $\pnormal{\AspEnum(p)}$ for $P_\tau$ for
%   each $\tau\in \ta(X)$, which is again fixed-parameter tractable by
%   assumption.  We now proceed exactly as in the proof of
%   Theorem~\ref{the:evaluation} to solve all the reasoning problems.
%   \stefan{perhaps we need to be more detailed}
% \end{proof}

\begin{EX}
  Consider the program~$P$ of Example~\ref{ex:running} with the 
%  the problem
%   \pname{Enumeration} when parameterized in the number of atoms that
%   appear in the negative rule bodies, the fixed-parameter tractable
%  algorithm that solves \pname{Enumeration} by~\citex{Ben-Eliyahu96}
%  and 
  the \delBds{\Normal}~$X_1 = \{c\}$ from
  Example~\ref{ex:lifting_parameters}.  We want to enumerate all
  answer sets of $P$.  We obtain with Ben-Eliyahu's algorithm
  \cite{Ben-Eliyahu96} the sets~$\stableset(P_{\bar c})=\{ \{e,f\}\}$
  and $\stableset(P_{c})=\{\{b,f\}\}$, and so we get the set
  $\stableset(P,X)=\{ \{e,f\}, \{b,c,f\}\}$ of answer set
  candidates. By means of the algorithm that solves the problem
  \BdCheck{\CCC} (see Lemma~\ref{lem:mincheck}) we observe that
  $\{b,c,f\}$ is the only answer set of~$P$.
\end{EX}

\section{Theoretical Comparison of \ASP-Parameters}
\label{sec:comparision}

In this section we compare several \ASP-parameters in terms of their
\emph{generality}. Let $p$ and $q$ be \ASP-parameters.  We say that
$p$ \emph{dominates} $q$ (in symbols $p \preceq q$) if there is a
function $f$ such that $p(P)\leq f(q(P))$ holds for all
programs~$P$. The parameter~$p$ \emph{strictly dominates} $q$ (in
symbols $p \prec q$) if $p \preceq q$ but not $q \preceq p$, and $p$
and $q$ are \emph{incomparable} (in symbols $p \bowtie q$) if neither
$p \preceq q$ nor $q \preceq p$. For simplicity we only consider
programs that contain no tautological rules. It is easy to adapt the
results to the more general case where tautological rules are allowed.

\begin{OBS}
  Let $p$ and $q$ be ASP-parameters and $L\in \AspFull$. If $p$
  dominates $q$ and $L[p]\in \FPT$, then also $L[q]\in \FPT$.
\end{OBS}

\begin{OBS}\label{obs:compare-lift}
  Let $p$ and $q$ be normal ASP-parameters and $\circ \in
  \{\preceq,\prec,\bowtie\}$. Then $p\circ q$ if and only if $\lift{p}
  \circ \lift{q}$.
\end{OBS}

%\begin{samepage}
In the following we define various auxiliary programs which we will
use as examples, to separate the parameters from each other and
establish incomparability or strictness results.  

\begin{EX}\label{ex:comp_progs}
  Let $m$ and $n$ be some large integers. We define the following
  programs:
\begin{align*}
  P^n_1  := {} & \SB a \leftarrow \pnot b_1, \ldots, \pnot b_n \SE,\\
  P^n_2  := {}& \SB a_i \leftarrow \pnot b \SM 1 \leq i \leq n \SE,\displaybreak[1]\\
  P^n_{31} :={} & \SB b_i \leftarrow \pnot a \rsep a \leftarrow \pnot
  b_i \SM 1 \leq i \leq n \SE,\displaybreak[1]\\
  P^n_{32} := {}& \SB b_i \leftarrow a \rsep a \leftarrow b_i
  \SM 1 \leq i \leq n \SE,\displaybreak[1]\\
  P^n_{33} := {}& P^n_{31} \cup \SB a \leftarrow d_1 \rsep d_i
  \leftarrow d_{i+1} \SM 1 \leq i \leq n \SE \cup \SB c_i \leftarrow
  b_i \rsep d_i \leftarrow c_i
  \rsep d_i \leftarrow b_i\SM 1 \leq i \leq n \SE,\\
  P^n_{34} := {}& P^n_{33} \cup \SB d_i \leftarrow \pnot b_i
  \SM 1 \leq i \leq n \SE,\displaybreak[1]\\
%  P^n_{35} := {}& P^n_{33} \setminus \SB a \leftarrow \pnot b_i \SM 1
%  \leq i \leq n \SE,\displaybreak[1]\\
  P^n_{35} := {}& P^n_{35} \setminus \SB a \leftarrow \pnot b_i\rsep
  b_i \leftarrow \pnot a \SM 1 \leq i \leq n \SE \cup \SB a_0
  \leftarrow \pnot a\SE \cup
  \SB b_i \leftarrow a_0 \SM 1 \leq i \leq n \SE ,\displaybreak[1]\\
  P^n_{4} := {}& \SB c_i \leftarrow \pnot a_i\rsep c_i \leftarrow
  b_i\rsep b_i \leftarrow \pnot a_i \rsep a_i \leftarrow e_i \rsep e_i
  \leftarrow d_i \rsep  d_i \leftarrow a_i \SM 1 \leq i \leq n \SE,\displaybreak[1]\\
  P^n_{51} := {}& \SB b_i \leftarrow \pnot a_i \rsep a_i \leftarrow
  \pnot b_i \SM 1 \leq i \leq n \SE,\\
  P^n_{52} := {}& \SB b_i \leftarrow a_i \rsep
  a_i \leftarrow \pnot b_i \SM 1 \leq i \leq n \SE,\\
  P^n_{53} := {}& \SB b_i \leftarrow a_i \rsep
  a_i \leftarrow b_i \SM 1 \leq i \leq n \SE,\displaybreak[1]\\
  P^n_6 := {}& \SB a \leftarrow b_{1}, \ldots, b_{n},c_i \SM 1
  \leq i \leq n \SE,\\
  P^n_7 := {}& \SB a_j
  \leftarrow a_i \SM 1 \leq i < j \leq n \SE\displaybreak[1],\\
  P^{m,n}_8 :={}& \SB b \leftarrow a_1, \ldots, a_m \SE \cup \SB c_i
  \leftarrow c_{i+1} \SM 1 \leq i \leq n \SE \cup \SB
  c_{n+1} \leftarrow c_1 \SE,\\
  P^n_9 := {}& \SB a_2 \leftarrow \pnot a_1 \rsep a_3 \leftarrow \pnot
  a_2 \SE \cup \SB b_i \leftarrow a_3 \rsep a_1 \leftarrow b_i \SM 1
  \leq i \leq n
  \SE, \text{ and }\\
  P^n_{11} := {}& \SB a_i \por b \leftarrow c \rsep c \leftarrow b
  \rsep b \leftarrow a_i \SM 1 \leq i \leq n \SE.
\end{align*}
\end{EX}
% \end{samepage}

\subsection{ASP-Parameters Based on Backdoor Size}
Backdoor-based ASP-parameters can be related to each other in terms of
their underlying base classes. We just need a very weak assumption
which holds for all target classes considered in the paper:

\begin{PRO}\label{obs:bdsubclasses}
  Let $\CCC,\CCC'$ be classes of programs that are closed under the
  union of disjoint copies\footnote{A class~$\CCC$ of programs is
    \emph{closed under the union of disjoint copies} if for every $P
    \in \CCC$ also $P \cup P' \in \CCC$ where $P'$ is a copy of $P$
    with $\at(P) \cap \at(P') = \emptyset$.}. If $\CCC \subseteq
  \CCC'$ then $\dbd_{\CCC'} \preceq \dbd_{\CCC}$ and $\sbd_{\CCC'}
  \preceq \sbd_{\CCC}$, even $\dbd_{\CCC'}(P)\leq \dbd_{\CCC}(P)$ and
  $\sbd_{\CCC'}(P) \leq \sbd_{\CCC}(P)$ for every program~$P$. If
  $\CCC'\setminus \CCC$ contains a program with at least one atom,
  then $\CCC\subseteq \CCC'$ implies $\dbd_{\CCC'} \prec \dbd_{\CCC}$
  and $\sbd_{\CCC'} \prec \sbd_{\CCC}$.
\end{PRO}
\begin{proof}
  The first statement is obvious. For the second statement, let $P \in
  \CCC'\setminus \CCC$ with $\Card{at(P)}\geq 1$.  We construct the
  program~$P^n$ consisting of $n$ disjoint copies of $P$ and observe
  that $P^n \in \CCC'$ but $\dbd_{\CCC}(P^n),\sbd_{\CCC}(P^n) \geq n$.
\end{proof}

Hence the relationships between target classes as stated in
Observation~\ref{obs:lattice} carry over to the corresponding
backdoor-based ASP-parameters that is, if $\CCC \subseteq \CCC'$ then
a smallest \strongdelBds{\CCC'} is at most the size of a smallest
\strongdelBds{\CCC}.

Accoring to Lemma~\ref{lem:rule-induced} every \delBds{\CCC} is a
\strongBds{\CCC} only if $\CCC$ is hereditary, hence it also holds for
smallest backdoors and we immediately get from the definitions:
\begin{OBS}\label{obs:ri_classes}
  If $\CCC$ is hereditary, then $\sbd_\CCC$ dominates $\dbd_\CCC$.
\end{OBS}

According to Lemma~\ref{lem:Horn-strong-deletion} every
\strongBds{\Horn} of a program is a \delBds{\Horn} and vice versa and
we observe:
\begin{OBS}
  $\sbd_\Horn=\dbd_\Horn$.
\end{OBS}

\begin{OBS}\label{ex:bdclasses}
  We make the following observations about programs from
  Example~\ref{ex:comp_progs}.
\begin{enumerate}
\item Consider program~$P^n_{31}$ and $P^n_{32}$ and let $P\in
  \{P^n_{31}, P^n_{32}\}$. Since $P - \{a\}$ is Horn and contains no
  cycle and no directed cycle, we obtain $\dbd_\Horn(P) \leq 1$,
  $\dbd_\C(P) \leq 1$, and $\dbd_\DC(P) \leq 1$. According to
  Observation~\ref{obs:bdsubclasses} we have $\dbd_\CCC (P^n_{31})
  \leq 1$ and $\dbd_\CCC (P^n_{32}) \leq 1$ where $\CCC \in \{\Horn\}
  \cup \Acyc $.
\item Consider program~$P^n_{33}$. Since $P^n_{33}-\{a\}$ is Horn and
  contains no directed cycle and no bad cycle, we obtain
  $\dbd_\Horn(P^n_{33})=0$, $\dbd_\DC(P^n_{33})\leq 1$, and
  $\dbd_\BC(P^n_{33}) \leq 1$. According to
  Observation~\ref{obs:bdsubclasses} we have $\dbd_\CCC (P^n_{33})
  \leq 1$ where $\CCC \in \{\Horn, \BC, \BEC \} \cup \DAcyc$.
\item Consider program~$P^n_{34}$. Since $P^n_{34}-\{a\}$ contains no
  even cycle, $\dbd_\EC(P^n_{34}) \leq 1$.
\item Consider program~$P^n_4$. The negation dependency graph of
  $P^n_{4}$ contains $2n$ disjoint paths~$a_ib_i$ and $a_ic_i$, thus
  smallest \delBds{\Horn} are of size at least~$n$. $P^n_{4}$ contains
  $n$ disjoint bad cycles, $n$ directed cycles of length at least~$3$,
  and $n$ directed even cycles. Hence smallest \delBds{\CCC}s are of
  size at least~$n$ and thus $\dbd_{\CCC}(P^n_{4}) \geq n$ where $\CCC
  \in \{\Horn$, \C, \BC, \DC, \DCTWO, \EC, \BEC, $\DEC\}$.
\item Consider program~$P^n_{51}$. The negation dependency graph of
  $P^n_{51}$ contains $n$ disjoint paths and thus
  $\dbd_\Horn(P^n_{51}) = n$. $P^n_{51}$ contains $n$ disjoint
  directed bad even cycles and thus $\dbd_\DBEC(P^n_{51}) =
  n$. According to Observation~\ref{obs:bdsubclasses} we obtain
  $\dbd_\CCC(P^n_{51}) \geq n$ where $\CCC \in \{\Horn\} \cup \Acyc$.
\item Consider program~$P^n_{52}$. Since $P^n_{52}$ contains $n$
  disjoint directed bad cycles, $\dbd_\DBC(P^n_{52}) = n$.
\item Consider program~$P^n_{53}$. Since $P^n_{53}$ contains $n$
  disjoint even cycles, $n$ disjoint directed cycles of length at
  least~$3$, and $n$ disjoint directed even cycles, we obtain by
  Observation~\ref{obs:bdsubclasses} $\dbd_\CCC(P^n_{53}) \geq n$
  where $\CCC \in \{\C,$ \DC, \DCTWO, \EC, $\DEC\}$.
\item Consider program~$P^n_6$. Since $P^n_6$ is Horn and contains no
  cycle and no directed cycle, $\dbd_\Horn(P^n_6) = \dbd_\C(P^n_6) =
  \dbd_\DC(P^n_6) = 0$. According to
  Observation~\ref{obs:bdsubclasses} we have $\dbd_\CCC(P^n_6) = 0$
  where $\CCC \in \{\Horn\} \cup \Acyc$.
\item Consider program~$P^n_7$. Since $P^n_7$ is Horn and contains no
  bad cycle and no directed cycle, $\dbd_\Horn(P^n_7) =
  \dbd_\BC(P^n_7) = \dbd_\DC(P^n_7) = 0$. According to
  Observation~\ref{obs:bdsubclasses} we have $\dbd_\CCC(P^n_7) = 0$
  where $\CCC \in \{\Horn, \BC, \BEC\} \cup \DAcyc$.
\item Consider program~$P^{m,n}_8$. Since $P^{m,n}_8$ is Horn and
  $P^{m,n}_8 - \{c_1\}$ contains no cycle and no directed cycle, we
  obtain $\dbd_\Horn(P^{m,n}_8) = 0$, $\dbd_\C(P^{m,n}_8) \leq 1$,
  $\dbd_\DC(P^{m,n}_8) \leq 1$. According to
  Observation~\ref{obs:bdsubclasses} we have $\dbd_\CCC(P^{m,n}_8) \leq
  1$ where $\CCC \in \{\Horn\} \cup \Acyc$.
\item Consider program~$P^n_9$. Since $P^n_9 - \{a_2\}$ is Horn and
  $P^n_9 - \{a_1\}$ contains no cycle and no directed cycle, we have
  $\dbd_\Horn(P^n_9) \leq 1$, $\dbd_\C(P^n_9) \leq 1$, and
  $\dbd_\DC(P^n_9) \leq 1$. According to
  Observation~\ref{obs:bdsubclasses} we have $\dbd_\CCC(P^n_9) \leq 1$
  where $\CCC \in \{\Horn\} \cup \Acyc$.
\item Consider program~$P^n_{11}$ and let $X:=\{b\}$. Since $P^n_{11}-
  X$ is normal, $X$ is a \delBds{\Normal} of $P^n_{11}$.  Observe that
  $X$ is the only inclusion-minimal \delBds{\Normal} of $P^n_{11}$.
  Since $P^n_{11} - X$ is Horn, $\dbd_\Horn(P^n_{11} - X) = 0$. Since
  $P^n_{11} - X$ contains no cycle, no even cycle, and no directed
  cycle, $\dbd_\CCC(P^n_{11} - X) = 0$ where $\CCC \in
  \Acyc$. Consequently, $\lift{\dbd_\CCC}(P^n_{11}) = \Card{X} +
  \dbd_\CCC(P^n_{11} - X) = 1$ where $\CCC \in \{\Horn\} \cup \Acyc$.
\end{enumerate}
%\vspace{-5ex}
\end{OBS}

%   Consider program~$P^{m,n}_8$. The program~$P^{m,n}_8$ is Horn and
%   thus and $\dbd_\Horn(P^{m,n}_8) = 0$. Let $X = \{ c_1 \}$, then
%   $P^{m,n}_8 - X$ contains no cycle and no directed cycle. Hence
%
%   $\dbd_\CCC(P^{m,n}_8) \leq 1$ for $\CCC \in \{ \Horn \} \cup \Acyc$.
%   Considert program~$P^n_{32}$. The program~$P^n_{32}$ is Horn and
%   thus $\dbd_\Horn(P^n_{32}) = 0$. Let $X =\{a\}$, then $P^n_{32} - X
%   \in \C$ and $P^n_{32} - X \in \DC$. Hence $X$ is a \delBds{\CCC} of
%   $P^n_{32}$. Consequently, $\dbd_\CCC(P^n_{32}) \leq 1$ for $\CCC \in
%   \{\Horn\} \cup \Acyc$.
%
%   Next we consider the previously studied parameters. Let $P:=P^n_9$
%   and $X_1=\{a_1\}$. We observe $P - X_1 \in \C$ and $P - X \in
%   \DC$. Then for $\CCC \in \Acyc$ we have $X_1$ is a \delBds{\CCC} and
%   hence $\dbd_\CCC(P^n_9) \leq 1$. Let $X_2=\{a_1,a_2\}$. Since $P - X
%   \in \Horn$, $X$ is a \delBds{\Horn}. Thus $\dbd_\Horn(P^n_9) \leq
%   2$.
%
%   Let $P:=P^n_{52}$. Since $P^n_{52}$ contains $n$ disjoint bad
%   cycles, $\dbd_\DBC(P^n_{52}) = n$.

\begin{figure}[ht]
\centering
\begin{tikzpicture}[latex-,level 1/.style={sibling distance=12em},level distance=6em,font=\small]
  \node[](root){} 
    child[xshift=-5em]{ node(head) {$\parm{\#headCycles}^{\dagger}$} edge from parent[draw=none] 
      child{ node(even) {$\dbd_\DBEC$} 
      child{ node(dbec) {$\parm{\#badEvenCycles}$}
        edge from parent node[left] {\ref{pro:zhao}}
      }
        child{node(dbc){$\dbd_\DBC$}
          child{ node(horn){$\dbd_\Horn$}
            child{ node(neg){$\parm{\#neg}$}
              edge from parent node[left] {\ref{lem:del_neg}}
            }
            edge from parent node[left] {\ref{obs:bdsubclasses}}
          }
          child{ node(lstr){$\parm{lstr}$}
            child{ node(nonH){$\parm{\#non-Horn}$}
            edge from parent node[right] {\ref{pro:lstr}}
            }
            edge from parent node[right] {\ref{pro:lstr}}
          }
          child{ node(wfw){$\parm{wfw}$}
            child{ node(c){$\dbd_{\C}$}
              edge from parent node[left] {\ref{pro:del_wfw}}
            }
            edge from parent node[right] {\ref{pro:del_wfw}}
          }
          edge from parent node[right] {\ref{obs:bdsubclasses}}
        }
      edge from parent node[left] {\ref{pro:hcf1}}
      }
      child{ node(posc) {$\parm{\#posCycles}$}
        edge from parent node[right] {\ref{pro:hcf1}}
      }
    }
    child{node(empty){} edge from parent[draw=none]
    child{ node(deptw){$\parm{deptw}^{\ddagger}$} edge from parent[draw=none] 
        child{ node(inctw){$\parm{inctw}$}
            child{ node(cluster){$\parm{cluster}$}
              child{ node(cyclecut){$\parm{cyclecut}$}
                edge from parent node[right] {\ref{pro:cluster}}
              }
              edge from parent node[right] {\ref{pro:inctw}}
            }
            edge from parent node[right] {\ref{pro:deptw}}
          }
        }}
        ;
        \path(deptw) edge node[left] {\ref{pro:hcf1}} (c);
        \path(lstr) edge node[near start, yshift=1em] {\ref{pro:lstr}} (neg);
        \path(horn) edge node [near start,yshift=1em] {\ref{lem:del_neg}} (nonH);
      \end{tikzpicture}
      \caption{Domination Lattice (relationship between normal
        \ASP-parameters): An arrow from $p$ to $p'$ indicates that
        $p'$ strictly dominates $p$. ${}^\dagger$:~\parm{\#headCycles}
        is not strictly more general when we apply lifting
        (Observation~\ref{obs:headCycles}).
        ${}^\ddagger$:~$\parm{deptw}$ does not yield tractability
        (Proposition~\ref{pro:deptw}). A label~$i$ of an edge
        indicates that Proposition~$i$ establishes the result.}
\label{fig:diagram}
\end{figure}
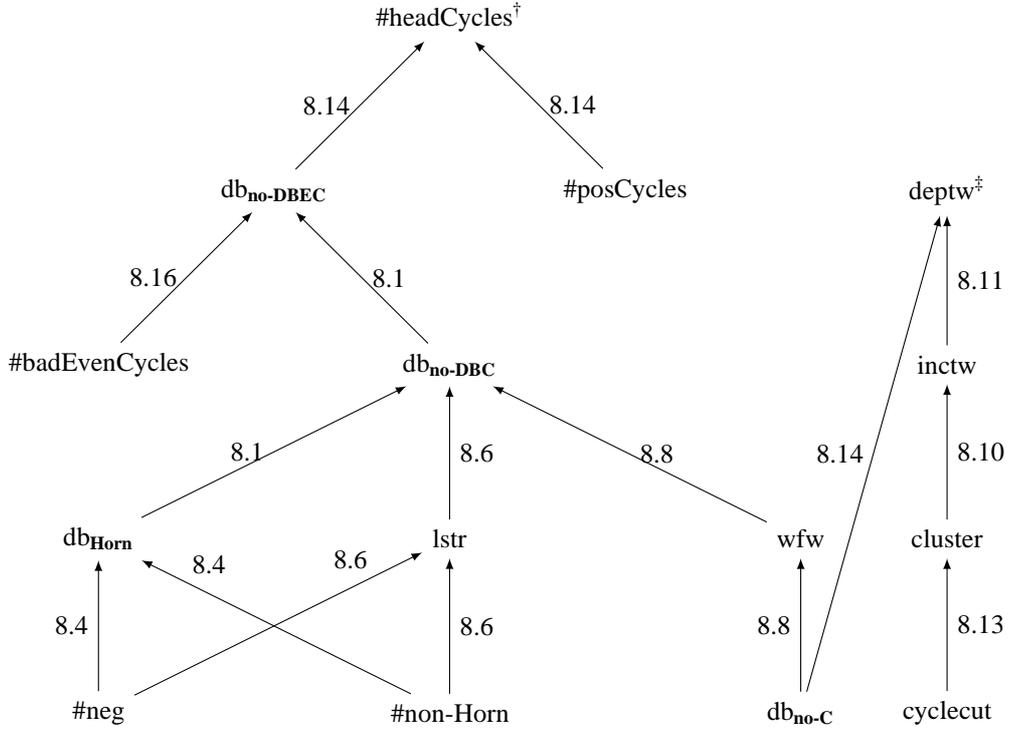

\subsection{\ASP-Parameters Based on the Distance from Horn}\label{sec:horn-based-parameters}
Our backdoor-based ASP-parameter~$\dbd_\Horn$ can be considered as a
parameter that measures the distance of a program from being a Horn
program. In the literature some normal ASP-parameters have been
proposed, that also can be considered as distance measures from
Horn. In this section we compare them with $\dbd_\Horn$.  Since the
ASP-parameters considered in the literature are normal, we compare the
parameters for normal programs only. However, in view of
Observation~\ref{obs:compare-lift} the results also hold for the
lifted parameters to disjunctive programs.

\begin{DEF}[\citex{Ben-Eliyahu96}]
  Let $P$ be a normal program. Then 
  \begin{align*}
    \parm{\#neg}(P):={}&\Card{\SB a\in
      \at(P)\SM a\in B^-(r) \text{ for some rule }r\in P \SE},\\
    \parm{\#non-Horn}(P):={}&\Card{\SB r\in P \SM r \text{  is not Horn
      }\SE}.
  \end{align*}
\end{DEF}

\begin{PRO}[\citex{Ben-Eliyahu96}]
  For each~$L\in \AspFull$, $\pnormal{L[\parm{\#neg}]} \in \FPT$ and
  $\pnormal{L[\parm{\#non-Horn}]}\in \FPT$.
\end{PRO}

Since $\pnormal{\Bound[p]}$ for $p \in
\{\parm{\#neg}, \parm{\#non-Horn}\}$ is clearly solvable in polynomial
time and thus fixed-parameter tractable, we can use the Lifting
Theorem (Theorem~\ref{the:lift}) to obtain the following result.

\begin{COR}
  For each~$L\in \AspFull$, $L[\lift{\parm{\#neg}}] \in \FPT$ and
  $L[\lift{\parm{\#non-Horn}}] \in \FPT$.
\end{COR}

\begin{OBS}\label{ex:num-neg}
  We make the following observations about programs from
  Example~\ref{ex:comp_progs}.
%  We continue Example~\ref{ex:comp_progs} and observe:
%~
%
\begin{enumerate}
\item Consider program~$P^n_1$ which contains $n$ atoms that occur in
  $B^-(r)$ for some rule~$r \in P$ and exactly one non-Horn rule, so
  $\parm{\#neg}(P^n_1) = n$ and $\parm{\#non-Horn}(P^n_1) = 1$.
\item Consider program~$P^n_2$ which contains only the atom~$b$ that
  occurs in $B^-(r)$ for some rule~$r \in P^n_2$ and $n$ non-Horn
  rules, so $\parm{\#neg}(P^n_2) = 1$ and $\parm{\#non-Horn} ( P^n_2)
  = n$.
\item Consider program~$P^n_{31}$ which contains for $1 \leq i \leq n$
  the atoms~$a$, $b_i$ that occur in $B^-(r)$ for some rule~$r \in
  P^n_{31}$ and the non-Horn rules~$b_i \leftarrow \pnot a$ and $a
  \leftarrow \pnot b_i$, hence $\parm{\#neg} (P^n_{31}) = n+1$ and
  $\parm{\#non-Horn}(P^n_{31}) = 2n$.
\item Consider program~$P^n_{32}$ which is Horn. Thus
  $\parm{\#neg}(P^n_{32}) = \parm{\#non-Horn}(P^n_{32}) = 0$.
\item Consider program~$P^n_{35}$ which contains only the atom $a$
  that occurs in $B^-(r)$ for some rule~$r \in P^n_{35}$ and exactly
  one non-Horn rule, so $\parm{\#neg}(P^n_{35})
  = \parm{\#non-Horn}(P^n_{35}) = 1$.
\item Consider program~$P^n_4$ which contains for $1 \leq i \leq n$
  the atoms~$a_i$ that occur in $B^-(r)$ for some rule~$r \in P^n_4$
  and the non-Horn rules~$b_i \leftarrow \pnot a_i$ and $c_i
  \leftarrow \pnot a_i$, thus $\parm{\#neg}(P^n_4) = n$ and
  $\parm{\#non-Horn}(P^n_4) = 2n$.
\item Consider program~$P^n_{51}$ which contains for $1 \leq i \leq n$
  the atoms~$a_i$ and $b_i$ that occur in $B^-(r)$ for some rule~$r
  \in P$ and the non-Horn rules~$b_i \leftarrow \pnot a_i$ and $a_i
  \leftarrow \pnot b_i$, hence $\parm{\#neg} (P^n_{51})
  = \parm{\#non-Horn}(P^n_{51}) = 2n$.
\item Consider the program~$P^n_{52}$ which contains the atoms~$b_i$
  that occur in $B^-(r)$ for some rule~$r \in P^n_{52}$ and the
  non-Horn rules~$a_i \leftarrow \pnot b_i$, hence
  $\parm{\#neg}(P^n_{52}) = \parm{\#non-Horn}(P^n_{52}) = n$.
\item Consider programs~$P^n_{53}$, $P^n_6$, $P^n_7$, and $P^{m,n}_8$
  which are Horn. Thus $\parm{\#neg}(P^n_{53})
  = \parm{\#non-Horn}(P^n_{53}) = \parm{\#neg}(P^n_{6})
  = \parm{\#non-Horn}(P^n_{6}) = \parm{\#neg}(P^n_{7})
  = \parm{\#non-Horn}(P^n_{7}) = \parm{\#neg}(P^{m,n}_{8})
  = \parm{\#non-Horn}(P^{m,n}_{8}) =0$.
\item Consider the program~$P^n_9$ which contains only the atoms~$a_1$
  and $a_2$ that occur in $B^-(r)$ for some rule~$r \in P^n_9$ and
  only the non-Horn rules~$a_2 \leftarrow \pnot a_1$ and $a_3
  \leftarrow \pnot a_2$, hence $\parm{\#neg}(P^n_9)
  = \parm{\#non-Horn}(P^n_9) = 2$.
\item Consider the program~$P^n_{11}$. The set~$X := \{b\}$ is the
  only inclusion-minimal \delBds{\Normal} of $P^n_{11}$. Since
  $P^n_{11} - X$ is Horn, we have $\parm{\#neg}(P^n_{11} - X)
  = \parm{\#non-Horn}(P^n_{11} - X) = 0$. Thus
  $\lift{\parm{\#neg}}(P^n_{11}) = \Card{X} + \parm{\#neg}(P^n_{11} -
  X) = 1$ and $\lift{\parm{\#non-Horn}}(P^n_{11}) = \Card{X}
  + \parm{\#non-Horn}(P^n_{11} - X) = 1$.
\end{enumerate}
%\vspace{-5ex}
\end{OBS}

%  By applying Theorem~\ref{the:lift} to \parm{\#neg}
%   and \parm{\#non-Horn}, we can lift them to disjunctive
%   programs.

\begin{PRO}\label{lem:del_neg}
  $\parm{\#neg}$ and $\parm{\#non-Horn}$ are incomparable.
\end{PRO}
\begin{proof}
  The proposition directly follows from considering $P^n_1$ and
  $P^n_2$ where $\parm{\#neg}(P^n_1) = n$ and
  $\parm{\#non-Horn}(P^n_1) = 1$; and $\parm{\#neg}(P^n_2) = 1$ and
  $\parm{\#non-Horn} ( P^n_2) = n$ by Observation~\ref{ex:num-neg}.
  % Hence there are programs where $\parm{\#neg}$ is of constant size
  % and $\parm{\#non-Horn}$ can be arbitarily large, and there are
  % programs where the converse sustains. Thus we conclude $\parm{\#neg}
  % \bowtie \parm{\#non-Horn}$.
\end{proof}

However, it is easy to see that $\dbd_\Horn$ dominates both parameters.

\begin{PRO}\label{lem:del_neg}
  $\dbd_{\Horn}$ strictly dominates $\parm{\#neg}$ and
  $\parm{\#non-Horn}$. $\dbd_\CCC$ and $\parm{\#neg}$; and $\dbd_\CCC$
  and $\parm{\#non-Horn}$ are incomparable where~$\CCC \in \{\C$, \DC,
  \DCTWO, \EC, $\DEC\}$.
\end{PRO}
\begin{proof}
  For a normal program~$P$ define the sets~$B^-(P)=\SB a\in \at(P)\SM
  a\in B^-(r) \text{ for some rule }r\in P \SE$ and $H(P)=\SB a \in
  H(r) \SM r \in P,\, r \text{ is not Horn}\SE$. We observe that
  $B^-(P)$ and $H(P)$ are \delBds{\Horn}s of $P$, hence
  $\dbd_{\Horn}(P) \leq \parm{\#neg}(P)$ and $\dbd_{\Horn}(P)
  \leq \parm{\#non-Horn}(P)$.
  To show that $\dbd_{\Horn}$ strictly dominates the two parameters,
  consider $P^n_{31}$ where $\dbd_\Horn(P^n_{31})\leq 1$, but
  $\parm{\#neg}(P^n_{31}) = n+1$ and $\parm{\#non-Horn} (P^n_{31}) =
  2n$ by Observations~\ref{ex:bdclasses} and \ref{ex:num-neg}.

  The second statement follows from considering the
  programs~$P^n_{31}$ and $P^n_{53}$ where $\dbd_\CCC(P^n_{31}) \leq
  1$ and $p(P^n_{31}) \geq n +1 $; and $\dbd_\CCC(P^n_{53}) \geq n$
  and $p(P^n_{53}) = 0$ for $\CCC \in \{\C$, \DC, \DCTWO, \EC,
  $\DEC\}$ and $p \in \{\parm{\#neg}, \parm{\#non-Horn}\}$ by
  Observations~\ref{ex:bdclasses} and \ref{ex:num-neg}. Hence
  $\dbd_\CCC \bowtie \parm{\#neg}$ and $\dbd_\CCC
  \bowtie \parm{\#non-Horn}$ for $\CCC \in \{\C$, \DC, \DCTWO, \EC,
  $\DEC\}$.
\end{proof}

\subsection{\ASP-Parameters Based on the Distance from Being
  Stratified}
\citex{Ben-Eliyahu96} and \citex{GottlobScarcelloSideri02} have
considered \ASP-parameters that measure in a certain sense how far
away a program is from being stratified. In this section we will
investigate how these parameters fit into our landscape of
\ASP-parameters. Similar to the last section the parameters have been
considered for normal programs only, hence we compare the parameters
for normal programs only. Again, in view of
Observation~\ref{obs:compare-lift} the results also hold for the
lifted parameters to disjunctive programs.

Recall from Section~\ref{sec:graphs} that $\SCC(G)$ denotes the
partition of the vertex set of a digraph into strongly connected
components.

\begin{DEF}[\citex{Ben-Eliyahu96}]
  Let $P$ be a normal program, $D_P$ its dependency digraph, and $A
  \subseteq \at(P)$.  $\restrictBE{P}{A}$ denotes the program obtained
  from $P$ by (i)~deleting all rules~$r$ in the program~$P$ where
  $H(r) \cap A = \emptyset$ and (ii)~removing from the bodies of the
  remaining rules all literals $\neg a$ with $a \notin A$. Then
  \begin{align*}
    \parm{lstr}(P):={}&\sum_{C\in \SCC(D_P)}
    \min\{\parm{\#neg}(\restrictBE{P}{C}), 
    \parm{\#non-Horn}(\restrictBE{P}{C})\}.
  \end{align*}
  $\parm{lstr}(P)$ is called the \emph{level of stratifiability} of
  $P$.
\end{DEF}

% \begin{EX}[\cite{Ben-Eliyahu96}]
%   Consider program~$P= \SB a \leftarrow \pnot b, e \rsep c \leftarrow
%   \pnot d, a\SE$ and $A=\{a,b,c\}$ we obtain $\restrictBE{P}{A} = \SB
%   a \leftarrow \pnot b, e \rsep c \leftarrow a\SE$ and $\parm{lstr}(P)
%   = 0$.
% \end{EX}

\begin{PRO}[\citex{Ben-Eliyahu96}]
  For each~$L\in \AspFull$, $\pnormal{L[\parm{lstr}]}\in \FPT$.
\end{PRO}

Since $\pnormal{\Bound[\parm{lstr}]}$ and $\pnormal{
  \Bound[\parm{lstr}] }$ are clearly solvable in polynomial time and
thus fixed-parameter tractable, we can use the Lifting Theorem
(Theorem~\ref{the:lift}) to obtain the following result.

\begin{COR}
  For each~$L\in \AspFull$, $L[\lift{\parm{lstr}}]\in \FPT$.
\end{COR}

\begin{OBS}\label{ex:lstr}
  We make the following observations about programs from
  Example~\ref{ex:comp_progs}.
%  We continue Example~\ref{ex:comp_progs} and observe:
%~
\begin{enumerate}
\item Consider program~$P^n_{31}$ and let $P:=P^n_{31}$. The
  partition~$\SCC(D_P)$ contains only the set~$C := \at(P)$ and thus
  $\restrictBE{P}{C} = P$. By Observation~\ref{ex:num-neg}
  $\parm{\#neg}(P) = n+1$ and $\parm{\#non-Horn}(P) = 2n$ and hence
  $\parm{lstr}(P^n_{31}) = n+1$.
\item Consider program~$P^n_{32}$ and let $P:=P^n_{32}$. The
  partition~$\SCC( D_P )$ contains only the set~$C:=\at(P)$ and
  $\restrictBE{P}{C} = P$. Since $\parm{\#neg}(P) = 0$ by
  Observation~\ref{ex:num-neg}, we have $\parm{lstr}(P^n_{32}) = 0$.
\item Consider program~$P^n_{35}$ and let $P:=P^n_{35}$. The
  partition~$\SCC( D_P)$ contains only the set~$C := \at(P)$. Thus $P
  = \restrictBE{P}{C}$. Since
% $P$ contains only the atom $a$ that
%   occurs in $B^-(r)$ for some rule~$r \in P^n_{35}$, 
  $\parm{\#neg}(P^n_{35}) = 1$ by Observation~\ref{ex:num-neg}, we
  conclude $\parm{lstr}(P^n_{35}) \leq 1$.
\item Consider program~$P^n_{4}$ and let $P:=P^n_4$. We have $\SCC(
  D_{P} )$ contains exactly the sets~$A_i:=\{a_i, e_i, d_i \}$,
  $B_i:=\{b_i\}$, and $C_i:= \{c_i\}$ where $ 1 \leq i \leq n$. Hence
  $\restrictBE{P}{A_i} = \SB a_i \leftarrow e_i \rsep e_i \leftarrow
  d_i \rsep d_i \leftarrow a_i \SE$ and $\restrictBE{P}{B_i} = \SB b_i
  \SE$ and $\restrictBE{P}{C_i} = \SB c_i \rsep c_i \leftarrow b_i
  \SE$. Since $\parm{\#neg} (\restrictBE{P}{C}) = 0$ for every~$C \in
  \SCC(D_{P})$, we have $\parm{lstr} (P^n_4) = 0$.
\item Consider program~$P^n_{51}$ and $P^n_{52}$ and let $P \in
  \{P^n_{51}, P^n_{52}\}$. The partition~$\SCC( D_{P} )$ contains
  exactly the sets~$C_i := \{a_i, b_i \}$ where $1 \leq i \leq n$ and
  hence $\restrictBE{P}{C_i} = \SB b_i \leftarrow \pnot a_i \rsep a_i
  \leftarrow \pnot b_i \SE$, respectively $\restrictBE{P}{C_i} = \{
  b_i \leftarrow a_i \rsep a_i \leftarrow \pnot b_i \SM 1 \leq i \leq
  n \}$. Since $\parm{\#neg}(\restrictBE{P}{C_i})
  = \parm{\#non-Horn}(\restrictBE{P}{C_i}) = 2$, respectively
  $\parm{\#neg}(\restrictBE{P}{C_i})
  = \parm{\#non-Horn}(\restrictBE{P}{C_i}) = 1$, and there are $n$
  components we obtain $\parm{lstr}(P^n_{51}) = 2n$ and
  $\parm{lstr}(P^n_{52}) = n$.
\item Consider program~$P^n_{53}$ and let $P:=P^n_{53}$. The
  partition~$\SCC(D_P)$ contains only the set~$C:=\at(P)$ and
  $\restrictBE{P}{C}=P$. Since $\parm{\#neg}(P) = 0$ by
  Observation~\ref{ex:num-neg}, we have $\parm{lstr}(P^n_{53}) = 0$.
\item Consider program~$P^n_6$ and let $P:= P^n_6$. The
  partition~$\SCC(D_P)$ contains exactly the sets~$A:=\{a\}$,
  $B_i:=\{b_i\}$, and $C_i:=\{c_i\}$ where $1 \leq i \leq n$. Hence
  $\restrictBE{P}{A} = \SB a \leftarrow b_1, \ldots, b_n, c_i \SM 1
  \leq i \leq n \SE$ and $\restrictBE{P}{B_i} = \restrictBE{P}{C_i} =
  \emptyset$ where $1 \leq i \leq n$. Since $\parm{\#neg}
  (\restrictBE{P}{C}) = 0$ for every~$C \in \SCC(D_{P})$, we have
  $\parm{lstr} (P^n_6) = 0$.
\item Consider program~$P^n_7$ and let $P:=P^n_7$. The
  partition~$\SCC(D_P)$ contains exactly the sets~$C_i:=\{a_i\}$ where
  $1 \leq i \leq n$. Thus $\restrictBE{P}{C_i}= \SB a_i \leftarrow a_j
  \SM 1 \leq j < i\SE$. Hence $\parm{\#neg}(\restrictBE{P}{C_i}) = 0$
  for every $C \in \SCC(D_{P})$. We obtain $\parm{lstr}(P^n_7) = 0$.
\item Consider program~$P^{m,n}_8$ and let $P:= P^{m,n}_8$. The
  partition~$\SCC( D_{P} )$ contains exactly the sets~$A_i:=\{a_i\}$
  where $1 \leq i \leq m$, $B:=\{b\}$, and $C := \SB c_i \SM 1 \leq i
  \leq n \SE$. Hence $\restrictBE{P}{A_i} = \emptyset$ where $1 \leq i
  \leq m$, $\restrictBE{P}{B} = \{ b \leftarrow a_1, \ldots, a_m \}$,
  and $\restrictBE{P}{C} = \SB c_i \leftarrow c_{i+1} \SM 1 \leq i
  \leq n \SE \cup \SB c_{n+1} \leftarrow c_1 \SE$. Since
  $\parm{\#neg}(\restrictBE{P}{A_i}) = 0$ where $1 \leq i \leq m$,
  $\parm{\#neg} (\restrictBE{P}{B}) = 0$, and $\parm{\#neg}
  (\restrictBE{P}{C}) = 0$, we obtain $\parm{lstr}(P^{m,n}_8) = 0$.
\item Consider program~$P^n_9$ and let $P:=P^n_9$. The
  partition~$\SCC( D_{P} )$ contains only the set~$C:=\at(P)$. Hence
  $\restrictBE{P}{C} = P$. Since $\parm{\#neg}(P) = \parm{\#non-Horn}
  (P) = 2$, we have $\parm{lstr}(P^n_9) = 2$.
\item Consider program~$P^n_{11}$. The set~$X = \{b\}$ is the
  inclusion-minimal \delBds{\Normal} of $P^n_{11}$ by
  Observation~\ref{ex:bdclasses}. We have $P:=P^n_{11} - X = \SB a_i
  \leftarrow c \rsep c \rsep \leftarrow a_i \SM 1 \leq i \leq n
  \SE$. The partition~$\SCC( D_{P} )$ contains the sets~$A_i :=
  \{a_i\}$ where $1 \leq i \leq n$ and $C:=\{c\}$. Hence
  $\restrictBE{P}{A_i} = \SB a_i \leftarrow c \SE$ where $1 \leq i
  \leq n$ and $\restrictBE{P}{C} = \SB c \SE$. Since $\parm{\#neg}
  (\restrictBE{P}{C}) = 0$ for every~$C \in \SCC(D_{P})$, we obtain
  $\parm{lstr} (P) = 0$. Consequently, $\lift{\parm{lstr}}(P^n_{11}) =
  \Card{X} + \parm{lstr}(P^n_{11} - X) = 1$.
\end{enumerate}
%\vspace{-5ex}
\end{OBS}

\begin{OBS}\label{obs:num-neg_lstr}
  $\parm{lstr}$ strictly dominates $\parm{\#neg}$ and
  $\parm{\#non-Horn}$.
\end{OBS}
\begin{proof}
  Let $P$ be a normal program.  We first show that $\sum_{C \in
    \SCC(D_P)} \parm{\#neg}(\restrictBE{P}{C})\leq\parm{\#neg}(P)$. Define
  the set~$B^-(P)=\SB a\in \at(P)\SM a\in B^-(r) \text{ for some rule
  }r\in P \SE$. By definition $B^-(\restrictBE{P}{A})\subseteq B^-(P)$
  for some~$A \subseteq \at(P)$, thus $\bigcup_{C \in \SCC(D_P)}
  B^-(\restrictBE{P}{C}) \subseteq B^-(P)$. Let $C, C' \in \SCC(D_P)$
  and $C \neq C'$. By definition of a strongly connected component we
  have $C \cap C' = \emptyset$ and by definition we have that
  $B^-(\restrictBE{P}{C}) \subseteq C$ and $B^-(\restrictBE{P}{C'})
  \subseteq C'$. Hence $B^-(\restrictBE{P}{C}) \cap
  B^-(\restrictBE{P}{C'}) = \emptyset$.  Consequently $\sum_{C \in
    \SCC(D_P) } \parm{\#neg} (\restrictBE{P}{C})
  \leq \parm{\#neg}(P)$. A similar argument shows that $\sum_{C \in
    \SCC(D_P)} \parm{\#non-Horn} (\restrictBE{P}{C}) \leq \parm{
    \#non-Horn }(P)$.  Since $\parm{lstr}(P) =
  \sum_{C\in\SCC(D_P)}\min\{\parm{\#neg}(
  \restrictBE{P}{C}), \parm{\#non-Horn}(\restrictBE{P}{C})\}$, we have
  $\parm{lstr}(P)\leq \parm{\#neg}(P)$ and
  $\parm{lstr}(P)\leq \parm{\#non-Horn}(P)$.
  To show that $\parm{lstr}$ strictly dominates the two parameters,
  consider program~$P^n_4$ where $\parm{lstr}(P^n_4) = 0$, but
  $\parm{\#neg}(P^n_4) \geq n$ and $\parm{ \#non-Horn }(P^n_4) \geq
  2n$ by Observations~\ref{ex:num-neg} and \ref{ex:lstr}. Hence the
  observation is true.
\end{proof}

\begin{PRO}\label{pro:lstr}
  $\dbd_\DBC$ strictly dominates $\parm{lstr}$. Moreover, $\dbd_\CCC$
  and $\parm{lstr}$ are incomparable for the remaining target classes
  namely~$\CCC\in \Acyc\setminus\{\DBC,\DBEC\} \cup \{\Horn\}$.
\end{PRO}
\begin{proof}
  We first show that $\dbd_\DBC$ dominates $\parm{lstr}$. For a normal
  program~$P$ define the sets~$B^-(P)=\SB a \in \at(P) \SM a\in B^-(r)
  \text{ for some rule } r \in P \SE$ and $H(P)=\SB a \in H(r) \SM r
  \in P, r \text{ is not Horn} \SE$. Let $C \in \SCC(D_P)$, we define
  \[
  X_C = 
  \begin{cases} 
    B^-(\restrictBE{P}{C}), & \text{if }\Card{B^-(\restrictBE{P}{C})}
    \leq \Card{H(\restrictBE{P}{C})};\\
    H( \restrictBE{P}{C}), & \text{ otherwise.}
  \end{cases}
  \]
  and $X = \SB X_C \SM C \in \SCC(D_P)\SE$. We show that $X$ is a
  \delBds{\DBC} of $P$.
  By definition for every directed bad cycle~$c = (x_1, \ldots, x_l)$
  of $D_P$ the atom~$x_i \in C'$ where $1 \leq i \leq l$ and $C' \in
  \SCC(D_P)$ (all vertices of $c$ belong to the same strongly
  connected component). Moreover, by definition we have for every
  negative edge~$x_i,x_j \in D_P$ of the dependency digraph~$D_P$ a
  corresponding rule~$r \in P$ such that $x_j \in H(r)$ and $x_i \in
  B^-(r)$. Since $X_C$ consists of either $B^-(\restrictBE{P}{C})$ or
  $H(\restrictBE{P}{C})$, at least one of the atoms~$x_i,x_j$ belongs
  to $X_C$. Thus for every directed bad cycle~$c$ of the program~$P$
  at least one atom of the cycle belongs to $X$.  Hence $P - X \in
  \DBC$ and $X$ is a \delBds{\DBC} of $P$. We obtain $\dbd_\DBC(P)
  \leq \parm{lstr}(P)$.
  To show that $\dbd_{\DBC}$ strictly dominates $\parm{lstr}$,
  consider program~$P^{n}_{31}$ where $\dbd_\DBC(P^n_{31}) \leq 1$ and
  $\parm{lstr}(P^n_{31}) = n+1$ by Observations~\ref{ex:bdclasses} and
  \ref{ex:lstr}. Hence $\dbd_\DBC \prec \parm{lstr}$.

  Then we show that the parameters~$\dbd_\CCC$ and $\parm{lstr}$ are
  incomparable.
  Consider the programs~$P^{n}_{3}$ and $P^n_{4}$ where
  $\dbd_\CCC(P^n_{31}) \leq 1$ and $\parm{lstr}(P^n_{31})=n+1$; and
  $\parm{lstr}(P^n_{4}) = 0$ and $\dbd_\CCC(P^n_{4}) \geq n$ for $\CCC
  \in \{\Horn$, \C, \BC, \DC, \DCTWO, \EC, \BEC, $\DEC\}$ by
  Observations~\ref{ex:bdclasses} and \ref{ex:lstr}.
% Since there are
%   programs that have smallest \delBds{\CCC} of constant size and an
%   arbitarily large parameter~\parm{lstr}, and there are programs where
%   the converse sustains, 
  We conclude $\dbd_\CCC \bowtie \parm{lstr}$.
\end{proof}

%   (i)~deleting all rules~$r$ in the program~$P$ where $H(r) \cap A =
%   \emptyset$ and (ii)~removing from the bodies of 
% %every rule 
%   the remaining rules all literals~$a$, $\pnot p$ with $a \notin A$. $\Good{P}$
%   denotes the maximal set~$W \subseteq \at(P)$ such that there is no
%   bad $W$\hy cycle in the dependency graph of~$U_P$, in other words
%   the set of all atoms that do not lie on a bad cycle of $P$. Then

% We denote by \Good{P} the set of atoms of $P$ that do lie on a bad
% cycle of $P$.
\begin{DEF}[\citex{GottlobScarcelloSideri02}]
  Let $P$ be a normal program, $D_P$ its dependency digraph, $U_P$ its
  dependency graph, and $A \subseteq \at(P)$.  $\restrictGSS{P}{A}$
  denotes the program obtained from $\restrictBE{P}{A}$ by removing
  from the bodies of every rule all literals~$a$ with $a \notin
  A$. $\Good{P}$ denotes the maximal set~$W \subseteq \at(P)$ such
  that there is no bad $W$\hy cycle in the dependency graph~$U_P$, in
  other words the set of all atoms that do not lie on a bad cycle of
  $P$. Then
  \begin{align*}
    \parm{fw}(P):={}&\min \SB \Card{S} \SM S \text{ is a feedback 
      vertex set of } U_P \SE \text{ and }\\
    \parm{wfw}(P):={}&\parm{fw}(\SB r \in \restrictGSS{P}{C} -
    \Good{\restrictGSS{P}{C}} \SM C \in \SCC(D_P), \restrictGSS{P}{C}
    \notin \DBC \SE).
  \end{align*}
  $\parm{fw}(P)$ is called the \emph{feedback-width} of $P$, and
  $\parm{wfw}(P)$ is called the \emph{weak-feedback-width} of $P$.
\end{DEF}

\begin{OBS}\label{obs:fw}
  Let $P$ be a normal program and $D_P$ its dependency digraph. Then
  $\parm{fw(P)}=\dbd_\C(P)$ and hence
  \begin{align*}
    \parm{wfw(P)}={}&\dbd_\C(\SB r \in \restrictGSS{P}{C} - \Good{
      \restrictGSS{P}{C} } \SM C \in \SCC(D_P), \restrictGSS{P}{C}
    \notin \DBC \SE).
  \end{align*}
\end{OBS}

% $\restrictGSS{P}{A} = \SB a \leftarrow \pnot b \rsep c \leftarrow a
% \SE$.

\begin{PRO}[\citex{GottlobScarcelloSideri02}]
  For each~$L\in \AspFull$, $\pnormal{L[\parm{fw}]}\in \FPT$ and
  $\pnormal{L[\parm{wfw}]}\in \FPT$.
\end{PRO}

Since $\pnormal{\Bound[\parm{fw}]}$ and $\pnormal{ \Bound[\parm{wfw}]
}$ is fixed-parameter tractable, we can use the Lifting Theorem
(Theorem~\ref{the:lift}) to obtain the following result.

\begin{COR}
  For each~$L\in \AspFull$, $L[\lift{\parm{fw}}]\in \FPT$ and
  $L[\lift{\parm{wfw}}]\in \FPT$.
\end{COR}

\begin{OBS}\label{ex:wfw}
  We make the following observations about programs from
  Example~\ref{ex:comp_progs}.
%  We continue Example~\ref{ex:comp_progs} and observe:
%~
\begin{enumerate}
\item Consider the program~$P^n_{31}$ and define $P:= P^n_{31}$. The
  partition~$\SCC(D_{P})$ contains only the set~$C := \at(P)$. For
  every atom~$a \in C$ the program~$P$ contains a bad $\{a\}$\hy cycle
  and thus $\Good{\restrictGSS{P}{C}} = \emptyset$. Consequently,
  $\restrictGSS{P}{C} - \Good{\restrictGSS{P}{C}} = \restrictGSS{P}{C}
  = P$. As $P \notin \DBC$, $\SB r \in \restrictGSS{P}{C} -
  \Good{\restrictGSS{P}{C}}, C \in \SCC(D_P), \restrictGSS{P}{C}
  \notin \DBC \SE = P$. We have $\dbd_\C(P)=1$ by
  Observation~\ref{ex:bdclasses} and according to
  Observation~\ref{obs:fw} we obtain $\parm{wfw}(P^n_{31})=1$.
\item Consider program~$P^n_{32}$ and let $P:=P^n_{32}$. The
  partition~$\SCC( D_{P} )$ contains only the set~$C := \at(P)$,
  $\restrictGSS{P}{C} = P$. For every atom~$a \in C$ we have
  $\restrictGSS{P}{C} \in \DBC$ and thus $\SB r \in \restrictGSS{P}{C}
  - \Good{\restrictGSS{P}{C}} \SM C \in \SCC(D_P), \restrictGSS{P}{C}
  \notin \DBC \SE = \emptyset$. Consequently, $\parm{wfw}(P^n_{32}) =
  0$.
\item Consider the programs~$P^n_{33}$, $P^n_{34}$, and $P^n_{35}$ and
  let $P \in \{P^n_{33}, P^n_{34}, P^n_{35}\}$. We first observe that
  the dependency digraph of~$P$ contains only one strongly connected
  component. Hence the partition~$\SCC( D_{P})$ contains only the
  set~$C := \at(P)$. For every atom~$a \in C$ program~$P$ contains a
  bad $\{a\}$\hy cycle and thus $\Good{\restrictGSS{P}{C}} =
  \emptyset$. Consequently, $\restrictGSS{P}{C} -
  \Good{\restrictGSS{P}{C}} = \restrictGSS{P}{C} = P$. Since $P \notin
  \DBC$, we obtain $\SB r \in \restrictGSS{P}{C} - \Good{
    \restrictGSS{P}{C}}, C \in \SCC(D_P), \restrictGSS{P}{C} \notin
  \DBC \SE = P$. We have $\dbd_\C(P) = n$ since $P$ contains $n$
  disjoint $\{b_i\}$-cycles. According to Observation~\ref{obs:fw} we
  conclude $\parm{wfw}(P^n_{33}) = \parm{wfw}(P^n_{34})
  = \parm{wfw}(P^n_{35}) = n$.
\item Consider program~$P^n_{4}$ and let $P:=P^n_4$. The partition
  $\SCC( D_{P} )$ contains exactly the sets $A_i:= \{a_i, d_i, e_i
  \}$, $B_i:=\{b_i\}$, and $C_i:= \{c_i\}$ where $1 \leq i \leq
  n$. Hence $\restrictGSS{P}{A_i} = \SB a_i \leftarrow e_i \rsep e_i
  \leftarrow d_i \rsep d_i \leftarrow a_i \SE$, $\restrictGSS{P}{B_i}
  = \SB b_i \SE$ and $\restrictGSS{P}{C_i} = \SB c_i \SE$. For every
  $C \in \SCC(D_P)$ the program~$\restrictGSS{P}{C} \in
  \DBC$. Consequently, $\SB r \in \restrictGSS{P}{C} -
  \Good{\restrictGSS{P}{C}}, C \in \SCC(D_P), \restrictGSS{P}{C}
  \notin \DBC \SE = \emptyset$ and we obtain $\parm{wfw}(P^n_4)=0$.
\item Consider program~$P^n_{51}$ and let $P:=P^n_{51}$. The
  partition~$\SCC( D_{P} )$ contains exactly the sets~$C_i := \{a_i,
  b_i \}$ where $1 \leq i \leq n$ and thus $\restrictGSS{P}{C_i} = \SB
  a_i \leftarrow \pnot b_i \rsep b_i \leftarrow \pnot a_i \SE$. Since
  $\dbd_\C (\restrictGSS{P}{C_i}) = 1$ and there are $n$ components we
  obtain $\parm{wfw}(P^n_{51}) = n$.
\item
  Consider program~$P^n_{52}$ and let $P:=P^n_{52}$.  We observe that
  the partition~$\SCC( D_{P} )$ contains exactly the
  sets~$C_i:=\{a_i,b_i\}$.  For every atom~$a \in C_i$ where $1 \leq i
  \leq n$ there is a bad $\{a\}$\hy cycle in the dependency graph
  of~$\restrictGSS{P}{C_i}$ and thus $\Good{\restrictGSS{P}{C_i}} =
  \emptyset$. Consequently, $\restrictGSS{P}{C_i} -
  \Good{\restrictGSS{P}{C_i}} = \restrictGSS{P}{C_i}$. Since
  $\restrictGSS{P}{C_i} \notin \DBC$, $\SB r \in \restrictGSS{P}{C} -
  \Good{\restrictGSS{P}{C}} \SM C \in \SCC(D_P), \restrictGSS{P}{C}
  \notin \DBC \SE = P$.  We observe that $\dbd_\C(P) = n$ and
  according to Observation~\ref{obs:fw} we obtain
  $\parm{wfw}(P^n_{52}) = n$.
\item Consider program~$P^n_6$ and let $P:= P^n_6$. The
  partition~$\SCC(D_P)$ contains exactly the sets~$A:=\{a\}$,
  $B_i:=\{b_i\}$, and $C_i:=\{c_i\}$ where $1 \leq i \leq n$. Hence
  $\restrictGSS{P}{A} = \SB a \SE$ and $\restrictGSS{P}{B_i} =
  \restrictGSS{P}{C_i} = \emptyset$ where $1 \leq i \leq n$. Since
  $\dbd_\C (\restrictGSS{P}{C}) = 0$ for every~$C \in \SCC(D_{P})$, we
  obtain $\parm{wfw} (P^n_6) = 0$.
\item Consider program~$P^n_7$ and let $P:=P^n_7$. Since the
  partition~$\SCC(D_P)$ contains exactly the sets~$\{a_i\}$ where $1
  \leq i \leq n$, $\restrictGSS{P}{\{a_i\}}= \SB a_i \SE$ and thus
  $\parm{wfw}(\restrictGSS{P}{\{a_i\}}) = 0$. We obtain
  $\parm{wfw}(P^n_7) = 0$.
\item Consider program~$P^{m,n}_8$ and let $P:= P^{m,n}_8$.  The
  partition~$\SCC( D_{P} )$ contains exactly the sets~$A_i:=\{a_i\}$
  for $1 \leq i \leq m$, $B:=\{b\}$, and $C := \SB c_i \SM 1 \leq i
  \leq n \SE$. Hence $\restrictGSS{P}{A_i} = \emptyset$ for $1 \leq i
  \leq m$, $\restrictGSS{P}{B} = \emptyset$, and $\restrictGSS{P}{C} =
  \SB c_i \leftarrow c_{i+1} \SM 1 \leq i \leq n \SE \cup \SB c_{n+1}
  \leftarrow c_1 \SE$.  The programs~$\restrictGSS{P}{A_i}$,
  $\restrictGSS{P}{B}$, and $\restrictGSS{P}{C}$ belong to the
  class~$\DBC$ for $1 \leq i \leq m$. Consequently $\SB r \in
  \restrictGSS{P}{C} - \Good{\restrictGSS{P}{C}} \SM C \in \SCC(D_P),
  \restrictGSS{P}{C} \notin \DBC \SE = \emptyset$. Hence we conclude
  that $\parm{wfw} (P^{m,n}_8) = 0$.
\item Consider program~$P^n_9$ and let $P:=P^n_9$. The
  partition~$\SCC( D_{P} )$ contains only the set~$C:=\at(P)$. For
  every atom~$a \in C$ there is a bad $\{a\}$\hy cycle in the
  dependency graph of~$P$ and thus $\Good{\restrictGSS{P}{C}} =
  \emptyset$. Consequently, $\restrictGSS{P}{C} -
  \Good{\restrictGSS{P}{C}} = \restrictGSS{P}{C} = P$. Since $P \notin
  \DBC$, $\SB r \in \restrictGSS{P}{C} - \Good{\restrictGSS{P}{C}} \SM
  C \in \SCC(D_P), \restrictGSS{P}{C} \notin \DBC \SE = P$. By
  Observation~\ref{ex:bdclasses} $\dbd_\C(P) \leq 1$ and according to
  Observation~\ref{obs:fw} we obtain $\parm{wfw}(P^n_9) \leq 1$.
\item Consider program~$P^n_{11}$ and let $P:=P^n_{11}$. The set~$X =
  \{b\}$ is the inclusion-minimal \delBds{\Normal} of $P^n_{11}$ by
  Observation~\ref{ex:bdclasses} and $P:=P^n_{11} - X = \SB a_i
  \leftarrow c \rsep c \rsep \leftarrow a_i \SM 1 \leq i \leq n
  \SE$. The partition~$\SCC( D_{P} )$ contains exactly the
  sets~$\{a_i\}$ for $1 \leq i \leq n$ and $\{c\}$. Hence
  $\restrictGSS{P}{\{a_i\}} = \SB a_i \SE$ for $1 \leq i \leq n$ and
  $\restrictGSS{P}{\{c\}} = \SB c \SE$.  We observe that $\dbd_\C
  (\restrictGSS{P}{C}) = 0$ for every~$C \in \SCC(D_{P})$ and
  according to Observation~\ref{obs:fw} we obtain $\parm{wfw} (P) =
  0$. Consequently, $\lift{\parm{wfw}}(P^n_{11}) = \Card{X}
  + \parm{wfw}(P^n_{11} - X) = 1$.
\end{enumerate}
%\vspace{-5ex}
\end{OBS}

In the following proposition we state the relationship between the
parameter~$\parm{wfw}$ and our backdoor-based \ASP parameters. The
first result ($\dbd_\DBC$ strictly dominates $\parm{wfw}$) was
anticipated by~\citex{GottlobScarcelloSideri02}.

\begin{PRO}\label{pro:del_wfw}
  $\parm{wfw}$ strictly dominates $\dbd_\C$ and $\dbd_\DBC$ strictly
  dominates $\parm{wfw}$. Moreover, $\dbd_\CCC$ and $\parm{wfw}$ are
  incomparable for the remaining target classes namely~$\CCC \in
  \{\Horn$, \BC, \DC, \DCTWO, \EC, \BEC, $\DEC\}$.
\end{PRO}
\begin{proof}
  We first show that $\parm{wfw}$ strictly dominates $\dbd_\C$. Let
  $P$ be a normal program and $X$ be a \delBds{\C} of $P$. Define
  $\hat P = \SB \restrictGSS{P}{C} - \Good{ \restrictGSS{P}{C} } \SM C
  \in \SCC(D_P), \restrictGSS{P}{C} \notin \DBC \SE$. Since $\hat P
  \subseteq P$ and $\C$ is hereditary
  (Observation~\ref{obs:acyclic_rule_induced}), $\hat P - X \in \C$
  and hence $X$ is a \delBds{\C} of $\hat P$. Consequently,
  $\parm{wfw}(P) \leq \dbd_\C(\hat P)$.
  To show that $\parm{wfw}$ is strictly more general than $\dbd_{\C}$,
  consider the program~$P^{n}_{4}$ where $\parm{wfw}(P^n_{4}) = 0$ and
  $\dbd_\C(P^n_{4}) = n$. Hence $\parm{wfw} \prec \dbd_\C$ by
  Observations~\ref{obs:bdsubclasses} and \ref{ex:wfw}.

  % Next we show that there are programs that have a \delBds{\CCC}
  % of constant size, but $\parm{wfw}$ is arbitrarily large and there
  % are programs where the converse sustains. Consider the

  Next, we show that $\dbd_\DBC$ strictly dominates $\parm{wfw}$. Let
  $P$ be a normal program and $\hat P = \SB \restrictGSS{P}{C} -
  \Good{ \restrictGSS{P}{C} } \SM C \in \SCC(D_P), \restrictGSS{P}{C}
  \notin \DBC \SE$. According to Observation~\ref{obs:fw}
  $\parm{wfw}(P) = \dbd_\C(\hat P)$ and thus it is sufficient to show
  that $\dbd_\DBC(P) < \dbd_\C(\hat P)$. Let $X$ be an arbitrary
  \delBds{\C} of $\hat P$. Since $\C \subseteq \DBC$
  Observation~\ref{obs:bdsubclasses} yields that $X$ is also a
  \delBds{\DBC} of~$\hat P$. Let $c$ be an arbitrary directed bad cycle
  of $D_P$. As all vertices of $c$ belong to the same partition~$C \in
  \SCC(D_P)$, $\at(\restrictGSS{P}{C}) \subseteq C$, and
  $D_{\restrictGSS{P}{C}}$ is an induced subdigraph of $D_P$ on
  $\at(\restrictGSS{P}{C})$, we obtain $c$ is a directed bad cycle in
  $D_{\restrictGSS{P}{C}}$. Since $\hat P = \SB \restrictGSS{P}{C} -
  \Good{\restrictGSS{P}{C}} \SM C \in \SCC(D_P), \restrictGSS{P}{C}
  \notin \DBC \SE$ and by definition there is no
  $\Good{\restrictGSS{P}{C}}$\hy cycle in~$U_P$, there is no directed
  bad $\Good{\restrictGSS{P}{C}}$-cycle in~$D_P$ and hence $c$ is also
  a directed bad cycle in~$D_{\restrictGSS{P}{C}}$. Since $X$ is a
  \delBds{\DBC} of $D_{\restrictGSS{P}{C}}$ and $c$ is a directed bad
  $X$\hy cycle in~$D_{\restrictGSS{P}{C}}$, $X$ is also a
  \delBds{\DBC} of $P$. Consequently, $\dbd_\DBC(P)\leq\dbd_\C(\hat P)
  = \parm{wfw}(P)$. To show that $\dbd_{\DBC}$ is strictly more
  general than the parameter~$\parm{wfw}$, consider the
  program~$P^{n}_{33}$ where $\dbd_\DBC(P^n_{33})=0$ and
  $\parm{wfw}(P^n_{33}) = n$ by Observations~\ref{ex:bdclasses} and
  \ref{ex:wfw}. Hence $\dbd_\DBC \prec \parm{lstr}$.

  The third statement follows from considering the
  programs~$P^n_{33}$, $P^n_{34}$, and $P^n_{4}$ where
  $\dbd_\CCC(P^n_{33}) \leq 1$ for $\CCC\in \{\Horn$, \BC, \DC,
  \DCTWO, \BEC, $\DEC\}$ and $\dbd_\EC(P^n_{34}) \leq 1$ and
  $\parm{wfw}(P^n_{33}) = \parm{wfw}(P^n_{34}) = n$; and
  $\parm{wfw}(P^n_{4}) = 0$ and $\dbd_\CCC(P^n_{4}) = n$ by
  Observations~\ref{ex:bdclasses} and \ref{ex:wfw}. Hence $\dbd_\CCC
  \bowtie \parm{wfw}$ for $\CCC \in \{\Horn$, \BC, \DC, \DCTWO, \EC,
  \BEC, $\DEC\}$.

\end{proof}

% Since $\parm{wfw}$ and $\dbd_\Horn$ are incomparable and $\dbd_\Horn$
% strictly dominates $\parm{\#neg}$ and $\parm{\#non-Horn}$ we directly
% obtain:

% \begin{COR}
%   ${\#neg}$ and $\parm{wfw}$ are incomparable, and $\parm{\#non-Horn}$
%   and $\parm{wfw}$ are incomparable.
% \end{COR}
% \johannes{Fix Me}

\begin{OBS}
  Let $p\in\{\parm{\#neg}, \parm{\#non-Horn}, \parm{lstr} \}$, then
  $p$ and $\parm{wfw}$ are incomparable.
\end{OBS}
\begin{proof}
  To show that $p$ and $\parm{wfw}$ are incomparable consider the
  programs~$P^n_{31}$ and $P^n_{35}$ where $p(P^n_{31}) \geq n + 1$
  and $\parm{wfw}(P^n_{31}) = 1$; and $p(P^n_{35}) \leq 1$ and
  $\parm{wfw}(P^n_{35}) = n$ by Observations~\ref{ex:num-neg},
  \ref{ex:lstr} and \ref{ex:wfw}.
\end{proof}

\subsection{Incidence Treewidth}
Treewidth is graph parameter introduced by Robertson and
Seymour~\shortcite{RobertsonSeymour84,RobertsonSeymour85,RobertsonSeymour86}
that measures in a certain sense the tree-likeness of a
graph. See~\shortcite{Bodlaender93a,Bodlaender97,Bodlaender05,GottlobPichlerWei10}
for further background and examples on treewidth.  Treewidth has been
widely applied in knowledge representation, reasoning, and artificial
intelligence~\cite{Dunne07,GottlobPichlerWei10,JaklPichlerWoltran09,MorakWoltran12,PichlerRummeleWoltran09}.

\begin{DEF}
  Let $G=(V,E)$ be a graph, $T$ a tree, and $\chi$ a labeling that
  maps any node $t$ of $T$ to a subset $\chi(t) \subseteq V$. We call
  the sets $\chi(\cdot)$ \emph{bags} and denote the vertices of $T$ as
  \emph{nodes}. The pair~$(T,\chi)$ is a \emph{tree decomposition} of
  $G$ if the following conditions hold:
  \begin{enumerate}
  \item for every vertex~$v\in V(G)$ there is a node~$t\in V(T)$ such
    that $v \in \chi(t)$ (``vertices covered'');
  \item for every edge~$vw \in E(G)$ there is a node~$t \in V(T)$ such
    that $v,w \in \chi(t)$ (``edges covered''); and
  \item for any three nodes~$t_1,t_2,t_3\in V(T)$, if $t_2$ lies on
    the unique path from~$t_1$ to~$t_3$, then $\chi(t_1)\cap \chi(t_3)
    \subseteq \chi(t_2)$ (``connectivity'').
  \end{enumerate}
  The \emph{width} of the tree decomposition~$(T,\chi)$ is $\max\SB
  |\chi(t)| -1 \SM t \in V(T) \SE$. The \emph{treewidth} of~$G$,
  denoted by \tw(G), is the minimum taken over the widths of all
  possible tree decompositions of~$G$.
\end{DEF}

We will use the following basic properties of treewidth.
\begin{LEM}[Folklore, e.g.,~\cite{RobertsonSeymour86}]
  \label{lem:tw-components}
  Let $G$ be a graph and $C_1, \ldots, C_l$ its connected components,
  then $\tw(G)=\max \SB \tw(C_j) \SM 1 \leq i \leq l \SE$.
\end{LEM}

\begin{LEM}[Folklore, e.g.,~\cite{BodlaenderKoster08}]\label{lem:fw-tw}
  Let $G$ be a graph. If $G$ has a feedback vertex set size at
  most~$k$, then $\tw(G) \leq k +1$.
\end{LEM}

Treewidth can be applied to programs by means of various graph
representations.

\begin{DEF}[\citex{JaklPichlerWoltran09}]
  Let $P$ be a normal program. The \emph{incidence graph}~$I_P$ of~$P$
  is the bipartite graph which has as vertices the atoms and rules of
  $P$ and where a rule and an atom are joined by an edge if and only
  if the atom occurs in the rule. Then $\parm{inctw}(P) :=
  \tw(I_P)$. The parameter $\parm{inctw}(P)$ is called the
  \emph{incidence treewidth} of $P$.
\end{DEF}

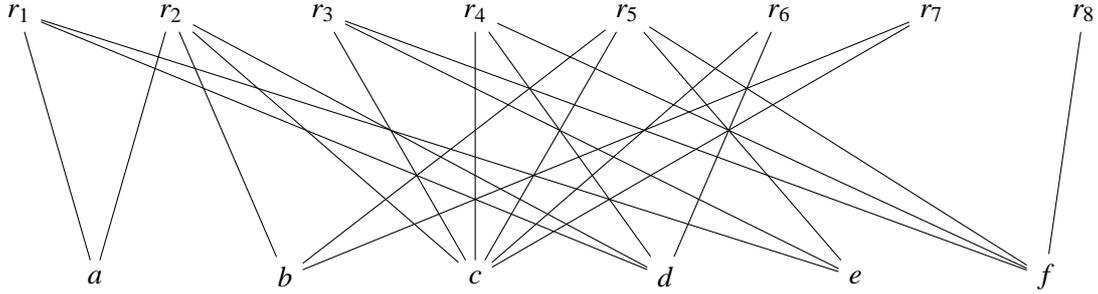
\begin{figure}
\centering
\begin{tikzpicture}[-,node distance=15mm]
\def \m {8}

\foreach \x in {1,...,\m}
{
  \node[](r\x) at (\x*2,0){$r_\x$};
}
\foreach \x/\y in {1/a,2/b,3/c,4/d,5/e,6/f}
{
  \node[](\y) at (0.5+\x*2.5,-3.5){$\y$};
}

    \path(r1) edge[] (a);
    \path(r1) edge[] (d);
    \path(r1) edge[] (e);

    \path(r2) edge[] (a);
    \path(r2) edge[] (b);
    \path(r2) edge[] (c);
    \path(r2) edge[] (d);

    \path(r3) edge[] (c);
    \path(r3) edge[] (f);
    \path(r3) edge[] (e);

    \path(r4) edge[] (c);
    \path(r4) edge[] (d);
    \path(r4) edge[] (f);

    \path(r5) edge[] (b);
    \path(r5) edge[] (c);
    \path(r5) edge[] (f);
    \path(r5) edge[] (e);

    \path(r6) edge[] (c);
    \path(r6) edge[] (d);

    \path(r7) edge[] (b);
    \path(r7) edge[] (c);

    \path(r8) edge[] (f);

  \end{tikzpicture}
  \caption{Incidence graph~$I_P$ of the program~$P$ of
    Example~\ref{ex:running}.}
\label{fig:incidence_graph}
\end{figure}

\newcommand{\DelayL}{\text{\normalfont DelayL}\xspace}

\begin{PRO}[\citex{JaklPichlerWoltran09}]
  For each~$L \in \AspFull \setminus \{ \AspEnum \}$,
  $\pnormal{L[\parm{inctw}]} \in \FPT$ and for
  $\pnormal{\AspEnum[\parm{inctw}]}$ the solutions can be enumerated
  with fixed-parameter linear delay between any two consecutive
  solutions.
\end{PRO}

\begin{OBS}\label{ex:inctw}
  We make the following observations about programs from
  Example~\ref{ex:comp_progs}.
%  We continue Example~\ref{ex:comp_progs} and observe:
%
%~
\begin{enumerate}
\item Consider the programs~$P^n_{32}$ and $P^n_{51}$. We observe that
  its incidence graph consists of the $n$ cycles~$b_i, r_i, a_i,
  r_{2i}$, $a_i, r_i, b_i, r_{2i}$ respectively, where $1 \leq i \leq
  n$. According to Lemma~\ref{lem:fw-tw} a cycle has treewidth at
  most~$2$ and according to Lemma~\ref{lem:tw-components} we have
  $\parm{inctw}(P^n_{32}) \leq 2$ and $\parm{inctw}(P^n_{51}) \leq 2$.
%
  % Consider program~$P^n_{51}$. The incidence graph of $P^n_{51}$ consists of
  % $n$ connected components consisting only of a cycle~$a_i, r_i, b_i,
  % r_{2i}$ for $1 \leq i \leq n$. According to Lemma~\ref{lem:fw-tw}
  % the treewidth of the cycle is at most~$2$ and thus by
  % Lemma~\ref{lem:tw-components} we obtain $\parm{inctw}({P^n_{51}}) \leq
  % 2$.
%
\item Consider the programs~$P^n_6$ and $P^n_7$. Its incidence graph
  contains a clique on $n$ vertices. Thus by definition
  $\parm{inctw}(P^n_6) \geq n-1$ and $\parm{inctw}(P^n_6) \geq n-1$.
%
  % Consider program~$P^n_7$. We observe that The incidence graph of
  % $P^n_7$ contains a clique on $n$ vertices as a subgraph. Hence
  % $\parm{deptw}(P^n_7) = n - 1$.
%
\item Consider program~~$P^{m,n}_8$. The incidence graph consists of a
  tree on the vertices $r_1, b,a_1,\ldots, a_m$ and a
  cycle~$r_1,c_1,\ldots,r_n,c_n,r_{n+1},c_{n+1},r_{n+2}$. By
  definition a tree has treewidth~$1$, according to
  Lemma~\ref{lem:fw-tw} a cycle has treewidth at most~$2$, and
  according to Lemma~\ref{lem:tw-components} we obtain $\parm{inctw}
  (P^{m,n}_8) \leq 2$.
\end{enumerate}
%\vspace{-5ex}
\end{OBS}

The following observation states why we cannot apply our lifting
theorem and extend the parameter treewidth from normal to disjunctive
programs.

\begin{OBS}\label{obs:inctw_enum}
  $\pnormal{\AspEnum[\parm{inctw}]}\not \in \FPT$.
\end{OBS}
\begin{proof}
  Consider the program~$P^n_{51}$ where $\parm{inctw}({P^n_{51}}) \leq
  2$.  Let $M \subseteq \at(P)$ such that either $a_i \in M$ or $b_i
  \in M$. According to the definitions we obtain the GL\hy
  reduct~$P^M:=\SB a_i \SM a_i \in M\SE \cup \SB b_i \SM b_i \in M
  \SE$. Since $M$ is a minimal model of $P^M$, $M$ is also an answer
  set of $P$. Thus the program~$P$ has $2^n$ many answer
  sets. Consequently, enumerating the answer sets of $P$ takes
  time~$\Omega(2^n)$.
\end{proof}

\begin{PRO}\label{pro:inctw}%
  Let $\CCC\in \{\Horn\} \cup \Acyc$ and $p \in \{\dbd_{\CCC},
  \parm{\#neg}, \parm{\#non-Horn}, \parm{lstr}, \parm{wfw} \}$, then
  $p$ and $\parm{inctw}$ are incomparable.
\end{PRO}
\begin{proof}
  We observe incomparability from the programs~$P^n_{51}$ and $P^n_6$
  where $p(P^n_{51}) \geq n$ and $\parm{inctw}(P^n_{51}) = 2$; and
  $p(P^n_6) \leq 1$ and $\parm{inctw}(P^n_6) \geq n-1$ by
  Observations~\ref{ex:bdclasses}, \ref{ex:num-neg}, \ref{ex:lstr},
  \ref{ex:wfw}, and \ref{ex:inctw}.
\end{proof}

\subsection{Dependency Treewidth}
One might ask whether it makes sense to consider restrictions on the
treewidth of the dependency graph. In this section we show that the
dependency treewidth strictly dominates the incidence treewidth and
backdoors with respect to the target class~$\C$, but unfortunately
parameterizing the main \ASP problems by the dependency treewidth does
not yield fixed-parameter tractability.

\begin{DEF}
  Let $P$ be a program, then $\parm{deptw}(P) = \tw(U_P)$. We call
  $\parm{deptw}(P)$ the \emph{dependency treewidth} of $P$.
\end{DEF}

\begin{OBS}\label{ex:deptw}
  We make the following observations about programs from
  Example~\ref{ex:comp_progs}.
%  We continue Example~\ref{ex:comp_progs} and observe:
%~
\begin{enumerate}
\item Consider programs~$P^n_{32}$ and $P^n_6$ where the dependency
  graph is a tree. Thus $\parm{deptw}(P^n_{32}) = \parm{deptw}(P^n_6)
  = 1$.
\item Consider program~$P^n_{51}$. We observe that its dependency
  graph consists of $n$ disjoint cycles~$b_i, v_{b_i,a_i}, a_i,
  v_{a_i,b_i}$ for $1 \leq i \leq n$. According to
  Lemma~\ref{lem:fw-tw} a cycle has treewidth at most~$2$ and
  according to Lemma~\ref{lem:tw-components} we obtain
  $\parm{deptw}(P^n_{51})\leq 2$.
%
% \item Consider program~$P^n_6$. Its dependency graph is a tree and
%   thus $\parm{deptw}(P^n_6) = 1$.
%
\item Consider program~$P^n_7$. Its dependency graph contains a clique
  on $n$ vertices as a subgraph. Hence $\parm{deptw}(P^n_7) \geq n -
  1$.
\end{enumerate}
%\vspace{-5ex}
\end{OBS}

\begin{PRO}\label{pro:deptw}
  $\parm{deptw}$ strictly dominates $\parm{inctw}$ and
  $\dbd_{\C}$. Let $\CCC \in \{\Horn\} \cup \Acyc \setminus \{\C,
  \EC\}$ and $p \in
  \{\dbd_\CCC, \parm{\#neg}, \parm{\#non-Horn}, \parm{lstr}, \parm{wfw}\}$,
  then $p$ and $\parm{deptw}$ are incomparable.
\end{PRO}
\begin{proof}
  Let $P$ be a normal program, and~$I_P$ its incidence graph. Let
  $(T,\chi)$ be an arbitrary tree decomposition of~$I_P$. We create a
  tree decomposition~$(T,\chi')$ for~$U_P$ as follows: For every~$r
  \in P$ let $v_r$ be the corresponding vertex in~$I_P$. We replace
  the occurrence of a $v_r\in \chi(t)$ by $H(r)$ for all nodes~$t \in
  V(T)$.
  % For every node $t \in V(T)$ we replace the occurrence $v_{r_i} \in
  % \chi(t)$ by $H(r_i)$, where the vertex $v_{r_i}\in V(\IG(P))$
  % corresponds to a rule $r_i \in P$.
  Then the pair~$(T,\chi')$ satisfies Condition~1 and 2 of a tree
  decomposition of~$U_P$. Since all edges of $I_P$ are covered
  in~$(T,\chi)$ for every~$r \in P$ exists a~$t \in V(T)$ such that
  $v_r \in \chi(T)$ and $h \in \chi(T)$ where $H(r)=\SB h
  \SE$. Because all $v_r$ are connected in the bags of the tree
  decomposition~$(T,\chi)$ and all corresponding elements~$h$ are
  connected in~$(T,\chi)$, the Condition~3 holds for~$(T,\chi')$. Thus
  $(T,\chi')$ is a tree decomposition of the dependency
  graph~$U_P$. Since the width of~$(T,\chi')$ is less or equal to the
  width of~$(T,\chi)$ it follows $\tw(U_P)\leq \tw(I_P)$ for a normal
  program~$P$. 
  To show that $\parm{deptw}$ strictly dominates $\parm{inctw}$,
  consider the program~$P^n_6$ where $\parm{deptw}(P^n_6) \leq 1$ and
  $\parm{inctw}(P^n_6) \geq n$. Hence $\parm{deptw}
  \prec \parm{inctw}$.

  Let $P$ be a normal program and $X$ a \delBds{\C} of $P$. Thus $X$
  is a feedback vertex set of the dependency graph~$U_P$. According to
  Lemma~\ref{lem:fw-tw} $\tw(U_P)\leq k+1$.
  % Hence $P-X \in \CCC$. Let $(T,\chi)$ and $T=(V_T,E_T)$ be a tree
  % decomposition of $U_{P-X}$ of size~$1$ which exists by definition
  % since $U_{P-X}$ is a forest. For every $t \in V_T$ we add to the
  % bag~$\chi(t)$ the atoms in $X$. Thus we have a trivial tree
  % decomposition of $U_{P-X}$ of size at most $\Card{X}+1$. 
  Hence $\parm{deptw} \preceq \dbd_\C$. To show that $\parm{deptw}$
  strictly dominates $\dbd_\C$ consider the program~$P^n_{51}$ where
  $\parm{deptw}(P^n_{51}) \leq 2$ and $\dbd_\C(P^n_{51}) \geq
  n$. Consequently, $\parm{deptw} \prec \dbd_\C$ and the proposition
  sustains.

  To show the last statement, consider again the programs~$P^n_{51}$
  and $P^n_7$ where $\parm{deptw}(P^n_{51}) \leq 2$ and $p(P^n_{51})
  \geq n$; and $\parm{deptw}(P^n_7) \geq n - 1$ and $p(P^n_7) = 0$ by
  Observations~\ref{ex:bdclasses}, \ref{ex:lstr}, \ref{ex:wfw}, and
  \ref{ex:deptw}.
\end{proof}

% However, we show that this restriction does not yield tractability.
% , as the reduction of~\cite{EiterGottlob95} produces programs with
% dependency graphs of treewidth~$2$.

\begin{PRO}
  For each~$L \in \AspReason$, $\pnormal{L}$ is \NP-hard, even for
  programs that have dependency treewidth~$2$.
\end{PRO}

\begin{proof}
  First consider the problem \pname{Consistency}. From a $3$-CNF
  formula~$F$ with $k$ variables we construct a program~$P$ as
  follows: Among the atoms of our program~$P$ will be two
  atoms~$a_{x}$ and $a_{\bar x}$ for each variable~$x\in \var(F)$ and
  a new atom~$f$. We add the rules~$a_{\bar x} \leftarrow \pnot a_x$
  and $a_x \leftarrow \pnot a_{\bar x}$ for each variable~$x \in
  \var(F)$. For each clause~$\{l_1,l_2,l_3\} \in F$ we add the rule~$f
  \leftarrow h(l_1), h(l_2), h(l_3), \pnot f$ where $h(\neg x) =
  a_{x}$ and $h(x) = a_{\bar x}$. Now it is easy to see that the
  formula~$F$ is satisfiable if and only if the program~$P$ has an
  answer set. Let~$U_P$ be the undirected dependency graph of $P$. We
  construct the following tree decomposition~$(T,\chi)$ for~$U_P$: the
  tree~$T$ consists of the node~$t_f$ and for each~$x \in \var(F)$ of
  the nodes~$t_{fx}$, $t_{x\bar x}$, and $t_{\bar x x}$ and the
  edges~$t_f t_{fx}$, $t_{fx} t_{x \bar x}$, and $t_{x \bar x} t_{\bar
    x x}$. We label the nodes by $\chi(t_f):=\{f,v_f\}$ and for
  each~$x \in \var(F)$ by $\chi(t_{fx}):= \{ a_x, a_{\bar x}, f\}$,
  $\chi(t_{x\bar x}) := \{a_x, a_{\bar x}, v_{a_x \bar a_x}\}$, and
  $\chi(t_{\bar x x}) := \{a_x, a_{\bar x}, v_{\bar a_x a_x}\}$. We
  observe that the pair~$(T,\chi)$ satisfies Condition~1. The rules
  $a_{\bar x} \leftarrow \pnot a_x$ and $a_x \leftarrow \pnot a_{\bar
    x}$ yield the edges~$a_xv_{a_x \bar a_x}$, $v_{a_x \bar a_x} \bar
  a_x$, $a_xv_{\bar a_x a_x}$, $v_{\bar a_x a_x} \bar a_x$ in $U_P$
  which are all ``covered'' by $\chi(t_{x \bar x})$ and $\chi(t_{\bar
    x x})$. The rule $f \leftarrow h(l_1), h(l_2), h(l_3), \pnot f$
  yields the edge~$fv_{f}$ which is covered by $\chi(t_f)$ and yields
  the edges~$f a_x$ or $fa_{\bar x}$ which are covered by
  $\chi(t_{fx})$. Thus Condition~2 is satisfied. We easily observe
  that Condition~3 also holds for the pair~$(T,\chi)$.  Hence
  $(T,\chi)$ is a tree decomposition of the dependency
  graph~$U_P$. Since $\max\SB |\chi(t)| -1 \SM t \in V(T) \SE = 2$,
  the tree decomposition~$(T,\chi)$ is of width~$2$ and
  $\parm{deptw}(P)=2$. Hence the problem~$\pnormal{
    \AspCons[ \parm{deptw} ] }$ is \NP-hard, even for programs that
  have dependency treewidth~$2$. We observe hardness for the problems
  \AspBrave and \AspCaut by the very same argument as in the proof of
  Theorem~\ref{the:non_poly_kernels} and the proposition holds.
\end{proof}

\subsection{Interaction Treewidth}\label{sec:tw_I}
%
% In this subsection we restrict ourselves to head-cycle-free programs
% (see Section~\ref{sec:hcf} for the definitions).

%
%ORIGINAL DEFINITION
%
% is a graph where each atom is associated with a vertex and for every
% atom~$a$, the set of all literals that appear in rules that have~$a$
% in their heads are connected as a clique.
%
% SIMILAR VERSION
%
% the graph that has as vertices the atoms~$\at(P)$ and for every
% atom~$a\in \at(P)$ a clique on the atoms in the set $C_a = \SB x \SM
% x \in H(r), a \in H(r), r \in P \SE \cup \SB x \SM x \in B^+(r), a
% \in H(r), r \in P \SE \cup \SB x \SM x \in B^-(r), a \in H(r), r \in
% P \SE$

\begin{DEF}[\citex{Ben-EliyahuDechter94}]\label{def:interact}
  Let $P$ be a normal program. The \emph{interaction graph} is the
  graph~$A_P$ which has as vertices the atoms of $P$ and an edge~$xy$
  between any two atoms~$x$ and $y$ for which there are rules~$r, r'
  \in P$ such that $x \in \at(r)$, $y \in \at(r')$, and $H(r) \cap
  H(r') \neq \emptyset$.\footnote{This definition is equivalent to the
    original definition in \shortcite{Ben-EliyahuDechter94} which is
    given in terms of cliques: the interaction graph is the graph
    where each atom is associated with a vertex and for every atom~$a$
    the set of all literals that appear in rules that have~$a$ in
    their heads are connected as a clique.
    % the graph that has as vertices the atoms~$\at(P)$ and for every
    % atom~$a\in \at(P)$ a clique on the atoms in the set $C_a = \SB x
    % \SM x \in H(r), a \in H(r), r \in P \SE \cup \SB x \SM x \in
    % B^+(r), a \in H(r), r \in P \SE \cup \SB x \SM x \in B^-(r), a \in
    % H(r), r \in P \SE$.
  }
\end{DEF}

\begin{DEF}[\citex{KanchanasutStuckey92}, \citex{Ben-EliyahuDechter94}]\label{def:pos_dep}
  Let $P$ be a program. The \emph{positive dependency digraph}~$D^+_P$
  of $P$ has as vertices the atoms~$\at(P)$ and a directed
  edge~$(x,y)$ between any two atoms~$x,y \in \at(P)$ for which there
  is a rule~$r\in P$ with $x\in H(r)$ and $y\in
  B^+(r)$.\footnote{\citex{Ben-EliyahuDechter94} used the term
    dependency graph while the term positive dependency graph was
    first used by \citex{KanchanasutStuckey92} and became popular
    by~\citex{ErdemLifschitz03}.}
\end{DEF}

% An graph~$G=(V,E)$ is \emph{chordal} if for every cycle~$c=(v_1,v_2,
% \ldots, v_l,v_{l+1})$ in $G$ of length~$l \geq 4$ there is an
% edge~$v_i,v_j \in E$ where~$v_i$ is not a neighbor of $v_j$ on the
% cycle~$c$, i.e., $v_i\notin N_C(v_j)$ for $C=\SB v_i \SM 1\leq i \leq
% l \SE$.

% A \emph{chord} is an edge joining two consecutive vertices in a cycle.

Let $G=(V,E)$ be a graph and $c=(v_1, \ldots, v_l)$ a cycle of
length~$l$ in $G$. A \emph{chord} of $c$ is an edge $v_iv_j \in E$
where $v_i$ and $v_j$ are not connected by an edge in $c$
(non-consecutive vertices). $G$ is \emph{chordal} (triangulated) if
every cycle in $G$ of length at least~4 has a chord.

\begin{DEF}[\citex{Ben-EliyahuDechter94}]
  Let $G$ be a digraph and $G'$ a graph. Then
  \begin{align*}
    % \text{length of a longest cycle in } G
    \mathop{lc}(G) :={}& \max\{ \{ 2 \} \cup \SB \Card{c} \SM c \text{
      is a cycle in } G \SE\},
    \\
    \mathop{cs}(G') :={}& \SB w \SM G' \text{ is a subgraph of a
      chordal graph with all cliques
      of size at most } w\SE \text{, and }\\
% \min \SB \Card{ C } \SM C \text{ is a maximal
%       clique in a chordal graph containing G'}\SE.
    \parm{fw}(G'):={}&\min \SB \Card{S} \SM S \text{ is a feedback 
      vertex set of } G' \SE.
  \end{align*}
  $\mathop{lc}$ is the \emph{length of the longest
    cycle}. $\mathop{cs}$ is the \emph{clique size}.\footnote{The
    original definition is based on the length of the longest acyclic
    path in any component of $G$ instead of the length of the longest
    cycle and the term clique width is used instead of clique size.}

  Let $P$ be a normal program, $A_P$ its interaction graph, and
  $D_P^+$ its positive dependency digraph. Then 
  \begin{align*}
    \parm{cluster}(P) :={}& \mathop{cs}(A_P) \cdot \log
    \mathop{lc}(D_P^+)\\
    \parm{cyclecut}(P) :={}& \parm{fw}(A_P) \cdot \log
    \mathop{lc}(D_P^+).
  \end{align*}
  $\parm{cluster}(P)$ is called the \emph{size of the tree
    clustering}.  $\parm{cyclecut}(P)$ is called the \emph{size of the
    cycle cutset decomposition}.
\end{DEF}

In fact the definition of $\mathop{cs}(G)$ is related to the
treewidth:

\begin{LEM}[\citey{RobertsonSeymour86}]\label{lem:cs-tw}
  Let $G$ be a graph. Then $\tw(G) = \mathop{cs}(G) + 1$.
\end{LEM}
\begin{COR}
  Let $P$ be a normal program, $A_P$ its interaction graph, and
  $D_P^+$ its dependency digraph. Then
  \begin{align*}
    \parm{cluster}(P) = {}& \left( \tw(A_P) - 1\right) \cdot \log
    \mathop{lc}(D_P^+)\\
  \end{align*}
\end{COR}
\begin{PRO}[\citex{Ben-EliyahuDechter94}]\label{pro:cluster}
  For each~$L\in \AspFull$, $\pnormal{L[\parm{cluster}]}\in \FPT$ and
  $\pnormal{L[\parm{cyclecut}]}\in \FPT$.
\end{PRO}

\begin{OBS}\label{ex:cyclecut}
  We make the following observations about programs from
  Example~\ref{ex:comp_progs}.
%  We continue Example~\ref{ex:comp_progs} and observe:
%~
\begin{enumerate}
\item Consider programs~$P^n_{51}$ and $P^n_{53}$ and let $P \in
  \{P^n_{51}, P^n_{53}\}$. The interaction graph~$A_p$ contains $n$
  disjoint paths $a_i,b_i$ for $1 \leq i \leq m$. Hence $A_P$ contains
  no cycles and $\parm{fw}(A_P) = 0$ and according to
  Lemma~\ref{lem:fw-tw} we obtain $\tw(A_P) \leq 1$.
  % every connected component in $A_P$ consists only of an edge~$a_ib_i$
  % which is obviously a tree. From the definitions one easily observes
  % that a tree has treewidth~1. According to Lemma~\ref{lem:cs-tw} we
  % obtain $\mathop{cs}(A_P) = 1$. Since $A_P$ contains no cycles
  % $\parm{fw}(A_P) = 0$.
  Moreover, the positive dependency graph~$D^+_P$ contains no edges,
  $n$ disjoint cycles of length exactly~$2$ respectively.
 % if
 %  $P = P^n_{51}$ and if
 %  $P^n_{53}$. 
  Thus $\mathop{lc} (D^+_{P}) = 2$. Consequently,
  $\parm{cluster}(P^n_{51}) \leq 1$ and $\parm{cyclecut}(P^n_{51})
  \leq 1$; and $\parm{cluster}(P^n_{53}) \leq 1$ and
  $\parm{cyclecut}(P^n_{53}) \leq 1$.

\item Consider program~$P^{m,n}_8$ and let $P:=P^{m,n}_8$. The
  interaction graph~$A_{P}$ contains a clique on $m$ vertices and thus
  $\tw(A_{P}) \geq m - 1$. According to Lemma~\ref{lem:cs-tw} we
  obtain $\mathop{cs}(A_P) \geq m - 2$. According to
  Lemma~\ref{lem:fw-tw} we have $\parm{fw} (A_P) \geq m -
  2$. Moreover, the positive dependency graph~$D^+_P$ contains the
  cycle~$c_1,c_2, \ldots, c_n,c_{n+1}$. Thus $\mathop{lc} (D^+_{P}) =
  n$. Consequently, $\parm{cluster}(P^{m,n}_8) \geq (m - 2) \cdot \log
  n$ and $\parm{cyclecut}(P^{m,n}_8) \geq (m - 2) \cdot \log n$.
\end{enumerate}
%\vspace{-5ex}
\end{OBS}

\begin{OBS}
  \parm{cluster} strictly dominates \parm{cyclecut}.
\end{OBS}
\begin{proof}
  Let $P$ be a normal program and $A_P$ its interaction
  graph. According to Lemma~\ref{lem:fw-tw} we obtain $\tw(A_P)
  \leq \parm{fw} (A_P) + 1$. Hence $\parm{cluster}(P)
  \prec \parm{cyclecut} (P)$.
\end{proof}

\begin{PRO}\label{pro:hcf1}
  $\parm{inctw}$ strictly dominates $\parm{cluster}$. Let $\CCC \in
  \{\Horn\} \cup \Acyc$ and $p \in \{\dbd_\CCC$, \parm{\#neg},
  \parm{\#non-Horn}, \parm{lstr}, $\parm{wfw}\}$, then $p$ and
  $\parm{cluster}$ are incomparable; and $p$ and $\parm{cyclecut}$ are
  incomparable.
\end{PRO}
\begin{proof}
  We first show that $\parm{inctw}$ dominates $\parm{cluster}$. Let
  $P$ be a normal program, $I_P$ its incidence graph, and $A_P$ its
  interaction graph. Let $(T,\chi)$ be an arbitrary tree decomposition
  of~$A_P$. We create a tree decomposition~$(T,\chi')$ for~$I_P$ as
  follows: For every~$r \in P$ let $v_r$ be the corresponding vertex
  in~$I_P$. By definition for every~$r \in P$ there is a bag $\chi(t)$
  where $t \in V(T)$ such that $\at(r) \subset \chi(t)$. We set
  $\chi'(t) = \chi(t) \cup \{ v_r \}$. Then the pair~$(T,\chi')$
  clearly satisfies Condition~1 and 2 of a tree decomposition of~$I_P$
  by definition. Since every $v_r$ occurs in exactly one bag
  Condition~3 holds for~$(T,\chi')$. Thus $(T,\chi')$ is a tree
  decomposition of the interaction graph~$A_P$. Since the width
  of~$(T,\chi')$ is less or equal to the width of~$(T,\chi)$ plus one
  it follows $\tw(I_P)\leq \tw(A_P) + 1$.
  To show that $\parm{inctw}$ strictly dominates $\parm{cluster}$,
  consider the program~$P^{m,n}_8$ where $\parm{inctw} (P^{m,n}_8)
  \leq 2$ and $\parm{cluster}(P^{m,n}_8) = (m-2) \log n$ by
  Observations~\ref{ex:inctw} and \ref{ex:cyclecut}. Hence
  $\parm{inctw} \prec \parm{cluster}$.

  Let $p \in\{\dbd_\CCC,$ \parm{\#neg}, \parm{\#non-Horn},
  \parm{lstr}, $\parm{wfw}\}$ and $\CCC \in \{\Horn\} \cup \Acyc$. We
  show the incomparability of the parameter~$p$ and
  $\parm{cyclecut}$. In fact we show something stronger, there are
  programs~$P$ where $p$ is of constant size, but both
  $\parm{tw}(D^+_P)$, $\parm{fw}(D^+_P)$ respectively, and
  $\mathop{lc}(I_P)$ can be arbitrarily large and there are programs
  where the converse sustains. Therefor we consider the
  programs~$P^n_{51}$ and $P^{m,n}_8$ where $p(P^n_{51}) \geq n$ and
  $\parm{cluster}(P^n_{51}) \leq 1 $ and $\parm{cyclecut}(P^n_{51})
  \leq 1$; and $p(P^{m,n}_8) \leq 1$ and $\parm{cyclecut}(P^{m,n}_8)
  \geq (m-2) \cdot \log n$ and $\parm{cluster}(P^{m,n}_8) \geq (m-2)
  \cdot \log n$ by Observations~\ref{ex:bdclasses}, \ref{ex:num-neg},
  \ref{ex:lstr}, \ref{ex:wfw}, and \ref{ex:cyclecut}. Consequently,
  the second statement holds.
\end{proof}

\subsection{Number of Bad Even Cycles}
% In the literature, parameters have been considered that are based on
% counting the number of various kinds of cycles in graph
% representations of a given program. In the following we will compare
% those parameters with smallest backdoors.

\begin{DEF}[\citey{LinZhao04}]
  Let $P$ be a normal program. Then
  \begin{align*}
    \parm{\#badEvenCycles}(P)  := {}&\Card{\SB c \SM c \text{ is a directed
        bad even cycle of } P \SE}
  \end{align*}
\end{DEF}

\begin{PRO}
  For each $L \in \AspFull$, $\pnormal{L[\parm{\#badEvenCycles}]} \in
  \FPT$.
\end{PRO}

\begin{OBS}\label{ex:num-badevencycles}
  We make the following observations about programs from
  Example~\ref{ex:comp_progs}.
%  We continue Example~\ref{ex:comp_progs} and observe:
%~
\begin{enumerate}
\item
  Consider program~$P^n_4$ which contains no directed bad even
  cycle. Hence $\parm{\#badEvenCycles}(P^n_4) = 0$.
\item Consider program~$P^n_{51}$ which contains $n$ disjoint directed
  bad even cycles. Thus $\parm{\#badEvenCycles}(P^n_{51}) = n$.
\item Consider programs~$P^n_{52}$, $P^n_7$, and $P^{m,n}_8$ which
  contain no directed bad even cycle. Consequently we obtain
  $\parm{\#badEvenCycles}(P^n_{52}) = \parm{\#badEvenCycles}(P^n_7)
  = \parm{\#badEvenCycles}(P^{m,n}_8)=0$.
% %
% \item
%   Consider program~ which contain no directed
%   bad even cycle. Hence $\parm{\#badEvenCycles}(P^n_7)
%   = \parm{\#badEvenCycles}(P^{m,n}_8) = 0$.
%
\item Consider program~$P^n_9$ which contains the directed bad even
  cycles~$a_1, a_2, a_3, b_i$ for $1 \leq i \leq n$. Since there are
  $n$ of those directed bad even cycles we obtain
  $\parm{\#badEvenCycles} (P^n_9) = n$.
\end{enumerate}
%\vspace{-5ex}
\end{OBS}

\begin{PRO}\label{pro:zhao}
  $\dbd_{\DBEC}$ strictly dominates $\parm{\#badEvenCycles}$.
  Moreover, $\dbd_\CCC$ and $\parm{\#badEvenCycles}$ are incomparable
  for the remaining target classes~$\CCC \in \Acyc \setminus \{\DBEC\}
  \cup \{ \Horn\}$.
  Let $p \in \{\parm{\#neg},$ \parm{\#non-Horn}, \parm{lstr},
  \parm{wfw}, \parm{inctw}, \parm{deptw}, \parm{cluster},
  $\parm{cyclecut}\}$, then $p$ and $\parm{\#badEvenCycles}$ are
  incomparable.
\end{PRO}
\begin{proof}
  To see that $\dbd_\DBEC$ strictly dominates
  $\parm{\#badEvenCycles}$. Let $P$ be a normal program. If $P$ has at
  most~$k$ directed bad even cycles, we can construct a
  \delBds{\DBEC}~$X$ for $P$ by taking one element from each directed
  bad even cycle into $X$. Thus $\dbd_\DBEC(P)
  \leq \parm{\#badEvenCycles}(P)$. If a program~$P$ has a
  \delBds{\DBEC} of size~$1$, it can have arbitrarily many even cycles
  that run through the atom in the backdoor, e.g. program~$P^n_9$
  where $\dbd_\DBEC(P^n_9) \leq 1$ and $\parm{\#badEvenCycles}(P^n_9)
  = n$ by Observations~\ref{ex:bdclasses} and
  \ref{ex:num-badevencycles}. It follows that $\dbd_\DBEC
  \prec \parm{\#badEvenCycles}$ and the proposition holds.

  To show the second statement, consider the programs~$P^n_4$,
  $P^n_{52}$, and $P^n_9$ where $\dbd_\CCC(P^n_{9}) = 1$ for $\CCC \in
  \Acyc \cup \{\Horn\}$ and $\parm{\#badEvenCycles}(P^n_{9}) = n$;
  conversely $\dbd_\CCC(P^n_4) \geq n$ for $\CCC \in \{\Horn$, \C,
  \BC, \DC, \DCTWO, \EC, \DEC, $\BEC\}$, $\dbd_\DBC(P^n_{52}) \geq n$,
  and $\parm{\#badEvenCycles}(P^n_4) = \parm{\#badEvenCycles}(
  P^n_{52}) = 0$. Hence $\dbd_\CCC \bowtie \parm{\#badEvenCycles}$ for
  $\CCC \in \Acyc \setminus \{\DBEC\} \cup \{\Horn\}$ by
  Observations~\ref{ex:bdclasses} and \ref{ex:num-badevencycles}.

  To show the third statement, consider the programs~$P^n_{51}$,
  $P^n_{52}$, $P^n_7$, and $P^{m,n}_8$, $P^n_9$ where
  $\parm{inctw}(P^n_7) \geq n-1$ and $\parm{deptw}(P^n_7) \geq n - 1
  $, $p(P^n_{52}) \geq n$ for $p \in
  \{\parm{\#neg}$, \parm{\#non-Horn}, \parm{lstr}, $\parm{wfw}\}$,
  $\parm{cyclecut}(P^{m,n}_8) \geq (m-2) \log n$,
  $\parm{cluster}(P^{m,n}_8) \geq (m-2) \log n$, and
  $\parm{\#badEvenCycles}(P^n_7) = \parm{\#badEvenCycles} (P^{m,n}_8)
  = \parm{\#badEvenCycles}(P^n_{52}) = 0$; conversely $p(P^n_{51})
  \leq 2$ for $p \in \{\parm{inctw}$, \parm{deptw}, \parm{cluster},
  $\parm{cyclecut}\}$, $p(P^n_{9}) \leq 2$ for $p \in
  \{\parm{\#neg}$, \parm{\#non-Horn}, \parm{lstr}, $\parm{wfw}\}$, and
  $\parm{\#badEvenCycles}(P^n_{51}) = \parm{\#badEvenCycles}(P^n_{9})
  = n$ by Observations~\ref{ex:bdclasses}, \ref{ex:num-neg},
  \ref{ex:lstr}, \ref{ex:wfw}, \ref{ex:inctw}, \ref{ex:deptw},
  \ref{ex:cyclecut}, and \ref{ex:num-badevencycles}. Hence $p
  \bowtie \parm{\#badEvenCycles}$ for $p \in \{\parm{\#neg}$,
  \parm{\#non-Horn}, \parm{lstr}, \parm{wfw}, \parm{inctw}, 
  \parm{deptw}, \parm{cluster}, $\parm{cyclecut}\}$.

\end{proof}

% In particular, it follows that the reduction that maps a program to
% itself and replaces the parameter~$\param{\#evenCycles}$ with
% $\sbd_{\tautext{\CCC}}$, $\dbd_{\tautext{\CCC}}$ respectively,
% provides a trivial fpt-reduction with respect to any decision
% problem on programs. The reverse mapping that replaces
% $\sbd_{\tautext{\CCC}}$, $\dbd_{\tautext{\CCC}}$ respectively, with
% $\param{\#evenCycles}$, however, is not an
% fpt-reduction. Consequently, whenever a problem is fixed-parameter
% tractable for parameter~$\sbd_{\tautext{\CCC}}$,
% $\dbd_{\tautext{\CCC}}$ respectively, the program is fixed-parameter
% tractable for parameter~$\param{\#evenCycles}$, but the converse is
% not necessarily always true.

\subsection{Number of Positive Cycles (Loop Formulas)}\label{sec:loop}

\begin{DEF}[\citey{Fages94}]
  Let $P$ be a normal program and $D^+_P$ its positive dependency
  digraph. Then
  \begin{align*}
    \parm{\#posCycles} :={}& \Card{\SB c \SM c \text{ is a directed
        cycle in } D^+_P \SE}
  \end{align*}
  The program~$P$ is called \emph{tight} if $\parm{\#posCycles} =
  0$.\footnote{\citex{Fages94} used the term positive-order consistent
    instead of tight.}
\end{DEF}

The parameter has been generalized to disjunctive programs
by~\citex{LeeLifschitz03}.

\begin{PRO}[\citey{Fages94}]
  For $L \in \AspReason$, $\pnormal{L[\parm{\#posCycles}]}$ is
  \NP-hard or \coNP-hard, even for tight programs.
\end{PRO}

% answer sets of a tight program are exactly the models of Clark's
% completion~\cite{Clark78}

% The following proposition establishes that the parameter is
% incomparable to strong (deletion) backdoors:

\begin{OBS}\label{ex:num-poscycles}
  We make the following observations about programs from
  Example~\ref{ex:comp_progs}.
%  We continue Example~\ref{ex:comp_progs} and observe:
%~
\begin{enumerate}
\item Consider programs~$P^n_{32}$ and $P^n_{53}$ where the positive
  dependency digraphs contain $n$ directed cycles, hence
  $\parm{\#posCycles} (P^n_{32})= \parm{\#posCycles}(P^n_{53}) = n$.
\item Consider program~$P^n_{51}$ and $P^n_7$ where the positive
  dependency digraphs contain no cycle. Hence
  $\parm{\#posCycles}(P^n_{51}) = \parm{\#posCycles}(P^n_7) = 0$.
\item Consider program~$P^{m,n}_8$. Its positive dependency digraph
  contains only the cycle~$c_1,c_2, \ldots, c_n, c_{n+1}$, thus
  $\parm{\#posCycles} (P^n_8) = 1$.
\end{enumerate}
%\vspace{-5ex}
\end{OBS}

%   Conversely, consider the programs~$P^n_2$ from the proof of
%   Proposition~\ref{pro:hcf1} and $P^n_{31} = \SB a_i \leftarrow b_1,
%   \ldots, b_n \SM 1 \leq i \leq n\SE$ where $n$ is large. 

% Since
%   $P^n_2$ contains $n$ disjoint directed bad even cycles,
%   $\dbd_\CCC(P^n_2) = \sbd_\CCC(P^n_2) = \parm{wfw} (P^n_2)
%   = \parm{\#badEvenCycles} (P^n_2) = n$ for $\CCC \in \{ \Horn \} \cup
%   \Acyc$. Moreover, $ \parm{\#neg} (P^n_2) = \parm{\#non-Horn} (P^n_2)
%   = \parm{lstr} (P^n_2) = 2n$. In fact $\parm{deptw} (P^n_{31})
%   = \parm{inctw} (P^n_{31}) = n - 1$. Futher, $\parm{inctw}(P^n_{31})
%   \leq \parm{cluster} (P^n_{31}) + 1$ and $\parm{inctw}(P^n_{31})
%   \leq \parm{cyclecut}(P^n_{31}) + 1$ (see proof of
%   Proposition~\ref{pro:hcf1}). However, neither $D^+_{P^n_2}$ nor
%   $D^+_{P^n_{31}}$ contains a directed cycle. Hence, $\parm{\#posCycles}
%   (P^n_2) = \parm{\#posCycles} (P^n_{31}) = 0$ and the proposition
%   sustains.

\begin{PRO}
  Let $\CCC \in \{\Horn\} \cup \Acyc$ and $p \in \{\dbd_\CCC$,
  \parm{\#neg}, \parm{\#non-Horn}, \parm{lstr}, \parm{wfw}, \parm{inctw},
  \parm{deptw}, \parm{cluster}, \parm{cyclecut},
  $\parm{\#badEvenCycles} \}$, then $p$ and $\parm{\#posCycles}$ are
  incomparable.
\end{PRO}
\begin{proof}
  % Let $\CCC \in\{\Horn\} \cup \Acyc$ be a class of programs and $p \in
  % \{\dbd_\CCC,\sbd_\CCC$, \parm{\#neg}, \parm{\#non-Horn}, \parm{lstr},
  % \parm{wfw}, \parm{inctw}, \parm{deptw}, \parm{cluster},
  % \parm{cyclecut}, $\parm{\#badEvenCycles}\}$. We show that there are
  % programs that have a constant number of directed cycles in the
  % positive dependency digraph but the parameter~$p$ can be arbitrarily
  % large and there are programs where the converse sustains.
%
%
  We observe incomparability from the programs~$P^n_{32}$, $P^n_{51}$,
  $P^n_{53}$, $P^n_7$, and $P^{n,m}_8$. We have $p(P^n_{51}) \geq n$
  for $p \in \{\dbd_\CCC$
  \parm{\#neg}, \parm{\#non-Horn}, \parm{lstr}, \parm{wfw},
  $\parm{\#badEvenCycles}\}$, $\parm{inctw} (P^n_7) \geq n-1$,
  $\parm{deptw}(P^n_7) \geq n - 1$, $\parm{cyclecut} (P^{n,m}_8)\geq
  (m-2) \cdot \log n$, $\parm{cluster}(P^{n,m}_8) \geq (m-2) \cdot
  \log n$, and $\parm{\#posCycles}(P^n_{51}) = \parm{\#posCycles}
  (P^n_7) = 0$ and $\parm{\#posCycles}(P^{m,n}_8) = 1$; conversely for
  $p \in \{\dbd_\CCC$, \parm{\#neg}, \parm{\#non-Horn}, \parm{lstr},
  \parm{wfw}, \parm{inctw}, $\parm{deptw}\}$ we have $p(P^n_{32})
  \leq 1$, for $p \in \{\parm{cluster}, \parm{cyclecut}\}$ we have
  $p(P^n_{53}) \leq 2$ and $\parm{\#posCycles}(P^n_{32})
  = \parm{\#posCycles}(P^n_{53}) = n$ by
  Observations~\ref{ex:bdclasses}, \ref{ex:num-neg}, \ref{ex:lstr},
  \ref{ex:wfw}, \ref{ex:inctw}, \ref{ex:deptw}, \ref{ex:cyclecut},
  \ref{ex:num-badevencycles}, and
  \ref{ex:num-poscycles}. Consequently, the proposition holds.
\end{proof}

\subsection{Head-Cycles}\label{sec:hcf}

\begin{DEF}[\citex{Ben-EliyahuDechter94}]
  Let $P$ be a program and $D_P^+$ its positive dependency digraph. A
  \emph{head-cycle} of $D_P^+$ is a $\{x,y\}$\hy cycle\footnote{See
    Section~\ref{sec:db-cycl} for the definition of a $W$\hy cycle.}
  where $x,y \in H(r)$ for some rule~$r \in P$. The program~$P$ is
  \emph{head-cycle-free} if $D^+_P$ contains no head-cycle.
\end{DEF}

One might consider the number of head-cycles as a parameter to
tractability.
\begin{DEF}
  Let $P$ be a program and $D_P^+$ its positive dependency
  digraph. Then
  \begin{align*}
    \parm{\#headCycles} := {}& \Card{\SB c \SM c \text{ is a
        head-cycle of } D_P^+ \SE}
  \end{align*}
\end{DEF}

But as the following proposition states that the \ASP-reasoning
problems are already \NP-complete for head-cycle-free programs.
\begin{PRO}[\citex{Ben-EliyahuDechter94}]
  Each $L \in \AspReason$ is \NP-hard or \coNP-hard, even for
  head-cycle-free programs.
\end{PRO}

\begin{samepage}
  \begin{OBS}\label{ex:num-headCycles}
  We make the following observations about programs from
  Example~\ref{ex:comp_progs}.
%  We continue Example~\ref{ex:comp_progs} and observe:
%~
\begin{enumerate}
\item
  Consider program~$P^n_{51}$. Since the positive dependency digraph
  of $P^n_{51}$ contains no cycle, $\parm{\#headCycles}(P^n_{51}) = 0$. 
\item
  Consider program~$P^n_{11}$. The positive dependency digraph of
  $P^n_{11}$ contains the head cycles~$a_ibc$ for $1 \leq i \leq
  n$. Thus $\parm{\#headCycles}(P^n_{11}) = n$.
\end{enumerate}
%\vspace{-5ex}
\end{OBS}
\end{samepage}

Even though the parameter~$\parm{\#headCycles}$ does not yield
tractability for the \ASP-reasoning problems we are interested in the
relationship between our lifted parameters and the
parameter~$\parm{\#headCycles}$. We will first restrict the input
programs to normal programs in Observation~\ref{obs:normal_headCycles}
and then consider disjunctive programs
Observation~\ref{obs:headCycles}.

\begin{OBS}\label{obs:normal_headCycles}
  Let $\CCC \in \{\Horn\} \cup \Acyc$ and $p \in
  \{\dbd_\CCC$, \parm{\#neg}, \parm{\#non-Horn}, \parm{lstr},
  \parm{wfw}, \parm{inctw}, \parm{deptw}, \parm{cluster}, \parm{cyclecut},
  \parm{\#badEvenCycles}, $\parm{\#posCycles}\}$, then
  $\parm{\#headCycles}$ strictly dominates $p$.
\end{OBS}
\begin{proof}
  By definition every normal program is head-cycle-free, hence
  $\parm{\#headCycles}$ strictly dominates $p$.
\end{proof}

\begin{OBS}\label{obs:headCycles}
  Let $\CCC \in \{\Horn\} \cup \Acyc$ and $p \in
  \{\dbd_\CCC,\parm{\#neg},\parm{\#non-Horn},\parm{lstr},\parm{wfw}\}$,
  then $\lift{p}$ and $\parm{\#headCycles}$ are incomparable.
\end{OBS}
\begin{proof}
  To that the parameters are incomparable consider the
  programs~$P^n_{51}$ and $P^n_{11}$ where $p(P^n_{51}) \geq n$ and
  $\parm{\#headCycles}(P^n_{51}) = 0$; and $p(P^n_{11}) = 1$ and
  $\parm{\#headCycles}(P^n_{11}) = n$ by
  Observations~\ref{ex:bdclasses}, \ref{ex:num-neg}, \ref{ex:lstr},
  \ref{ex:wfw}, and \ref{ex:num-headCycles}.
\end{proof}

\section{Practical Considerations}\label{sec:experiments}
Although the main focus of this paper is theoretical, we discuss in
this section some practical considerations and present some empirical
data.
\subsection{Backdoor Detection} 
We have determined \strongBds{\tautext{\Horn}}s for various benchmark
programs by means of encodings into answer set programming, integer
linear programming (\ILP), local search (LS), and propositional
satisfiability. It turned out that compilations into \ILP and \ASP
itself perform best. The integer linear program was generated using
the open source mathematics framework Sage~\cite{SteinEtAl12} with
Python~\cite{Rossum95}, solved using ILOG CPLEX~12~\cite{cplex11} and
Gurobi~\cite{Gurobi-Optimization14}. We did not check optimality
(considering LP duality gap and branch and bound tree). Hence the
found \strongBds{\tautext{\Horn}}s might be not optimal, but
presumably close to optimal. For some selected instances we verified
optimality using a \SAT solver and unary cardinality
constraints~\cite{Sinz05}. The answer set program that solves backdoor
detection was generated by means of \ASP meta
programming~\cite{GebserKaminskiSchaub11} and solved using
clasp~\cite{potassco13} and a variant
(unclasp)~\cite{AndresKaufmannMattheisSchaub12}.

\begin{table}[t]
\centering
\begin{tabular}{llcrrrrrrr}
  \toprule
  domain & instance set &	disj. & \#atoms &	horn bd(\%) &	stdev\\
  \midrule
  \textit{AI} &\texttt{HanoiTower} & -- & 32956.7 &	4.28 &	0.08 \\
  &\texttt{StrategicCompanies} &	+& 2002.0 & 	6.03 &	0.04 \\
  &\texttt{MinimalDiagnosis} &	+& 111856.5 &	10.74 &	1.72 \\
  \textit{Graph} &\texttt{GraphColoring} &	--& 3544.4 & 19.47 &	0.80 \\
  \textit{Planning} 
%&\texttt{HydraulicPlanning} &	--& 538128.3 &	0.0 &	0.0 \\
  &\texttt{MSS/MUS} &	+& 49402.3 & 3.80 &	0.70 \\
  &\texttt{ConformantPlanning} &	+& 1378.2 &	8.76 &	2.14 \\
  \textit{Cryptography}&\texttt{Factoring} &	--& 3336.8 &	16.20 &	1.30 \\
  \textit{Puzzle} &\texttt{Labyrinth}&	--& 55604.9 &	3.42 &	0.82\\
  &\texttt{KnightTour} &	--& 23156.9 &	33.08 &	0.20 \\
  &\texttt{Solitaire} &	--& 11486.8 &	38.88 &	0.20\\
 % \textit{Crafted}
%  &\texttt{EqTest} &	--& 174.5 &	42.0 &	3.0 \\
%  &\texttt{DisjunctiveLoops} & +& 501.0 &	0.0 & 0.0 \\
%  &\texttt{Mutex} &	+& 937.0 &	50.0 &	0.0 \\
  \midrule
  \midrule
  %\cmidrule
  \textit{Random} &\texttt{RandomQBF} &	+& 160.1 &	50.00 &	0.00 \\
  &\texttt{RLP} & --& 184.2 &	68.00 &	5.00 \\
  &\texttt{RandomNonTight} & -- & 50.0 & 93.98 &	1.08 \\
  \bottomrule
\end{tabular}
\caption{Size of smallest strong $\Horn$\hy backdoors (bd) for various benchmark 
  sets, given as \% of the total number of atoms (\#atoms) by the mean over the instances. \\
  \small{  
    \texttt{ConformantPlanning}: secure planning under incomplete
    initial states~\protect\cite{ToPontelliSon09} instances provided by
    Gebser and Kaminski~\protect\shortcite{GebserKaminski12}.
    \texttt{Factoring}: factorization of a number where an efficient algorithm
    would yield a cryptographic attack by
    Gebser~\protect\shortcite{DrescherGebserKaufmannSchaub10} instances
    provided by Gebser~\protect\shortcite{Asparagus09}.
    \texttt{HanoiTower}: classic Towers of Hanoi puzzle by Truszczynski, Smith
    and Westlund; for instances see~\protect\shortcite{CalimeriEtAl11}.
    \texttt{GraphColoring}: classic graph coloring problem by Lierler and  
    Balduccini; for instances see~\protect\shortcite{CalimeriEtAl11}.
    \texttt{KnightTour}: finding a tour for the knight piece travelling any square 
    following the rules of chess by Zhou, Calimeri, and Santoro; for instances
    see~\protect\shortcite{CalimeriEtAl11}.
    \texttt{Labyrinth}: classical Ravensburger's Labyrinth puzzle by Gebser;
    for instances see~\protect\shortcite{CalimeriEtAl11}.
    \texttt{MinimalDiagnosis}: an application in systems
    biology~\protect\cite{GebserSchaubThieleUsadelVeber08}; for instances
    see~\protect\shortcite{CalimeriEtAl11}.
    \texttt{MSS/MUS}: problem whether a clause belongs to some minimal unsatisfiable 
    subset~\protect\cite{JanotaMarques-Silva11} instances provided by Gebser and 
    Kaminski~\protect\shortcite{GebserKaminski12}.   
    \texttt{Solitaire}: classical Peg Solitaire puzzle by Lierler and Balduccini;
    for instances see~\protect\shortcite{CalimeriEtAl11}.
    \texttt{StrategicCompanies}: encoding the $\Sigma^P_2$-complete 
    problem of producing and owning companies and strategic sets 
    between the companies~\protect\cite{GebserEtAl07}.
    \texttt{Mutex}: equivalence test of partial implementations of
    circuits, instances provided by~\protect\citex{MarateaRiccaFaberLeone08} 
    based on QBF instances of~\protect\citex{AyariBasin00}.
    \texttt{RandomQBF}: translations of randomly generated $2$-QBF
    instances using the method by Chen and
    Interian~\protect\shortcite{ChenInterian05} instances provided by
    Gebser~\protect\shortcite{GebserEtAl07}.
    \texttt{RLP}: Randomly generated normal programs, of various density (number of
    rules divided by the number of atoms)~\protect\cite{ZhaoLin03} instances provided
    by~\protect\cite{GebserEtAl07}. 
    \texttt{RandomNonTight}: Randomly generated normal programs
    provided by Schultz and Gebser~\cite{Asparagus09} with $n=40$, $50$,
    and $60$ variables, respectively with $40$ instances per step instances 
    provided by~\protect\citex{Asparagus09}.
  }
}
\label{tab:horn}
\end{table}
%
%
%

% \footnotetext{We are aware that one can preprocess extended rules
%   and compile them into normal rules. Even though recent versions of
%   the solver clasp provide such an option~\cite{gringo}, those
%   compilations blow up the instances significantly. Hence we omitted
%   it for pragmatic reasons.
%   % and one runs easily out of memory.  effect the solution process of
%   % a solver later on.  Hence, we think
%   % it is not helpful to correlate backdoors and the internal solution
%   % process of solvers or utilize them in heuristics.
% }.
%\paragraph{Results}
Table~\ref{tab:horn} illustrates our results on the size of small
\strongBds{\tautext{\Horn}}s of the considered benchmark instances.
We mainly used benchmark sets from the first three Answer Set
Programming
Competitions~\cite{CalimeriEtAl11,DeneckerEtAl09,GebserEtAl07},
because most of the instances contain only normal and/or disjunctive
rules and no extended rules
(cardinality/weight-constraints)\footnote{We are aware that one can
  preprocess extended rules and compile them into normal rules. Even
  though recent versions of the solver clasp provide such an
  option~\cite{gringo}, those compilations blow up the instances
  significantly. Hence we omitted it for pragmatic reasons.}. The
structured instances have, as expected, significantly smaller
\strongBds{\tautext{\Horn}}s than the random instances.  So far we
have no good evidence why in particular the sets \texttt{KnightTour}
and \texttt{Solitaire} have rather large \strongBds{\tautext{\Horn}}s
compared to the other structured instances.

For the acyclicity based target classes~$\CCC\in \Acyc$ we have
computed small \delBds{\CCC}s only for very few selected instances
with moderate size since the currently available algorithms can only
deal with rather small instances within a reasonable computation
time. The size of small \delBds{\tautext{\C}}s of selected instances
of \texttt{Solitaire} was about half of the size of small
\strongBds{\tautext{\Horn}}s.

\subsection{Backdoor Evaluation}
Instead of applying the algorithm from Section~\ref{sec:backdoors}
directly, one can possibly use backdoors to control modern heuristics
in \ASP solvers to obtain a speed-up. Most modern solver heuristics
work independently from the current truth assignment. They assign to
each atom in the program a score and incorporate into the score the
learned knowledge based on derived conflicts (history of the truth
assignments). Various studies on the effect of restricting decision
heuristics to a subset of variables based on structural properties
have been carried out in the context of \SAT, both positive
~\cite{GiunchigliaMassarottoSebastiani98,GiunchigliaMarateaTacchella02,Strichman00}
and negative effects~\cite{JarvisaloNiemela08} have been observed
depending on the domain of the instances. \citex{JarvisaloJunttila09}
have proven that a very restricted form of branching (branch only on a
subset of the input variables) implies a super-polynomial increase in
the length of the optimal proofs for learning-based
heuristics. However, very recent results by
\citex{GebserKaufmannOteroRomeroSchaubWanko13} suggest that modern
\ASP-solvers with a clause learning heuristic can benefit from
additional structural information on the instance when a relaxed form
of restricted branching is used, namely increasing the score of atoms
if a certain structural property prevails. Those properties have to be
manually identified.  Since backdoor atoms are of structural
importance for the problem it seems reasonable to initially increase
the score of the atoms if the atom is contained in the considered
backdoor.
%
% While allowing branching on additional variables speeds-up the
% solution process. 
% observed similar effects in heuristics of \ASP solvers, more detailled
% and also that in certain sets of instances restricted branching can be
% similar to other application areas beneficial.
As \strongBds{\Horn^*}s are relatively easy to compute and very easy to
approximate one could occasionally update the heuristic based on a
newly computation of a backdoor. So a solver could benefit from
backdoors in both the initial state and while learning new atoms.
%~\shortcite{GebserKaufmannOteroRomeroSchaubWanko13}.
A rigorous empirical study following these considerations is subject of
current research.%~\cite{FichteKonigSchaub14}.

\section{Summary and Future Work}\label{sec:conclusion}
We have introduced the backdoor approach to the domain of
propositional answer set programming. In a certain sense, the backdoor
approach allows us to augment known tractable classes and makes
efficient solving methods for tractable classes generally
applicable. Our approach makes recent progress in fixed-parameter
algorithmics applicable to answer set programming and establishes a
unifying approach that accommodates several parameters from the
literature. This framework gives rise to a detailed comparison of the
various parameters in terms of their generality. We introduce a
general method of lifting parameters from normal to disjunctive
programs and establish several basic properties of this method. We
further studied the preprocessing limits of \ASP rules in terms of
kernelization taking backdoor size as the parameter.

The results and concepts of this paper give rise to several research
questions. For instance, it would be interesting to consider backdoors
for target classes that contain programs with an exponential number of
answer sets, but where the set of all answer sets can be succinctly
represented. A simple example is the class of programs that consist of
(in)dependent components of bounded size.
%
% Moreover, it could be of theoretical interest whether there is an
% fpt-approximation for strong backdoor detection for the undirected
% acyclicity-based target classes
% %
% and whether there is a polynomial kernel for backdoor detection for
% various acyclicity-based target classes. 
% %
% We think the question whether there is a polynomial kernel of several
% polynomially sized kernels instead of a single kernel could be of
% practical interest.
% %
It would be interesting to enhance our backdoor approach to extended
rules in particular to weight constrains.
%
% We currently investigate whether one can exploit those ideas in
% decision heurstics of \ASP solvers.
%
Finally, it would be interesting to investigate whether backdoors can
help to improve problem encodings for \ASP-solvers.

\bibliographystyle{named}
\bibliography{../literature/johannes}

\begin{thebibliography}{}

\bibitem[\protect\citeauthoryear{Andres \bgroup \em et al.\egroup
  }{2012}]{AndresKaufmannMattheisSchaub12}
Benjamin Andres, Benjamin Kaufmann, Oliver Mattheis, and Torsten Schaub.
\newblock Unsatisfiability-based optimization in clasp.
\newblock In A.~Dovier and V.~{Santos Costa}, editors, {\em Technical
  Communications of the 28th International Conference on Logic Programming
  (ICLP'12)}, volume~17, pages 212--221. Leibniz International Proceedings in
  Informatics (LIPIcs), 2012.

\bibitem[\protect\citeauthoryear{Apt \bgroup \em et al.\egroup
  }{1988}]{AptBlairWalker88}
Krzysztof~R. Apt, Howard~A. Blair, and Adrian Walker.
\newblock Towards a theory of declarative knowledge.
\newblock {\em Foundations of deductive databases and logic programming}, pages
  89--148, 1988.

\bibitem[\protect\citeauthoryear{Ayari and Basin}{2000}]{AyariBasin00}
Abdelwaheb Ayari and David Basin.
\newblock Bounded model construction for monadic second-order logics.
\newblock In E.~Emerson and A.~Sistla, editors, {\em Computer Aided
  Verification}, volume 1855 of {\em Lecture Notes in Computer Science}, pages
  99--112. Springer Verlag, 2000.

\bibitem[\protect\citeauthoryear{Ben-Eliyahu and
  Dechter}{1994}]{Ben-EliyahuDechter94}
Rachel Ben-Eliyahu and Rina Dechter.
\newblock {Propositional semantics for disjunctive logic programs}.
\newblock {\em Ann. Math. Artif. Intell.}, 12(1):53--87, 1994.

\bibitem[\protect\citeauthoryear{Ben-Eliyahu}{1996}]{Ben-Eliyahu96}
Rachel Ben-Eliyahu.
\newblock A hierarchy of tractable subsets for computing stable models.
\newblock {\em J. Artif. Intell. Res.}, 5:27--52, 1996.

\bibitem[\protect\citeauthoryear{Bido{\'\i}t and
  Froidevaux}{1991}]{BidoitFroidevaux91}
Nicole Bido{\'\i}t and Christine Froidevaux.
\newblock Negation by default and unstratifiable logic programs.
\newblock {\em Theoretical Computer Science}, 78(1):85--112, 1991.

\bibitem[\protect\citeauthoryear{Bodlaender and
  Koster}{2008}]{BodlaenderKoster08}
Hans~L. Bodlaender and Arie M. C.~A. Koster.
\newblock Combinatorial optimization on graphs of bounded treewidth.
\newblock {\em The Computer Journal}, 51(3):255--269, 2008.

\bibitem[\protect\citeauthoryear{Bodlaender \bgroup \em et al.\egroup
  }{2009}]{BodlaenderDowneyFellowsHermelin09}
Hans~L. Bodlaender, Rodney~G. Downey, Michael~R. Fellows, and Danny Hermelin.
\newblock On problems without polynomial kernels.
\newblock {\em J. of Computer and System Sciences}, 75(8):423--434, 2009.

\bibitem[\protect\citeauthoryear{Bodlaender}{1993}]{Bodlaender93a}
Hans~L. Bodlaender.
\newblock A tourist guide through treewidth.
\newblock {\em Acta Cybernetica}, 11(1-2):1--22, 1993.

\bibitem[\protect\citeauthoryear{Bodlaender}{1997}]{Bodlaender97}
Hans~L. Bodlaender.
\newblock Treewidth: Algorithmic techniques and results.
\newblock In Igor Pr{\'\i}vara and Peter Ru{\v z}i{\v c}ka, editors, {\em
  Proceedings of the 22nd International Symposium on Mathematical Foundations
  of Computer Science (MFCS'97)}, volume 1295 of {\em Lecture Notes in Computer
  Science}, pages 19--36. Springer Verlag, 1997.

\bibitem[\protect\citeauthoryear{Bodlaender}{2005}]{Bodlaender05}
Hans~L. Bodlaender.
\newblock Discovering treewidth.
\newblock In Peter Vojt{\'a}{\v s}, M{\'a}ria Bielikov{\'a}, Bernadette
  Charron-Bost, and Ondrej S{\'y}kora, editors, {\em 31st Conference on Current
  Trends in Theory and Practice of Computer Science (SOFSEM'05)}, volume 3381
  of {\em Lecture Notes in Computer Science}, pages 1--16. Springer Verlag,
  2005.

\bibitem[\protect\citeauthoryear{Bondy and Murty}{2008}]{BondyMurty08}
John~A. Bondy and U.~S.~R. Murty.
\newblock {\em Graph theory}, volume 244 of {\em Graduate Texts in
  Mathematics}.
\newblock Springer Verlag, New York, 2008.

\bibitem[\protect\citeauthoryear{Bonsma and
  Lokshtanov}{2011}]{BonsmaLokshtanov11}
Paul Bonsma and Daniel Lokshtanov.
\newblock Feedback vertex set in mixed graphs.
\newblock In Frank Dehne, John Iacono, and J{\"o}rg-R{\"u}diger Sack, editors,
  {\em Algorithms and Data Structures}, volume 6844 of {\em Lecture Notes in
  Computer Science}, pages 122--133. Springer Verlag, 2011.

\bibitem[\protect\citeauthoryear{Brass and Dix}{1998}]{BrassDix98}
Stefan Brass and J{\"u}rgen Dix.
\newblock Characterizations of the disjunctive well-founded semantics:
  Confluent calculi and iterated {GCWA}.
\newblock {\em Journal of Automated Reasoning}, 20:143--165, 1998.

\bibitem[\protect\citeauthoryear{Buccafurri \bgroup \em et al.\egroup
  }{1997}]{BuccafurriLeoneRullo97}
Francesco Buccafurri, Nicola Leone, and Pasquale Rullo.
\newblock Strong and weak constraints in disjunctive datalog.
\newblock In J{\"u}rgen Dix, Ulrich Furbach, and Anil Nerode, editors, {\em
  Logic Programming and Nonmonotonic Reasoning}, volume 1265 of {\em Lecture
  Notes in Computer Science}, pages 2--17. Springer Verlag, 1997.

\bibitem[\protect\citeauthoryear{Cadoli and
  Lenzerini}{1994}]{CadoliLenzerini94}
Marco Cadoli and Maurizio Lenzerini.
\newblock The complexity of propositional closed world reasoning and
  circumscription.
\newblock {\em J. of Computer and System Sciences}, 48(2):255--310, 1994.

\bibitem[\protect\citeauthoryear{Calimeri \bgroup \em et al.\egroup
  }{2011}]{CalimeriEtAl11}
Francesco Calimeri, Giovambattista Ianni, Francesco Ricca, Mario Alviano,
  Annamaria Bria, Gelsomina Catalano, Susanna Cozza, Wolfgang Faber, Onofrio
  Febbraro, Nicola Leone, Marco Manna, Alessandra Martello, Claudio Panetta,
  Simona Perri, Kristian Reale, Maria Santoro, Marco Sirianni, Giorgio
  Terracina, and Pierfrancesco Veltri.
\newblock The third answer set programming competition: Preliminary report of
  the system competition track.
\newblock In James Delgrande and Wolfgang Faber, editors, {\em Logic
  Programming and Nonmonotonic Reasoning}, volume 6645 of {\em Lecture Notes in
  Computer Science}, pages 388--403. Springer Verlag, 2011.
\newblock \url{https://www.mat.unical.it/aspcomp2011/OfficialProblemSuite}.

\bibitem[\protect\citeauthoryear{Chandra and Harel}{1985}]{ChandraHarel85}
Ashok~K. Chandra and David Harel.
\newblock Horn clause queries and generalizations.
\newblock {\em The Journal of Logic Programming}, 2(1):1--15, 1985.

\bibitem[\protect\citeauthoryear{Chen and Interian}{2005}]{ChenInterian05}
Hubie Chen and Yannet Interian.
\newblock A model for generating random quantified boolean formulas.
\newblock In Leslie~Pack Kaelbling and Alessandro Saffiotti, editors, {\em
  Proceedings of the 19th International Joint Conference on Artificial
  Intelligence (IJCAI'05)}, volume~19, pages 66--71, Edinburgh, Scotland,
  August 2005. Morgan Kaufmann.

\bibitem[\protect\citeauthoryear{Chen \bgroup \em et al.\egroup
  }{2008}]{ChenLiuLuOsullivanRazgon08}
Jianer Chen, Yang Liu, Songjian Lu, Barry O'Sullivan, and Igor Razgon.
\newblock {A fixed-parameter algorithm for the directed feedback vertex set
  problem}.
\newblock {\em Journal of the ACM (JACM)}, 55(5):1--19, 2008.

\bibitem[\protect\citeauthoryear{Chen \bgroup \em et al.\egroup
  }{2010}]{ChenKanjXia10}
Jianer Chen, Iyad~A. Kanj, and Ge~Xia.
\newblock Improved upper bounds for vertex cover.
\newblock {\em Theoretical Computer Science}, 411(40--42):3736--3756, September
  2010.

\bibitem[\protect\citeauthoryear{Chitnis \bgroup \em et al.\egroup
  }{2012}]{ChitnisCyganHajiaghayiMarx12}
Rajesh Chitnis, Marek Cygan, Mohammadtaghi Hajiaghayi, and D{\'a}niel Marx.
\newblock Directed subset feedback vertex set is fixed-parameter tractable.
\newblock In Artur Czumaj, Kurt Mehlhorn, Andrew Pitts, and Roger Wattenhofer,
  editors, {\em Automata, Languages, and Programming}, volume 7391 of {\em
  Lecture Notes in Computer Science}, pages 230--241. Springer Verlag, 2012.

\bibitem[\protect\citeauthoryear{Cygan \bgroup \em et al.\egroup
  }{2011}]{CyganPilipczukPilipczukWojtaszczyk11}
Marek Cygan, Marcin Pilipczuk, Michal Pilipczuk, and Jakub~Onufry Wojtaszczyk.
\newblock Subset feedback vertex set is fixed-parameter tractable.
\newblock In {\em Proceedings of the 38th International Colloquium on Automata,
  Languages and Programming (ICALP'11)}, volume 6755 of {\em Lecture Notes in
  Computer Science}, pages 449--461. Springer Verlag, 2011.

\bibitem[\protect\citeauthoryear{Dantsin \bgroup \em et al.\egroup
  }{2001}]{DantsinEiterGottlobVoronkov01}
Evgeny Dantsin, Thomas Eiter, Georg Gottlob, and Andrei Voronkov.
\newblock Complexity and expressive power of logic programming.
\newblock {\em ACM Computing Surveys (CSUR)}, 33(3):374--425, 2001.

\bibitem[\protect\citeauthoryear{Denecker \bgroup \em et al.\egroup
  }{2009}]{DeneckerEtAl09}
Marc Denecker, Joost Vennekens, Stephen Bond, Martin Gebser, and Miros{\l}aw
  Truszczy{\'n}ski.
\newblock The second answer set programming competition.
\newblock In Esra Erdem, Fangzhen Lin, and Torsten Schaub, editors, {\em Logic
  Programming and Nonmonotonic Reasoning}, volume 5753 of {\em Lecture Notes in
  Computer Science}, pages 637--654. Springer Verlag, 2009.

\bibitem[\protect\citeauthoryear{Diestel}{2000}]{Diestel00}
Reinhard Diestel.
\newblock {\em Graph Theory}, volume 173 of {\em Graduate Texts in
  Mathematics}.
\newblock Springer Verlag, New York, 2nd edition, 2000.

\bibitem[\protect\citeauthoryear{Dowling and Gallier}{1984}]{DowlingGallier84}
William~F. Dowling and Jean~H. Gallier.
\newblock Linear-time algorithms for testing the satisfiability of
  propositional horn formulae.
\newblock {\em J. Logic Programming}, 1(3):267--284, 1984.

\bibitem[\protect\citeauthoryear{Downey and Fellows}{1999}]{DowneyFellows99}
Rodey~G. Downey and Michael~R. Fellows.
\newblock {\em Parameterized Complexity}.
\newblock Monographs in Computer Science. Springer Verlag, New York, 1999.

\bibitem[\protect\citeauthoryear{Downey \bgroup \em et al.\egroup
  }{1999}]{DowneyFellowsStege99}
Rodey~G. Downey, Michael~R. Fellows, and Ulrike Stege.
\newblock Parameterized complexity: A framework for systematically confronting
  computational intractability.
\newblock In {\em Contemporary Trends in Discrete Mathematics: From DIMACS and
  DIMATIA to the Future}, volume~49 of {\em AMS-DIMACS}, pages 49--99. American
  Mathematical Society, 1999.

\bibitem[\protect\citeauthoryear{Drescher \bgroup \em et al.\egroup
  }{2008}]{DrescherEtAl08}
Christian Drescher, Martin Gebser, Torsten Grote, Benjamin Kaufmann, Arne
  K{\"o}nig, Max Ostrowski, and Torsten Schaub.
\newblock Conflict-driven disjunctive answer set solving.
\newblock In Gerhard Brewka and J{\'e}r{\^o}me Lang, editors, {\em Proceedings
  of the 11th International Conference on Principles of Knowledge
  Representation and Reasoning (KR'08)}, pages 422--432. AAAI Press, 2008.

\bibitem[\protect\citeauthoryear{Drescher \bgroup \em et al.\egroup
  }{2010}]{DrescherGebserKaufmannSchaub10}
Christian Drescher, Martin Gebser, Benjamin Kaufmann, and Torsten Schaub.
\newblock Heuristics in conflict resolution.
\newblock {\em CoRR}, abs/1005.1716, 2010.

\bibitem[\protect\citeauthoryear{Dunne}{2007}]{Dunne07}
Paul~E. Dunne.
\newblock Computational properties of argument systems satisfying
  graph-theoretic constraints.
\newblock {\em Artificial Intelligence}, 171(10--15):701 -- 729, 2007.
\newblock <ce:title>Argumentation in Artificial Intelligence</ce:title>.

\bibitem[\protect\citeauthoryear{Dvo{{\v r}{\'a}}k \bgroup \em et al.\egroup
  }{2012}]{DvorakOrdyniakSzeider12}
Wolfgang Dvo{{\v r}{\'a}}k, Sebastian Ordyniak, and Stefan Szeider.
\newblock Augmenting tractable fragments of abstract argumentation.
\newblock {\em Artificial Intelligence}, 186(0):157--173, 2012.

\bibitem[\protect\citeauthoryear{Eiter and Gottlob}{1995}]{EiterGottlob95}
Thomas Eiter and Georg Gottlob.
\newblock On the computational cost of disjunctive logic programming:
  Propositional case.
\newblock {\em Ann. Math. Artif. Intell.}, 15(3--4):289--323, 1995.

\bibitem[\protect\citeauthoryear{Erdem and Lifschitz}{2003}]{ErdemLifschitz03}
Esra Erdem and Vladimir Lifschitz.
\newblock Tight logic programs.
\newblock {\em Theory and Practice of Logic Programming}, 3:499--518, 2003.

\bibitem[\protect\citeauthoryear{et.al.}{2012}]{SteinEtAl12}
William A.~Stein et.al.
\newblock {\em {S}age {M}athematics {S}oftware ({V}ersion 5.1.rc0)}.
\newblock The Sage Development Team, 2012.
\newblock {\tt http://www.sagemath.org}.

\bibitem[\protect\citeauthoryear{Faber \bgroup \em et al.\egroup
  }{1999}]{FaberLeoneMateisPfeifer99}
Wolfgang Faber, Nicola Leone, Cristinel Mateis, and Gerald Pfeifer.
\newblock Using database optimization techniques for nonmonotonic reasoning.
\newblock In {\em Proceedings of the 7th International Workshop on Deductive
  Databases and Logic Programming (DDLP'99)}, pages 135--139. Prolog
  Association of Japan, I. O. Committee, 1999.

\bibitem[\protect\citeauthoryear{Fages}{1994}]{Fages94}
Francois Fages.
\newblock Consistency of {C}lark's completion and existence of stable models.
\newblock {\em Journal of Methods of Logic in Computer Science}, 1(1):51--60,
  1994.

\bibitem[\protect\citeauthoryear{Fichte and Szeider}{2011}]{FichteSzeider11}
Johannes~K. Fichte and Stefan Szeider.
\newblock Backdoors to tractable answer-set programming.
\newblock In Toby Walsh, editor, {\em Proceedings of the 22nd International
  Conference on Artificial Intelligence (IJCAI'11)}, pages 863--868, Barcelona,
  Catalonia, Spain, July 2011.

\bibitem[\protect\citeauthoryear{Fichte and Szeider}{2013}]{FichteSzeider13}
Johannes~K. Fichte and Stefan Szeider.
\newblock Backdoors to normality for disjunctive logic programs.
\newblock In Marie des Jardins and Michael Littman, editors, {\em Proceedings
  of the 27th AAAI Conference on Artificial Intelligence (AAAI'13)}, pages
  320--327, Bellevue, WA, USA, July 2013. AAAI Press.

\bibitem[\protect\citeauthoryear{Fichte}{2012}]{Fichte12}
Johannes~K. Fichte.
\newblock The good, the bad, and the odd: Cycles in answer-set programs.
\newblock In Daniel Lassiter and Marija Slavkovik, editors, {\em Proceedings of
  the 23th European Summer School in Logic, Language and Information
  (ESSLLI'11) and in New Directions in Logic, Language and Computation
  (ESSLLI'10 and ESSLLI'11 Student Sessions, Selected Papers Series)}, volume
  7415 of {\em Lecture Notes in Computer Science}, pages 78--90. Springer
  Verlag, 2012.

\bibitem[\protect\citeauthoryear{Flum and Grohe}{2006}]{FlumGrohe06}
J{\"o}rg Flum and Martin Grohe.
\newblock {\em Parameterized Complexity Theory}, volume XIV of {\em Theoretical
  Computer Science}.
\newblock Springer Verlag, Berlin, 2006.

\bibitem[\protect\citeauthoryear{Fortnow and
  Santhanam}{2011}]{FortnowSanthanam11}
Lance Fortnow and Rahul Santhanam.
\newblock Infeasibility of instance compression and succinct pcps for np.
\newblock {\em J. of Computer and System Sciences}, 77(1):91--106, 2011.

\bibitem[\protect\citeauthoryear{Fortune \bgroup \em et al.\egroup
  }{1980}]{FortuneHopcroftWyllie80}
Steven Fortune, John Hopcroft, and James Wyllie.
\newblock The directed subgraph homeomorphism problem.
\newblock {\em Theoretical Computer Science}, 10(2):111--121, 1980.

\bibitem[\protect\citeauthoryear{Gaspers and
  Szeider}{2012a}]{GaspersSzeider12b}
Serge Gaspers and Stefan Szeider.
\newblock Backdoors to acyclic sat.
\newblock In Artur Czumaj, Kurt Mehlhorn, Andrew Pitts, and Roger Wattenhofer,
  editors, {\em Automata, Languages, and Programming}, volume 7391 of {\em
  Lecture Notes in Computer Science}, pages 363--374. Springer Verlag, 2012.

\bibitem[\protect\citeauthoryear{Gaspers and
  Szeider}{2012b}]{GaspersSzeider12a}
Serge Gaspers and Stefan Szeider.
\newblock Backdoors to satisfaction.
\newblock In Hans Bodlaender, Rod Downey, Fedor Fomin, and D{\'a}niel Marx,
  editors, {\em The Multivariate Algorithmic Revolution and Beyond}, volume
  7370 of {\em Lecture Notes in Computer Science}, pages 287--317. Springer
  Verlag, 2012.

\bibitem[\protect\citeauthoryear{Gaspers and
  Szeider}{2012c}]{GaspersSzeider12c}
Serge Gaspers and Stefan Szeider.
\newblock Strong backdoors to nested satisfiability.
\newblock In Alessandro Cimatti and Roberto Sebastiani, editors, {\em
  Proceedings of the 15th International Conference on Theory and Applications
  of Satisfiability Testing (SAT'12)}, volume 7317 of {\em Lecture Notes in
  Computer Science}, pages 72--85. Springer Verlag, June 2012.

\bibitem[\protect\citeauthoryear{Gaspers and Szeider}{2013}]{GaspersSzeider13}
Serge Gaspers and Stefan Szeider.
\newblock Strong backdoors to bounded treewidth sat.
\newblock In {\em Proceedings of the 54th Annual IEEE Symposium on Foundations
  of Computer Science (FOCS'13)}, Berkeley, California, USA, October 27--29
  2013.
\newblock To appear.

\bibitem[\protect\citeauthoryear{Gaspers \bgroup \em et al.\egroup
  }{2013}]{GaspersOrdyniakRamanujanSaurabhSzeider13}
Serge Gaspers, Sebastian Ordyniak, M.~S. Ramanujan, Saket Saurabh, and Stefan
  Szeider.
\newblock {Backdoors to q-Horn}.
\newblock In Natacha Portier and Thomas Wilke, editors, {\em Proceedings of the
  {E}leventh {A}nnual {ACM}-{SIAM} {S}ymposium on {D}iscrete {A}lgorithms
  ({S}an {F}rancisco, {CA}, 2000)30th International Symposium on Theoretical
  Aspects of Computer Science (STACS'13)}, volume~20 of {\em Leibniz
  International Proceedings in Informatics (LIPIcs)}, pages 67--79, Dagstuhl,
  Germany, 2013. Schloss Dagstuhl.

\bibitem[\protect\citeauthoryear{Gebser and Kaminski}{2012}]{GebserKaminski12}
Martin Gebser and Roland Kaminski.
\newblock Personal communication, 2012.

\bibitem[\protect\citeauthoryear{Gebser and Schaub}{2009}]{Asparagus09}
Martin Gebser and Torsten Schaub.
\newblock Asparagus.
\newblock url: http://asparagus.cs.uni-potsdam.de, 2009.

\bibitem[\protect\citeauthoryear{Gebser \bgroup \em et al.\egroup
  }{2007}]{GebserEtAl07}
Martin Gebser, Lengning Liu, Gayathri Namasivayam, Andr{\'e} Neumann, Torsten
  Schaub, and Miros{\l}aw Truszczy{\'n}ski.
\newblock The first answer set programming system competition.
\newblock In Chitta Baral, Gerhard Brewka, and John Schlipf, editors, {\em
  Proceedings of the 9th Conference on Logic Programming and Nonmonotonic
  Reasoning (LPNMR'07)}, volume 4483 of {\em Lecture Notes in Computer
  Science}, pages 3--17. Springer Verlag, 2007.

\bibitem[\protect\citeauthoryear{Gebser \bgroup \em et al.\egroup
  }{2008a}]{GebserKaufmannNeumannSchaub08}
Martin Gebser, Benjamin Kaufmann, Andr{\'e} Neumann, and Torsten Schaub.
\newblock Advanced preprocessing for answer set solving.
\newblock In Malik Ghallab, Constantine~D. Spyropoulos, Nikos Fakotakis, and
  Nikolaos~M. Avouris, editors, {\em Proceedings of the 18th European
  Conference on Artificial Intelligence (ECAI'08)}, volume 178 of {\em
  Frontiers in Artificial Intelligence and Applications}, pages 15--19, Patras,
  Greece, July 2008. Advanced preprocessing for answer set solving.

\bibitem[\protect\citeauthoryear{Gebser \bgroup \em et al.\egroup
  }{2008b}]{GebserSchaubThieleUsadelVeber08}
Martin Gebser, Torsten Schaub, Sven Thiele, Bj{\"o}rn Usadel, and Philippe
  Veber.
\newblock Detecting inconsistencies in large biological networks with answer
  set programming.
\newblock In Maria Garcia de~la Banda and Enrico Pontelli, editors, {\em Logic
  Programming}, volume 5366 of {\em Lecture Notes in Computer Science}, pages
  130--144. Springer Verlag, 2008.

\bibitem[\protect\citeauthoryear{Gebser \bgroup \em et al.\egroup
  }{2010}]{gringo}
Martin Gebser, Roland Kaminski, Benjamin Kaufmann, Max Ostrowski, Torsten
  Schaub, and Sven Thiele.
\newblock A user's guide to gringo, clasp, clingo, and iclingo.
\newblock Technical report, University Potsdam, 2010.

\bibitem[\protect\citeauthoryear{Gebser \bgroup \em et al.\egroup
  }{2011a}]{GebserKaminskiKaufmannSchaub11a}
Martin Gebser, Roland Kaminski, Benjamin Kaufmann, and Torsten Schaub.
\newblock Challenges in answer set solving.
\newblock In Marcello Balduccini and TranCao Son, editors, {\em Logic
  Programming, Knowledge Representation, and Nonmonotonic Reasoning}, volume
  6565 of {\em Lecture Notes in Computer Science}, pages 74--90. Springer
  Verlag, 2011.

\bibitem[\protect\citeauthoryear{Gebser \bgroup \em et al.\egroup
  }{2011b}]{GebserKaminskiKaufmannSchaub11}
Martin Gebser, Roland Kaminski, Benjamin Kaufmann, and Torsten Schaub.
\newblock {Multi-Criteria Optimization in Answer Set Programming}.
\newblock In John Gallagher and Michael Gelfond, editors, {\em Technical
  Communications of the 27th International Conference on Logic Programming
  (ICLP'11)}, volume~11 of {\em Leibniz International Proceedings in
  Informatics (LIPIcs)}, pages 1--10, Dagstuhl, Germany, 2011. Schloss
  Dagstuhl.

\bibitem[\protect\citeauthoryear{Gebser \bgroup \em et al.\egroup
  }{2011c}]{GebserKaminskiSchaub11}
Martin Gebser, Roland Kaminski, and Torsten Schaub.
\newblock Complex optimization in answer set programming.
\newblock {\em Theory Pract. Log. Program.}, 11(4-5):821--839, 2011.

\bibitem[\protect\citeauthoryear{Gebser \bgroup \em et al.\egroup
  }{2011d}]{GebserEtAl11}
Martin Gebser, Benjamin Kaufmann, Roland Kaminski, Max Ostrowski, Torsten
  Schaub, and Marius Schneider.
\newblock Potassco: The potsdam answer set solving collection.
\newblock {\em AI Communications}, 24(2):107--124, 2011.

\bibitem[\protect\citeauthoryear{Gebser \bgroup \em et al.\egroup
  }{2013a}]{GebserGlaseSabuncuSchaub13}
Martin Gebser, Thomas Glase, Orkunt Sabuncu, and Torsten Schaub.
\newblock Matchmaking with answer set programming.
\newblock In Pedro Cabalar and Tran~Cao Son, editors, {\em Proceedings of 12th
  International Conference on Logic Programming and Nonmonotonic Reasoning
  (LPNMR'13)}, volume 8148 of {\em Lecture Notes in Computer Science}, pages
  342--347, Corunna, Spain, September 15--19 2013. Springer Verlag.

\bibitem[\protect\citeauthoryear{Gebser \bgroup \em et al.\egroup
  }{2013b}]{GebserKaufmannOteroRomeroSchaubWanko13}
Martin Gebser, Benjamin Kaufmann, Ramon~P. Otero, Javier Romero, Torsten
  Schaub, and Philipp Wanko.
\newblock Domain-specific heuristics in answer set programming.
\newblock In {\em Proceedings of 27th Conference on Artificial Intelligence
  (AAAI'13)}, 2013.

\bibitem[\protect\citeauthoryear{Gelder \bgroup \em et al.\egroup
  }{1991}]{GelderRossSchlipf91}
Allen~Van Gelder, Kenneth~A. Ross, and John~S. Schlipf.
\newblock The well-founded semantics for general logic programs.
\newblock {\em J. of the ACM}, 38(3):620--650, 1991.

\bibitem[\protect\citeauthoryear{Gelder}{1989}]{Gelder89a}
Allen~Van Gelder.
\newblock Negation as failure using tight derivations for general logic
  programs.
\newblock {\em The Journal of Logic Programming}, 6(1--2):109--133, 1989.

\bibitem[\protect\citeauthoryear{Gelfond and
  Lifschitz}{1988}]{GelfondLifschitz88}
Michael Gelfond and Vladimir Lifschitz.
\newblock The stable model semantics for logic programming.
\newblock In Robert~A. Kowalski and Kenneth~A. Bowen, editors, {\em Proceedings
  of the 5th International Conference and Symposium (ICLP/SLP'88)}, volume~2,
  pages 1070--1080. MIT Press, 1988.

\bibitem[\protect\citeauthoryear{Gelfond and
  Lifschitz}{1991}]{GelfondLifschitz91}
Michael Gelfond and Vladimir Lifschitz.
\newblock Classical negation in logic programs and disjunctive databases.
\newblock {\em New Generation Comput.}, 9(3/4):365--386, 1991.

\bibitem[\protect\citeauthoryear{Giunchiglia \bgroup \em et al.\egroup
  }{1998}]{GiunchigliaMassarottoSebastiani98}
E.~Giunchiglia, A.~Massarotto, and R.~Sebastiani.
\newblock Act, and the rest will follow: Exploiting determinism in planning as
  satisfiability.
\newblock In Jack Mostow and Charles Rich, editors, {\em Proceedings of the
  15th National Conference on Artificial Intelligence (AAAI'98)}, pages
  948--953, Madison, WI, USA, 1998. AAAI Press.

\bibitem[\protect\citeauthoryear{Giunchiglia \bgroup \em et al.\egroup
  }{2002}]{GiunchigliaMarateaTacchella02}
E.~Giunchiglia, M.~Maratea, and A.~Tacchella.
\newblock {Dependent and independent variables in propositional
  satisfiability}.
\newblock {\em Logics in Artificial Intelligence}, pages 296--307, 2002.

\bibitem[\protect\citeauthoryear{Gottlob and Szeider}{2008}]{GottlobSzeider08}
G.~Gottlob and S.~Szeider.
\newblock Fixed-parameter algorithms for artificial intelligence, constraint
  satisfaction and database problems.
\newblock {\em The Computer Journal}, 51(3):303--325, 2008.

\bibitem[\protect\citeauthoryear{Gottlob \bgroup \em et al.\egroup
  }{2002}]{GottlobScarcelloSideri02}
Georg Gottlob, Francesco Scarcello, and Martha Sideri.
\newblock Fixed-parameter complexity in {AI} and nonmonotonic reasoning.
\newblock {\em Artificial Intelligence}, 138(1-2):55--86, 2002.

\bibitem[\protect\citeauthoryear{Gottlob \bgroup \em et al.\egroup
  }{2010}]{GottlobPichlerWei10}
Georg Gottlob, Reinhard Pichler, and Fang Wei.
\newblock Bounded treewidth as a key to tractability of knowledge
  representation and reasoning.
\newblock {\em Artificial Intelligence}, 174(1):105--132, 2010.

\bibitem[\protect\citeauthoryear{Gurobi~Optimization}{2014}]{Gurobi-Optimization14}
Inc. Gurobi~Optimization.
\newblock Gurobi optimizer reference manual, 2014.
\newblock Version 5.0.2.

\bibitem[\protect\citeauthoryear{IBM}{2011}]{cplex11}
IBM.
\newblock {\em IBM ILOG CPLEX Optimization Studio CPLEX User's Manual}, version
  12 release 4 edition, 2011.

\bibitem[\protect\citeauthoryear{Jakl \bgroup \em et al.\egroup
  }{2009}]{JaklPichlerWoltran09}
Michael Jakl, Reinhard Pichler, and Stefan Woltran.
\newblock Answer-set programming with bounded treewidth.
\newblock In Craig Boutilier, editor, {\em Proceedings of the 21st
  International Joint Conference on Artificial Intelligence (IJCAI'09)},
  volume~2, pages 816--822, Pasadena, CA, USA, July 2009. Elsevier Science
  Publishers, North-Holland.

\bibitem[\protect\citeauthoryear{Janhunen and
  Niemela}{2011}]{JanhunenNiemela11}
Tomi Janhunen and Ilkka Niemela.
\newblock Compact translations of non-disjunctive answer set programs to
  propositional clauses.
\newblock In Marcello Balduccini and Tran Son, editors, {\em Logic Programming,
  Knowledge Representation, and Nonmonotonic Reasoning}, volume 6565 of {\em
  Lecture Notes in Computer Science}, pages 111--130. Springer Verlag, 2011.

\bibitem[\protect\citeauthoryear{Janhunen \bgroup \em et al.\egroup
  }{2006}]{JanhunenEtAl06}
T.~Janhunen, I.~Niemel{\"a}, D.~Seipel, P.~Simons, and J.H. You.
\newblock Unfolding partiality and disjunctions in stable model semantics.
\newblock {\em ACM Trans. Comput. Log.}, 7(1):1--37, 2006.

\bibitem[\protect\citeauthoryear{Janhunen \bgroup \em et al.\egroup
  }{2009}]{JanhunenNiemelaSevalnev09}
Tomi Janhunen, Ilkka Niemela, and Mark Sevalnev.
\newblock Computing stable models via reductions to difference logic.
\newblock In Esra Erdem, Fangzhen Lin, and Torsten Schaub, editors, {\em
  Proceedings of the 10th International Conference on Logic Programming and
  Nonmonotonic Reasoning (LPNMR '09)}, volume 5753 of {\em Lecture Notes in
  Computer Science}, pages 142--154. Springer Verlag, 2009.

\bibitem[\protect\citeauthoryear{Janota and
  Marques-Silva}{2011}]{JanotaMarques-Silva11}
Mikol{\'a}{\v s} Janota and Joao Marques-Silva.
\newblock A tool for circumscription-based mus membership testing.
\newblock In James Delgrande and Wolfgang Faber, editors, {\em Logic
  Programming and Nonmonotonic Reasoning}, volume 6645 of {\em Lecture Notes in
  Computer Science}, pages 266--271. Springer Verlag, 2011.

\bibitem[\protect\citeauthoryear{J{\"a}rvisalo and
  Junttila}{2009}]{JarvisaloJunttila09}
Matti J{\"a}rvisalo and Tommi Junttila.
\newblock {Limitations of restricted branching in clause learning}.
\newblock {\em Constraints}, 14(3):325--356, 2009.

\bibitem[\protect\citeauthoryear{J{\"a}rvisalo and
  Niemel{\"a}}{2008}]{JarvisaloNiemela08}
M.~J{\"a}rvisalo and Ilkka Niemel{\"a}.
\newblock {The effect of structural branching on the efficiency of clause
  learning SAT solving: An experimental study}.
\newblock {\em Journal of Algorithms}, 63(1-3):90--113, 2008.

\bibitem[\protect\citeauthoryear{Jost \bgroup \em et al.\egroup
  }{2012}]{JostSabuncuSchaub12a}
Holger Jost, Orkunt Sabuncu, and Torsten Schaub.
\newblock Suggesting new interactions related to events in a social network for
  elderly.
\newblock In {\em Proceedings of the 26th BCS Conference on Human Computer
  Interaction (HCI'12)}, Birmingham, UK, 12 - 14 September 2012. British
  Computer Society, Swinton.

\bibitem[\protect\citeauthoryear{Kakimura \bgroup \em et al.\egroup
  }{2012}]{KakimuraKawarabayashiKobayashi12}
Naonori Kakimura, Ken-ichi Kawarabayashi, and Yusuke Kobayashi.
\newblock Erd{\"o}s-p{\'o}sa property and its algorithmic applications: parity
  constraints, subset feedback set, and subset packing.
\newblock In Dana Randall, editor, {\em Proceedings of the 23rd Annual ACM-SIAM
  Symposium on Discrete Algorithms (SODA'12)}, pages 1726--1736, San Francisco,
  CA, USA, 2012. Society for Industrial and Applied Mathematics (SIAM).

\bibitem[\protect\citeauthoryear{Kanchanasut and
  Stuckey}{1992}]{KanchanasutStuckey92}
Kanchana Kanchanasut and Peter~J. Stuckey.
\newblock Transforming normal logic programs to constraint logic programs.
\newblock {\em Theoretical Computer Science}, 105(1):27 -- 56, 1992.

\bibitem[\protect\citeauthoryear{Kawarabayashi and
  Kobayashi}{2010}]{KawarabayashiKobayashi10}
Kenichi Kawarabayashi and Yusuke Kobayashi.
\newblock Fixed-parameter tractability for the subset feedback set problem and
  the s-cycle packing problem.
\newblock Technical report, University of Tokyo, Japan, 2010.

\bibitem[\protect\citeauthoryear{Kottler \bgroup \em et al.\egroup
  }{2008}]{KottlerKaufmannSinz08}
Stephan Kottler, Michael Kaufmann, and Carsten Sinz.
\newblock A new bound for an {NP}-hard subclass of 3-{SAT} using backdoors.
\newblock In Hans~Kleine B{\"u}ning and Xishun Zhao, editors, {\em Proceedings
  of the 11th International Conference on Theory and Applications of
  Satisfiability Testing (SAT'08)}, volume 4996 of {\em Lecture Notes in
  Computer Science}, pages 161--167, Guangzhou, China, May 2008. Springer
  Verlag.

\bibitem[\protect\citeauthoryear{Lapaugh and
  Papadimitriou}{1984}]{LapaughPapadimitriou84}
Andrea~S. Lapaugh and Christos~H. Papadimitriou.
\newblock The even-path problem for graphs and digraphs.
\newblock {\em Networks}, 14(4):507--513, 1984.

\bibitem[\protect\citeauthoryear{Lee and Lifschitz}{2003}]{LeeLifschitz03}
Joohyung Lee and Vladimir Lifschitz.
\newblock Loop formulas for disjunctive logic programs.
\newblock In Catuscia Palamidessi, editor, {\em Logic Programming}, volume 2916
  of {\em Lecture Notes in Computer Science}, pages 451--465. Springer Verlag,
  2003.

\bibitem[\protect\citeauthoryear{Leone \bgroup \em et al.\egroup
  }{1997}]{LeoneRulloScarcello97}
Nicola Leone, Pasquale Rullo, and Francesco Scarcello.
\newblock Disjunctive stable models: Unfounded sets, fixpoint semantics, and
  computation.
\newblock {\em Information and Computation}, 135:69--112, 1997.

\bibitem[\protect\citeauthoryear{Leone \bgroup \em et al.\egroup
  }{2006}]{LeoneEtAl06}
Nicola Leone, Gerald Pfeifer, Wolfgang Faber, Thomas Eiter, Georg Gottlob,
  Simona Perri, and Francesco Scarcello.
\newblock {The {D}{L}{V} system for knowledge representation and reasoning}.
\newblock {\em ACM Transactions on Computational Logic (TOCL)}, 7(3):499--562,
  2006.

\bibitem[\protect\citeauthoryear{Lierler}{2005}]{Lierler05a}
Yuliya Lierler.
\newblock cmodels -- sat-based disjunctive answer set solver.
\newblock In Chitta Baral, Gianluigi Greco, Nicola Leone, and Giorgio
  Terracina, editors, {\em Logic Programming and Nonmonotonic Reasoning},
  volume 3662 of {\em Lecture Notes in Computer Science}, pages 447--451.
  Springer Verlag, 2005.

\bibitem[\protect\citeauthoryear{Lin and Zhao}{2004a}]{LinZhao04}
Fangzhen Lin and Xishun Zhao.
\newblock On odd and even cycles in normal logic programs.
\newblock In Anthony~G. Cohn, editor, {\em Proceedings of the 19th national
  conference on Artifical intelligence (AAAI'04)}, pages 80--85, San Jose, CA,
  USA, July 2004. AAAI Press.

\bibitem[\protect\citeauthoryear{Lin and Zhao}{2004b}]{LinZhao04a}
Fangzhen Lin and Yuting Zhao.
\newblock {ASSAT: Computing answer sets of a logic program by SAT solvers}.
\newblock {\em Artificial Intelligence}, 157(1-2):115--137, 2004.

\bibitem[\protect\citeauthoryear{Liu \bgroup \em et al.\egroup
  }{2012}]{LiuJanhunenNiemela12}
Guohua Liu, Tomi Janhunen, and Ilkka Niemel{\"a}.
\newblock Answer set programming via mixed integer programming.
\newblock In Sheila McIlraith and Thomas Eiter, editors, {\em Proceedings of
  the 13th International Conference on the Principles of Knowledge
  Representation and Reasoning (KR'12)}, pages 32--42, Rome, Italy, 2012. AAAI
  Press.

\bibitem[\protect\citeauthoryear{Maratea \bgroup \em et al.\egroup
  }{2008}]{MarateaRiccaFaberLeone08}
Marco Maratea, Francesco Ricca, Wolfgang Faber, and Nicola Leone.
\newblock Look-back techniques and heuristics in dlv: Implementation,
  evaluation, and comparison to qbf solvers.
\newblock {\em Journal of Algorithms}, 63(1-3):70 -- 89, 2008.

\bibitem[\protect\citeauthoryear{Marek and
  Truszczynski}{1991a}]{MarekTruszczynski91a}
Wiktor Marek and M.~Truszczynski.
\newblock Computing intersection of autoepistemic expansions.
\newblock In {\em Proceedings of the 1st International Conference on Logic
  Programming and Nonmonotonic Reassoning (LPNMR'91)}, pages 37--50. MIT Press,
  1991.

\bibitem[\protect\citeauthoryear{Marek and
  Truszczy{\'n}ski}{1991b}]{MarekTruszczynski91}
Wiktor Marek and Miros{\l}aw Truszczy{\'n}ski.
\newblock Autoepistemic logic.
\newblock {\em J. of the ACM}, 38(3):588--619, 1991.

\bibitem[\protect\citeauthoryear{Marek and
  Truszczynski}{1999}]{MarekTruszczynski99}
Victor~W. Marek and Miroslaw Truszczynski.
\newblock Stable models and an alternative logic programming paradigm.
\newblock In Krzysztof~R. Apt, Victor~W. Marek, Miroslaw Truszczynski, and
  David~S. Warren, editors, {\em The Logic Programming Paradigm: a 25-Year
  Perspective}, pages 375--398. Springer Verlag, September 1999.

\bibitem[\protect\citeauthoryear{Misra \bgroup \em et al.\egroup
  }{2012}]{MisraRamanRamanujanSaurabh12}
Pranabendu Misra, Venkatesh Raman, M.S. Ramanujan, and Saket Saurabh.
\newblock Parameterized algorithms for even cycle transversal.
\newblock In MartinCharles Golumbic, Michal Stern, Avivit Levy, and Gila
  Morgenstern, editors, {\em Graph-Theoretic Concepts in Computer Science},
  volume 7551 of {\em Lecture Notes in Computer Science}, pages 172--183.
  Springer Verlag, 2012.

\bibitem[\protect\citeauthoryear{Montalva \bgroup \em et al.\egroup
  }{2008}]{MontalvaAracenaGajardo08}
Marco Montalva, Julio Aracena, and Anah{\'\i} Gajardo.
\newblock On the complexity of feedback set problems in signed digraphs.
\newblock {\em Electronic Notes in Discrete Mathematics}, 30(0):249--254, 2008.

\bibitem[\protect\citeauthoryear{Morak and Woltran}{2012}]{MorakWoltran12}
Michael Morak and Stefan Woltran.
\newblock {Preprocessing of Complex Non-Ground Rules in Answer Set
  Programming}.
\newblock In Agostino Dovier and V{\'\i}tor~Santos Costa, editors, {\em
  Technical Communications of the 28th International Conference on Logic
  Programming (ICLP'12)}, volume~17 of {\em Leibniz International Proceedings
  in Informatics (LIPIcs)}, pages 247--258, Dagstuhl, Germany, 2012. Schloss
  Dagstuhl.

\bibitem[\protect\citeauthoryear{Morak \bgroup \em et al.\egroup
  }{2010}]{MorakPichlerRummeleWoltran10}
Michael Morak, Reinhard Pichler, Stefan R{\"u}mmele, and Stefan Woltran.
\newblock A dynamic-programming based asp-solver.
\newblock In Tomi Janhunen and Ilkka Niemel{{\"a}}, editors, {\em Proceedings
  of 12th European Conference on Logics in Artificial Intelligence (JELIA'10)},
  volume 6341 of {\em Lecture Notes in Computer Science}, pages 369--372,
  Helsinki, Finland, September 2010. Springer Verlag.

\bibitem[\protect\citeauthoryear{Niedermeier}{2006}]{Niedermeier06}
Rolf Niedermeier.
\newblock {\em Invitation to Fixed-Parameter Algorithms}.
\newblock Oxford Lecture Series in Mathematics and its Applications. Oxford
  University Press, 2006.

\bibitem[\protect\citeauthoryear{Niemel{\"a} and
  Rintanen}{1994}]{NiemelaRintanen94}
Ilkka Niemel{\"a} and Jussi Rintanen.
\newblock On the impact of stratification on the complexity of nonmonotonic
  reasoning.
\newblock {\em Journal of Applied Non-Classical Logics}, 4(2):141--179, 1994.

\bibitem[\protect\citeauthoryear{Niemel{\"a} \bgroup \em et al.\egroup
  }{2000}]{NiemelaSimonsSyrjanen00}
Ilkka Niemel{\"a}, Patrik Simons, and Tommi Syrj{\"a}nen.
\newblock Smodels: A system for answer set programming.
\newblock {\em CoRR}, cs.AI/0003033, 2000.

\bibitem[\protect\citeauthoryear{Niemel{\"a}}{1999}]{Niemela99}
Ilkka Niemel{\"a}.
\newblock {Logic programs with stable model semantics as a constraint
  programming paradigm}.
\newblock {\em Ann. Math. Artif. Intell.}, 25(3):241--273, 1999.

\bibitem[\protect\citeauthoryear{Nishimura \bgroup \em et al.\egroup
  }{2004}]{NishimuraRagdeSzeider04-informal}
Naomi Nishimura, Prabhakar Ragde, and Stefan Szeider.
\newblock Detecting backdoor sets with respect to {Horn} and binary clauses.
\newblock In Holger~H. Hoos and David~G. Mitchell, editors, {\em Proceedings of
  the 7th International Conference on Theory and Applications of Satisfiability
  Testing (SAT'04)}, volume 3542 of {\em Lecture Notes in Computer Science},
  pages 96--103, Vancouver, BC, Canada, May 2004. Springer Verlag.

\bibitem[\protect\citeauthoryear{Nishimura \bgroup \em et al.\egroup
  }{2007}]{NishimuraRagdeSzeider07}
Naomi Nishimura, Prabhakar Ragde, and Stefan Szeider.
\newblock Solving \#{S}{A}{T} using vertex covers.
\newblock {\em Acta Informatica}, 44(7-8):509--523, 2007.

\bibitem[\protect\citeauthoryear{Ordyniak and
  Szeider}{2011}]{OrdyniakSzeider11}
Sebastian Ordyniak and Stefan Szeider.
\newblock Augmenting tractable fragments of abstract argumentation.
\newblock In Toby Walsh, editor, {\em Proceedings of the 22nd International
  Joint Conference on Artificial Intelligence (IJCAI'11)}, pages 1033--1038.
  AAAI Press/IJCAI, 2011.

\bibitem[\protect\citeauthoryear{Pfandler \bgroup \em et al.\egroup
  }{2013}]{PfandlerRummeleSzeider13}
Andreas Pfandler, Stefan R\"ummele, and Stefan Szeider.
\newblock Backdoors to abduction.
\newblock In Francesca Rossi, editor, {\em Proceedings of the 23nd
  International Joint Conference on Artificial Intelligence (IJCAI'13)}, pages
  1046--1052, Beijing, China, August 2013. AAAI Press/IJCAI.

\bibitem[\protect\citeauthoryear{Pichler \bgroup \em et al.\egroup
  }{2009}]{PichlerRummeleWoltran09}
Reinhard Pichler, Stefan R{\"u}mmele, and Stefan Woltran.
\newblock Belief revision with bounded treewidth.
\newblock In Esra Erdem, Fangzhen Lin, and Torsten Schaub, editors, {\em Logic
  Programming and Nonmonotonic Reasoning}, volume 5753 of {\em Lecture Notes in
  Computer Science}, pages 250--263. Springer Verlag, 2009.

\bibitem[\protect\citeauthoryear{Razgon and
  O'Sullivan}{2009}]{RazgonOSullivan09}
Igor Razgon and Barry O'Sullivan.
\newblock Almost 2-{SAT} is fixed parameter tractable.
\newblock {\em J. of Computer and System Sciences}, 75(8):435--450, 2009.

\bibitem[\protect\citeauthoryear{Ricca \bgroup \em et al.\egroup
  }{2012}]{RiccaGrassoAlvianoMannaLioIiritanoLeone12}
Francesco Ricca, G.~Grasso, Mario Alviano, Marco Manna, V.~Lio, S.~Iiritano,
  and Nicola Leone.
\newblock Team-building with answer set programming in the gioia-tauro seaport.
\newblock {\em Theory and Practice of Logic Programming}, 12:361--381, 4 2012.

\bibitem[\protect\citeauthoryear{Robertson and
  Seymour}{1984}]{RobertsonSeymour84}
Neil Robertson and P.D. Seymour.
\newblock {Graph minors. III. Planar tree-width}.
\newblock {\em Journal of Combinatorial Theory, Series B}, 36(1):49--64, 1984.

\bibitem[\protect\citeauthoryear{Robertson and
  Seymour}{1985}]{RobertsonSeymour85}
Neil Robertson and P.D. Seymour.
\newblock {Graph minors-a survey}.
\newblock In {\em Surveys in combinatorics, 1985: invited papers for the Tenth
  British Combinatorial Conference}, page 153. Cambridge Univ Pr, 1985.

\bibitem[\protect\citeauthoryear{Robertson and
  Seymour}{1986}]{RobertsonSeymour86}
Neil Robertson and P.D. Seymour.
\newblock Graph minors. ii. algorithmic aspects of tree-width.
\newblock {\em Journal of Algorithms}, 7(3):309--322, 1986.

\bibitem[\protect\citeauthoryear{Robertson \bgroup \em et al.\egroup
  }{1999}]{RobertsonSeymourThomas99}
Neil Robertson, P.D. Seymour, and Robin Thomas.
\newblock {Permanents, Pfaffian orientations, and even directed circuits}.
\newblock {\em Annals of Mathematics}, 150(3):929--975, 1999.

\bibitem[\protect\citeauthoryear{Rosamond}{2010}]{Rosamond10}
Frances Rosamond.
\newblock Table of races.
\newblock In {\em Parameterized Complexity Newsletter}, pages 4--5. 2010.
\newblock \url{http://fpt.wikidot.com/}.

\bibitem[\protect\citeauthoryear{Ruan \bgroup \em et al.\egroup
  }{2004}]{RuanKautzHorvitz04}
Yongshao Ruan, Henry~A. Kautz, and Eric Horvitz.
\newblock The backdoor key: A path to understanding problem hardness.
\newblock In Deborah~L. McGuinness and George Ferguson, editors, {\em
  Proceedings of the 19th National Conference on Artificial Intelligence, 16th
  Conference on Innovative Applications of Artificial Intelligence}, pages
  124--130. AAAI Press / The MIT Press, 2004.

\bibitem[\protect\citeauthoryear{Samer and Szeider}{2008}]{SamerSzeider08b}
Marko Samer and Stefan Szeider.
\newblock Backdoor trees.
\newblock In Robert~C. Holte and Adele~E. Howe, editors, {\em Proceedings of
  23rd Conference on Artificial Intelligence (AAAI'08)}, pages 363--368,
  Vancouver, BC, Canada, July 2008.

\bibitem[\protect\citeauthoryear{Samer and Szeider}{2009a}]{SamerSzeider09a}
Marko Samer and Stefan Szeider.
\newblock Backdoor sets of quantified {B}oolean formulas.
\newblock {\em Journal of Automated Reasoning}, 42(1):77--97, 2009.

\bibitem[\protect\citeauthoryear{Samer and Szeider}{2009b}]{SamerSzeider08c}
Marko Samer and Stefan Szeider.
\newblock Fixed-parameter tractability.
\newblock In Armin Biere, Marijn Heule, Hans van Maaren, and Toby Walsh,
  editors, {\em Handbook of Satisfiability}, chapter~13, pages 425--454. IOS
  Press, 2009.

\bibitem[\protect\citeauthoryear{Schaub \bgroup \em et al.\egroup
  }{2009}]{potassco13}
Torsten Schaub, Martin Gebser, Benjamin Kaufmann, Roland Kaminski, and Sven
  Thiele.
\newblock Potassco, the potsdam answer set solving collection, bundles tools
  for answer set programming, 2009.

\bibitem[\protect\citeauthoryear{Schnorr}{1981}]{Schnorr81}
C.P. Schnorr.
\newblock On self-transformable combinatorial problems.
\newblock In H.~K{\"o}nig, B.~Korte, and K.~Ritter, editors, {\em Mathematical
  Programming at Oberwolfach}, volume~14 of {\em Mathematical Programming
  Studies}, pages 225--243. Springer Verlag, 1981.

\bibitem[\protect\citeauthoryear{Sinz}{2005}]{Sinz05}
Carsten Sinz.
\newblock {Towards an optimal CNF encoding of boolean cardinality constraints}.
\newblock {\em Principles and Practice of Constraint Programming-CP 2005},
  pages 827--831, 2005.

\bibitem[\protect\citeauthoryear{Strichman}{2000}]{Strichman00}
Ofer Strichman.
\newblock {Tuning SAT checkers for bounded model checking}.
\newblock In {\em Computer Aided Verification}, pages 480--494. Springer
  Verlag, 2000.

\bibitem[\protect\citeauthoryear{Szeider}{2011}]{Szeider11}
Stefan Szeider.
\newblock Limits of preprocessing.
\newblock In Wolfram Burgard and Dan Roth, editors, {\em Proceedings of the
  25th Conference on Artificial Intelligence (AAAI'11)}, pages 93--98, San
  Francisco, CA, USA, August 2011.

\bibitem[\protect\citeauthoryear{Thielscher}{2009}]{Thielscher09}
Michael Thielscher.
\newblock Answer set programming for single-player games in general game
  playing.
\newblock In Patricia~M. Hill and David~S. Warren, editors, {\em Proceedings of
  the 25th International Conference on Logic Programming (ICLP'09)}, volume
  5649 of {\em Lecture Notes in Computer Science}, pages 327--341. Springer
  Verlag, Pasadena, CA, USA, July 14-17 2009.

\bibitem[\protect\citeauthoryear{Thomass{\'e}}{2009}]{Thomasse09}
St{\'e}phan Thomass{\'e}.
\newblock A quadratic kernel for feedback vertex set.
\newblock In Claire Mathieu, editor, {\em Proceedings of the twentieth Annual
  ACM-SIAM Symposium on Discrete Algorithms (SODA'09)}, pages 115--119, New
  York, NY, USA, January 2009. Society for Industrial and Applied Mathematics,
  Society for Industrial and Applied Mathematics (SIAM).

\bibitem[\protect\citeauthoryear{To \bgroup \em et al.\egroup
  }{2009}]{ToPontelliSon09}
Son~Thanh To, Enrico Pontelli, and Tran~Cao Son.
\newblock A conformant planner with explicit disjunctive representation of
  belief states.
\newblock In Alfonso Gerevini, Adele~E. Howe, Amedeo Cesta, and Ioannis
  Refanidis, editors, {\em Proceedings of the 19th International Conference on
  Automated Planning and Scheduling (ICAPS'09)}, pages 305--312, Thessaloniki,
  Greece, September 2009. AAAI Press.

\bibitem[\protect\citeauthoryear{Van~Emden and
  Kowalski}{1976}]{Van-EmdenKowalski76}
M.~H. Van~Emden and Robert.~A. Kowalski.
\newblock The semantics of predicate logic as a programming language.
\newblock {\em J. ACM}, 23:733--742, October 1976.

\bibitem[\protect\citeauthoryear{van Rossum}{1995}]{Rossum95}
Guido van Rossum.
\newblock Python tutorial.
\newblock Technical Report CS-R9526, Centrum voor Wiskunde en Informatica
  (CWI), Amsterdam, May 1995.

\bibitem[\protect\citeauthoryear{Vazirani and
  Yannakakis}{1988}]{VaziraniYannakakis88}
Vijay Vazirani and Mihalis Yannakakis.
\newblock Pfaffian orientations, 0/1 permanents, and even cycles in directed
  graphs.
\newblock In Timo Lepist{\"o} and Arto Salomaa, editors, {\em Automata,
  Languages and Programming}, volume 317 of {\em Lecture Notes in Computer
  Science}, pages 667--681. Springer Verlag, 1988.

\bibitem[\protect\citeauthoryear{Williams \bgroup \em et al.\egroup
  }{2003a}]{WilliamsGomesSelman03}
Ryan Williams, Carla Gomes, and Bart Selman.
\newblock Backdoors to typical case complexity.
\newblock In Georg Gottlob and Toby Walsh, editors, {\em Proceedings of the
  18th International Joint Conference on Artificial Intelligence (IJCAI'03)},
  pages 1173--1178, Acapulco, Mexico, August 2003. Morgan Kaufmann.

\bibitem[\protect\citeauthoryear{Williams \bgroup \em et al.\egroup
  }{2003b}]{WilliamsGomesSelman03a}
Ryan Williams, Carla Gomes, and Bart Selman.
\newblock On the connections between backdoors, restarts, and heavy-tailedness
  in combinatorial search.
\newblock In {\em Informal Proceedings of the 6th International Conference on
  Theory and Applications of Satisfiability Testing (SAT'03)}, pages 222--230,
  Portofino, Italy, May 2003.

\bibitem[\protect\citeauthoryear{Zhao and Lin}{2003}]{ZhaoLin03}
Yuting Zhao and Fangzhen Lin.
\newblock {Answer set programming phase transition: A study on randomly
  generated programs}.
\newblock In Catuscia Palamidessi, editor, {\em Proceedings of the 19th
  International Conference on Logic Programming (ICLP'03)}, volume 2916 of {\em
  Lecture Notes in Computer Science}, pages 239--253, Mumbai, India, December
  9-13 2003. Springer Verlag.

\bibitem[\protect\citeauthoryear{Zhao}{2002}]{Zhao02}
Jicheng Zhao.
\newblock A study of answer set programming.
\newblock Mphil thesis, The Hong Kong University of Science and Technology,
  Dept. of Computer Science, 2002.

\end{thebibliography}

\end{document}